\documentclass[10pt]{article}

\usepackage[normalem]{ulem}
\usepackage{amsmath,amsthm,amssymb,enumerate,hyperref}
\usepackage{graphicx}
\usepackage{subfig}
\usepackage{latexsym}
\usepackage{authblk}
\usepackage{color}
\usepackage{fullpage}

\numberwithin{equation}{section}  

\theoremstyle{plain}

\newtheorem{thm}{Theorem}[section]
\newtheorem{cor}[thm]{Corollary}
\newtheorem{lemma}[thm]{Lemma}
\newtheorem{prop}[thm]{Proposition}

\theoremstyle{definition}
\newtheorem{defn}[thm]{Definition}
\newtheorem{assumption}[thm]{Assumption}

\theoremstyle{remark}
\newtheorem{rem}[thm]{Remark}

\newcommand{\ra}{\rightarrow}

\newcommand{\abs}[1]{\left\lvert #1 \right\rvert}
\newcommand{\ket}[1]{\lvert #1 \rangle}

\newcommand{\slim}{{\rm s}\mbox{-}\lim}

\newcommand{\RR}{\mathbb R}

\newcommand{\CC}{\mathbb C}
\newcommand{\ZZ}{\mathbb Z}

\newcommand{\NN}{\mathbb N}

\newcommand{\calA}{\mathcal A}
\newcommand{\calB}{\mathcal B}

\newcommand{\calH}{\mathcal H}

\newcommand{\Aut}{\operatorname{Aut}}

\begin{document}

\title{On the stability of charges in infinite quantum spin systems}
\renewcommand\Affilfont{\itshape\small}
  \author[1]{Matthew Cha}
  \author[2,3]{Pieter Naaijkens\footnote{Present address: Departamento de An{\'a}lisis Matem{\'a}tico y Matem{\'a}tica Applicada, Universidad Complutense de Madrid, Spain}}
\author[2]{Bruno Nachtergaele}
\affil[1]{Department of Mathematics, Michigan State University, East Lansing, USA}
\affil[2]{Department of Mathematics, University of California, Davis, USA}
\affil[3]{JARA Institute for Quantum Computing, RWTH Aachen University, Germany}
  \date{\today}

\maketitle

\begin{abstract} 
	We consider a theory of superselection sectors for infinite quantum spin
	systems, describing charges that can be approximately localized in
	cone-like regions. The primary examples we have in mind are the anyons (or
	charges) in topologically ordered models such as Kitaev's quantum double models 
	and perturbations of such models. In order to cover the case of perturbed quantum
	double models, the Doplicher-Haag-Roberts approach, in which strict localization is
	assumed, has to be amended. To this end we consider endomorphisms of the 
	observable algebra that are almost localized in cones. Under natural conditions on
	the reference ground state (which plays a role analogous to the vacuum state in relativistic theories), 
	we obtain a braided tensor $C^*$-category describing the sectors.  We also introduce 
	a superselection criterion selecting excitations with energy below a
	threshold. When the threshold energy falls in a gap of the spectrum of the
	ground state, we prove stability of the entire superselection structure
	under perturbations that do not close the gap. We apply our results to
	prove that  all essential properties of the anyons in Kitaev's abelian quantum double
	models are stable against perturbations.
\end{abstract}

\section{Introduction}
For a quantum many-body system described by a Hamiltonian depending on one or more 
parameters, quantum phase transitions mark the boundaries between regions in parameter 
space where the ground states are characterized by qualitatively distinct physical properties. 
There is much interest currently in the classification of gapped ground state phases, for example by 
constructing a complete set of invariants of the phase. We consider this problem for a particular
class of models that belong to the class of \emph{topological} phases, accompanied by low-lying excitations 
described by anyons. Such phases are not associated with spontaneous breaking of a local symmetry
and a local order parameter, as is the case in the Landau's theory of phase transitions. 

We are particularly interested in topologically ordered quantum spin systems.
In a finite volume with boundaries or with non-trivial topology, the topological nature of the phase is reflected in a non-trivial ground state 
degeneracy~\cite{Wen} and a spectral gap above the ground state uniform in the system size. 
Such states are also characterized by long range entanglement~\cite{ChenGW}, which 
manifests itself as a correction term to the area law in the thermodynamic limit~\cite{KitaevP,LevinW2},
and, in two spatial dimensions, by the emergence of elementary excitations with braid statistics called anyons.
Perhaps the best known physical manifestation of anyons are the quasi-particle excitations in the fractional quantum 
Hall effect~\cite{ArovasSW,Halperin}. The prototypical models of the type we study in this paper are Kitaev's
quantum double models~\cite{KitaevQD} and the string-net models of Levin and Wen~\cite{LevinW}.

In infinite volume, a structure of superselection sectors appears, which are labeled by (topological) 
charges, or anyons.
Hence a fundamental question in the classification of such topological phases of matter is how to extract the physical properties of the anyons from first principles.
And moreover, based on the topological nature of the anyons, one expects that this structure remains unchanged if the underlying system is perturbed a little.
Or in other words, it is expected that the superselection sectors form an invariant of the topological phase.

In this paper we develop the mathematical tools to extract this ``anyonic charge content'' from first principles, and show that this is indeed an invariant of topological phases.
To make the latter statement more precise, recall that
a standard notion of equivalence of gapped ground states of quantum spin systems considers two states to 
be in the same phase if they are ground states of local Hamiltonians that are gapped and that can be continuously 
deformed into each other without closing the gap, possibly requiring in addition that certain symmetries are 
preserved~\cite{ChenGW,ChenGW2}. Instead of just continuity, one can demand the stronger condition that the 
path of Hamiltonians is piecewise $C^1$~\cite{BachmannO,Ogata1,Ogata2,Ogata3}, presumably without loss
of generality. Since we wish to use results from~\cite{BachmannMNS}, in this work we will only consider such 
piecewise differentiable paths.
A typical way of obtaining such a path of gapped ground states is by gently perturbing the dynamics of the original system.
In this paper we show that the structure of the anyons is stable under such perturbations of the dynamics.

Since we are interested in gapped quantum phases, we need to know that the perturbation does
not close the gap. Stability properties of the spectral gap for topologically ordered states have been well 
studied. Bravyi, Hastings and Michalakis~\cite{BravyiHM} first showed that for
topologically order systems with frustration-free Hamiltonians and commuting interaction 
terms the spectral gap is stable under local perturbations, uniform in the system size.  
A topologically ordered state therefore belongs to a gapped ground state phase~\cite{BachmannMNS},
and further, this phase is a topological phase. This result was generalized to Hamiltonians with non-commuting
terms by Michalakis and Zwolak~\cite{MichalakisZ}. Existing results regarding the stability properties of the 
anyon structure are however rather limited, but see~\cite{BravyiHM,Haah,KatoN,KitaevHC}.
Our methods are applicable to a wider class of models and retain more of the structure of the anyons.

We will now outline the main ideas behind our approach and the technical challenges that we have to address.
Mathematically the properties of the anyons are described by a \emph{modular tensor category} (see for example~\cite{Wang}).
A modular tensor category is in particular a fusion category, and these have a property known as \emph{Ocneanu rigidity}.
This essentially says that such categories do not have non-trivial deformations. That is, one cannot continuously deform them to obtain 
a new, inequivalent fusion category. Following ideas of Ocneanu, Blanchard and Wassermann, this has been shown for the 
more general ``multi-fusion'' case in~\cite{ENO}. While this suggests that the anyonic structure indeed should be stable, it 
does not provide a satisfactory answer to our question. In particular, what is missing is that it is not clear how to obtain a 
modular tensor category in the first place after perturbing the dynamics. Hence, the question whether the coefficients that can 
be used to define the modular tensor category depend continuously on, for example, the strength of the perturbation is not 
a priori well defined.

To address this issue we present a method that allows us to define a braided tensor category describing anyonic charges 
that is robust with respect to general, uniformly small perturbations of the model Hamiltonian. Our approach to describing 
the superselection sectors is inspired by work in local quantum physics, starting with Doplicher, Haag and Roberts (DHR) 
in the 1970's~\cite{DHR1,DHR2}, who considered bosonic and fermionic sectors. Later, in the late 80's, this was extended 
by Fredenhagen, Rehren, and Schroer, and later also by Fr\"ohlich and Gabbiani to describe sectors with braided 
statistics~\cite{FRSI,FRSII,FroehlichGabbiani}. Although these authors were interested in relativistic (and conformal) field 
theories in low dimensions, many of the ideas developed there can be adapted to quantum spin systems after taking into 
account some fundamental technical differences~\cite{Naaijkens11}.

The key realization in this approach is that charged states can be obtained by composing localized (that is, acting non-trivially 
only on a part of the system) endomorphisms of the observable algebra with the ground state (or vacuum) of the quantum system.
They are also transportable: for a different localization region one can find a unitary equivalent (to the original) endomorphism 
that is localized in that region. Physically these two conditions mean that charges can be localized in suitable regions, and can 
be moved around in the system. This leads to a semigroup of such endomorphisms, the semigroup operation being composition 
of morphisms (or \emph{fusion}). Studying this semigroup along the lines of the DHR program, one can for example recover the 
full modular tensor category in the abelian quantum double models~\cite{FiedlerN,Naaijkens11}.

For perturbations of the quantum double models, however, the analysis of previous approaches breaks down.
The reason is that generally it is not possible to localize charges \emph{exactly} in a region,
which is a prerequisite for the DHR approach to work.
Consider a path of gapped ground states of local Hamiltonians, informally denoted by $H(s) = H(0) + \Phi_s$, where $s\in \mathbb{R}$ is a parameter,
for example related to the strength of the perturbation. Then, generically $\omega_0$, a ground state of $H(0)$, is not a ground state of the perturbed dynamics, 
and the question is what the superselection sectors are with respect to the new ground (or reference) state. 
If these perturbations are sufficiently regular and $s$ small enough, the results mentioned above imply that the gap stays open~\cite{BravyiHastingS,BravyiHM,MichalakisZ}. This allows us to apply the spectral flow automorphism of~\cite{BachmannMNS}, 
which yields a cocycle $s \mapsto \alpha_s \in \Aut(\calA)$, such that $\omega_s \equiv \omega_0 \circ \alpha_s$ is a ground state of $H(s)$.

Naively, since $\alpha_s$ is an automorphism, it seems reasonable to expect that one could just compose the semigroup of automorphisms from the DHR analysis with this $\alpha_s$ to get the sectors of the perturbed theory.
This is however much more subtle: $\alpha_s$ in general does not send local observables to strictly local observables.
Hence the sectors obtained in this way lose their strict localization property, and the DHR analysis breaks down.
A related reason is that in the definition of the tensor product (which describes fusion and braiding of the anyons), it is necessary to extend the endomorphisms to slightly larger \emph{von Neumann} algebras.
The automorphisms $\alpha_s$ however act on the level of the $C^*$-algebra of observables, and it is far from clear if they can be extended to the von Neumann algebras.

Fortunately, however, $\alpha_s$ is quasi-local, in the sense that it satisfies a Lieb-Robinson type of bound.
See~\cite{NachtergaeleSY} for a detailed discussion of quasi-local maps.
In particular, while $\alpha_s(A)$ may not be strictly local for local $A$, it can be approximated by local operators where the error term decreases faster than 
any power law in the size of the support for a suitably chosen sequence of approximating operators. Motivated by 
these observations, we 
introduce the notion of an \emph{almost localized endomorphism}, generalizing the notion of a localized endomorphism.
Again, equivalence classes of such endomorphisms will correspond to the different types of anyons.
To get the full structure of the anyons, we define a braiding of such endomorphisms. 

Our study of the almost localized endomorphisms is motivated by the notion of bi-asymptopia introduced by Buchholz et al.~\cite{BuchholzAA}.
Qualitatively they consider a similar situation, where one has to deal with a lack of strict localization.
We show that endomorphisms satisfying a certain superselection criterion lead to bi-asymptopia in a natural way.
This criterion encodes the desired physical properties: almost localization and the ability to move anyons around.
The bi-asymptopia can then in turn be used to define the full braided tensor structure on the sectors by using the results of~\cite{BuchholzAA}.
And moreover, the result agrees with the DHR analysis in the case of strict localization.

The construction of the bi-asymptopia hinges on the almost locality of the endomorphisms.
Hence if applying some automorphism preserves this almost locality, this leads to a new set of bi-asymptopia.
This is in particular true for the automorphisms $\alpha_s$ coming from the spectral flow.
They satisfy a Lieb-Robinson type of bound.
Although these are usually stated for observables of finite support, they can be extended to quasi-local observables localized in cone-like regions.
Such regions are precisely the localization regions of interest for 2D topologically ordered models with anyons.
These bounds can then be used to show that one can ``perturb'' the original almost localized endomorphisms, and obtain a new superselection theory that is equivalent to the original one.


The result of this analysis is then applied to obtain the second major result of this paper: the stability of abelian quantum double models.
We show that under an additional (but natural) assumption on the energy of certain excitations, using the automorphic equivalence result of~\cite{BachmannMNS}, that the superselection structure is invariant under a broad range of perturbations of the model.
That is, the perturbed model and the unperturbed model give rise to the same braided tensor $C^*$-category describing the anyonic excitations.

The paper is organized as follows. In the next section, we first recall the basic setting, fix our notations and assumptions, 
and give an overview of our main results. Section~\ref{sec:superselect} discusses the superselection criterion and 
contains the proof of the main result. In Section~\ref{sec:stability} we consider suitable perturbations of the dynamics 
and their effect on the superselection structure. We conclude by applying the general theory to the class of abelian 
quantum double model in Sec.~\ref{sec:applications}. The proof of Lieb-Robinson bounds for cone localized observables 
is deferred to the Appendix, as are a brief overview on the spectral flow automorphisms, and on braided tensor categories.
Some results in this work were first reported as part of the Ph.D. dissertation of MC~\cite{Cha}.

\emph{Acknowledgements:} MC would like to thank the University of Tokyo, for support and hospitality during the summer 
of 2015 where part of the work on this paper was done.  MC was supported in part by the National Science Foundation 
under Grant OISE-1515557 and the Japan Society for the Promotion of Science Summer Program 2015. 
PN thanks Courtney Brell for helpful conversations and suggestions, and Luca Giorgetti and Michael M{\"u}ger for 
discussions on Appendix~\ref{app:braided}. PN has received funding from the European Union's Horizon 2020 research 
and innovation program under the Marie Sklodowska-Curie grant agreement No 657004 and the European Research Council (ERC) Consolidator Grant GAPS (No. 648913).
BN and MC were supported in part by the National Science Foundation under Grant DMS-1515850.

\section{Assumptions and main results}\label{sec:assumptions}
In this section we first state our main assumptions and give an overview of the most important results.
For the convenience of the reader, we recall the basic definitions and introduce our notation. For further background
on the $C^*$-algebraic approach to quantum spin systems see for example~\cite{BratteliR1,BratteliR2}.

Let  $\Gamma = \ZZ^\nu$ be the cubic lattice in $\RR^\nu$ with metric $d(x,y) = \abs{x-y}$.
We choose a cubic lattice for convenience. Most results can be easily generalised to other underlying sets, as
long as $\Gamma$ is regular enough (e.g., the number of points in a ball of size $r$ does not grow faster than
a power of $r$).
We will mostly work in the case $\Gamma = \ZZ^2$.
To each $x\in \Gamma$ assign a finite dimensional Hilbert space $ \calH_x$.
Let $\mathcal{P}_0(\Gamma)$ denote the set of finite subsets of $\Gamma$.
For $ \Lambda\in \mathcal{P}_0(\Gamma)$  we
define the Hilbert space of states for the composite system as the tensor product space $ \calH_\Lambda \equiv \bigotimes_{x\in \Lambda} \calH_x$
with the corresponding algebra of observables $\calA_\Lambda = \calB(\calH_\Lambda)$.
If $ \Lambda_1 \subset \Lambda_2$ there is a natural inclusion of 
$ \calA_{\Lambda_1} \hookrightarrow \calA_{\Lambda_2}$ via the map $ A \mapsto A\otimes I_{\Lambda_2\backslash \Lambda_1}$.
This gives a local net of $C^*$-algebras and allows us to define the algebra of 
\emph{local observables} as 
$\calA_{loc} = \bigcup_{\Lambda \in \mathcal{P}_0(\Gamma)} \calA_\Lambda$
and the $C^*$-algebra of {quasi-local observables} as the norm completion
$\calA = \overline{\calA_{loc}}^{\| \cdot \|}.$

Let $X \subset \Gamma$ be a potentially infinite set.
We define the quasi-local algebra of observables in $X$ as the $C^*$-subalgebra 
$\calA_X \equiv \overline{ \bigcup_{\Lambda\in \mathcal{P}_0(X)} \calA_{ \Lambda}}^{\| \cdot \|} \subset \calA.$
An observable $A$ is said to be \emph{localized} in a set $X$ if $ A \in \calA_X$.
The \emph{support} of $A$ is defined as the smallest set $X$ such that $A \in \calA_X$.
Notice that when $X \in \mathcal{P}_0(\Gamma)$ then we have $\calA_X \cong \calB(\calH_X)$, and that $\calA_\Gamma = \calA$.

A \emph{state} on $\calA_\Gamma$ is a linear functional $\omega: \calA_\Gamma \ra \CC$ such that 
$\omega(A) \geq 0$ if $A\geq 0$  and $ \omega(I) = 1$.  
The set of all states is denoted $ \calA_{+,1}^*$
and is a weak$^*$-compact convex set, its extremal points are called \emph{pure states}.
For each state $\omega$, the GNS construction gives a triple $( \pi_{\omega}, \Omega_\omega, \calH_{\omega}) $ where $ \pi_{\omega}: \calA_\Gamma\ra \calB(\calH_\omega)$
is a cyclic $*$-representation, $\Omega_\omega$ is a cyclic vector, and 
$\omega(A) = \langle \Omega_\omega, \pi_{\omega}(A) \Omega_\omega\rangle$.
It can be shown that $ \omega$ is a pure state if and only if its GNS representation is irreducible.
Since the quasi-local algebra is simple, and thus each $*$-representation is faithful, we will often abuse notation when the context is clear and identify  $ \pi_\omega(A) = A$.

\subsection{Reference state and main assumptions}
The discussion above applies to general quantum spin systems.
For our analysis we need a reference state $\omega_0$, roughly playing the role of the vacuum in relativistic quantum field theories.
We will first describe some of the assumptions on $\omega_0$ that we require.
In particular, let $\omega_0$ be a \emph{pure state}, and write $(\pi_0, \Omega_0, \calH_0)$ for the corresponding GNS representation.
As mentioned earlier, in concrete models there typically is a canonical choice for the reference state.
The example of abelian quantum double models will be discussed below.
In addition to purity, we will require some additional properties.
Again, these can be shown to hold in abelian quantum double models.
To introduce these assumptions it is necessary to first introduce the notion of a \emph{cone}.
They should be thought of as the localization regions of the excitations of the models of interest.
Cones are chosen here partly for convenience, the precise shape is not that important:
what is important is that they extend to infinity and are without holes in the interior.
For technical reasons we also want them to be such that any finite region can be translated
to the interior of the region. 

For concreteness, we will use round cones, with axis $a\in\RR^\nu$, $\Vert a\Vert =1$, and
opening angle $2\alpha$. We first define cones with apex at the origin and then obtain 
the ones with arbitrary apex by translation. Let $\bar{\Lambda}_\alpha(a)$ be the infinite open 
cone in $\RR^\nu$ with apex at the origin, parametrized by  a unit vector $a\in\RR^\nu$ and 
$\alpha\in (0,\pi)$:
$$
\bar{\Lambda}_\alpha(a) =  \left\{ x\in\RR^\nu \mid x\cdot a >\Vert x\Vert  \cos \alpha  \right\}.
$$
For $x \in \RR^\nu$, we will write $\bar{\Lambda}_\alpha+x$ for the translated cone. We will write
$\mathcal{C}$ for the set of all such cones and $\mathcal{C}_\alpha$ for the set with opening angle
less than $2\alpha$. 
To simplify notation we will typically suppress the axis $a$ in the notation.
Note that $\RR^\nu$ itself is not a cone.

If the spin variables of the model are associated with the vertices, as above, we will consider 
cones $\Lambda_\alpha$ in the lattice defined by $\Lambda_\alpha \equiv \bar{\Lambda}_\alpha \cap \ZZ^\nu$.
In models such as Kitaev's quantum double model the spins are associated with edges instead of vertices.
In that case, we identify a cone with all \emph{edges} such that both endpoints lie within the cone $\bar{\Lambda}_\alpha \subset \RR^\nu$.
As a special case of translations, we also define $\Lambda_\alpha + n$ for integers $n \in \mathbb{Z}$ as the translation over $n$ times the unit vector pointing along the axis of the cone.
By convention, we choose this vector such that $\Lambda_\alpha + n \subset \Lambda_\alpha$ for positive $n$.
Note that translations behavior naturally with respect to the operation of set complement, $( \Lambda_\alpha - n)^c = \Lambda_{ \alpha  }^c + n$.

For some of our constructions it is necessary that two cones, one contained in the other, have boundaries that are sufficiently far apart.
Or equivalently, we want the complement of one cone region to be sufficiently far separated from some other cone.
Here ``sufficient'' depends on the model of interest.
More precisely, we will require a technical property on the reference representation, called the \emph{approximate split property} (or rather, the \emph{strong approximate split property}).

To define this later on, we need to consider a relation $\ll$ on $\mathcal{C}$ as follows.
Fix an integer $n_0 \geq 1$. The choice of $n_0$ is model-dependent.
The relation is $\ll$ is then defined by saying that $\Lambda \ll \widehat{\Lambda}$ if $ \Lambda \subset \widehat{\Lambda}+n_0$.
Note that $\Lambda \ll \widehat{\Lambda}$ implies $ d( \Lambda, \widehat{\Lambda}^c)  > 0$.
That is, the cones $\Lambda$ and $\widehat{\Lambda}^c$ are spatially separated.


We now return to formulating our assumptions on the reference state $\omega_0$.
Typically, $\omega_0$ would be a pure and translation-invariant ground state
for a dynamical system $(\calA, \tau_t = e^{it \delta}),$
that is,  $\omega_0 \circ T_x = \omega_0$ for all $x\in \ZZ^\nu$ where $T_x$ is the natural action of $\ZZ^\nu$,
and $ \omega_0 (A^* \delta(A)) \geq 0$ for $A \in D(\delta)$.
In the standard DHR analysis of sectors,  the role of the reference state $\omega_0$ is played by the vacuum state.
For any subset $\Lambda \subset \Gamma$, denote the von Neumann closure in  $\calB(\calH_0)$ by
\begin{equation}
\mathcal{R}(\Lambda) \equiv \pi_0(\calA_\Lambda)'',
\end{equation} 
where $\pi_0$ is the GNS representation of $\omega_0$ as before.
If $\Lambda$ is an infinite set the weak-operator closure $\mathcal{R}(\Lambda)$ depends on $\pi_0$.

Clearly, if $\Lambda \ll \widehat{\Lambda}$ then the inclusion of quasi-local algebras $ \calA_\Lambda \subset \calA_{ \widehat{\Lambda}}$ implies the 
inclusion of cone algebras $\mathcal{R}(\Lambda) \subset \mathcal{R}(\widehat{\Lambda})$.
Usually these cone algebras are factors, but not of Type I. 
For instance, in the frustration-free ground state of the Kitaev quantum double models the cone algebras are non-Type I infinite factors~\cite{Naaijkens11}. 
We say that the inclusion is \emph{split}~\cite{DoplicherL} if there is a type I factor $\mathcal{N}$ such that $\mathcal{R}(\Lambda) \subset \mathcal{N} \subset \mathcal{R}(\widehat{\Lambda})$.
Our main technical assumption on $\omega_0$ is that it satisfies a variant of the split property for the set of infinite cones:

\begin{defn}\label{def:strongsplit}
	We say that $\omega_0$ satisfies a \emph{strong approximate split property} for $\mathcal{C}$ if for each pair $ \Lambda, 
	\widehat{\Lambda} \in \mathcal{C}$  with $\Lambda \ll \widehat{\Lambda}$ there exists a $N>0$ and a type I factor $\mathcal{N}$ such that
	\begin{equation}\label{eqn:approxsplit}
	\mathcal{R}(\Lambda+N) \subset \mathcal{N} \subset \mathcal{R}(\widehat{\Lambda})
	\end{equation}
	and 
	\begin{equation}\label{eqn:strongsplit}
	\mathcal{R}(\widehat{\Lambda}^c) \subset \mathcal{N}' \subset \mathcal{R}((\Lambda+N)^c).
	\end{equation}
\end{defn}

Recall that $\ll$ was defined in terms of some integer $n_0$.
In applications we choose this integer so that the strong approximate split property as defined here holds.
Notice that if $n>N$ then since $\mathcal{R}(\Lambda+n) \subset \mathcal{R}(\Lambda+N)$ and $ \mathcal{R}((\Lambda+N)^c) \subset \mathcal{R}((\Lambda+n)^c)$
we can replace $n$ for $N$ in Eqs.~\eqref{eqn:approxsplit} and~\eqref{eqn:strongsplit}.

The strong approximate split property can be interpreted as a statistical independence of the regions
$\Lambda$ and $\widehat{\Lambda}^c$ \cite{DoplicherL,Werner}.
This will be used in Section~\ref{sec:intertwiners} to establish locality properties of intertwiners.
The strong approximate split property implies that if $\Lambda \ll \widehat{\Lambda}$ then the von Neumann algebra generated by the regions $\Lambda$ and $\widehat{\Lambda}^c$ satisfy  $ \mathcal{R}(\Lambda) \vee  \mathcal{R}(\widehat{\Lambda}^c) \simeq \mathcal{R}(\Lambda) \overline{\otimes} \mathcal{R}(\widehat{\Lambda}^c)$~\cite{DoplicherL}.
Here $\mathcal{R}(\Lambda) \vee \mathcal{R}(\widehat{\Lambda}^c)$ is the smallest von Neumann algebra generated by the two algebras, while $\mathcal{R}(\Lambda) \overline{\otimes} \mathcal{R}(\widehat{\Lambda}^c)$ is the von Neumann algebraic tensor product.
The symbol $\simeq$ signifies that the two von Neumann algebras are \emph{naturally} isomorphic, in the sense that the map $AB \mapsto A \otimes B$ for $A \in \mathcal{R}(\Lambda)$ and $B \in \mathcal{R}(\widehat{\Lambda}^c)$ extends to a normal isomorphism of $\mathcal{R}(\Lambda) \vee \mathcal{R}(\widehat{\Lambda}^c)$ and $\mathcal{R}(\Lambda) \overline{\otimes} \mathcal{R}(\widehat{\Lambda}^c)$.
In the case that $\omega_0$ is a ground state, we expect the strong approximate split property to be closely related to the  existence of a spectral gap.
For quantum spin chains, i.e., one-dimensional systems, Matsui showed the existence of a spectral gap implies a (non-approximate) split property in the ground state~\cite{Matsui13}. In dimensions $\geq 2$, however, we do not have a general 
theory that demonstrates the split property but we do have interesting examples, such as the toric code model, for which it holds.

The condition on split inclusions we consider is approximate in the sense that if $\Lambda \ll \widehat{\Lambda}$ then
the regions $\Lambda$ and $ \widehat{\Lambda}^c$ are spatially separated, whereas the split property means 
that $\mathcal{R}(\Lambda) \subset \mathcal{N} \subset \mathcal{R}(\Lambda^c)'$, see for example~\cite{Matsui01}.
That is, there is only a minimal separation assumption on the volumes $\Lambda$ and $\Lambda^c$.
In algebraic quantum field theory a similar split condition is considered.
The condition we assume is \emph{strong} in that sense that each split inclusion \eqref{eqn:approxsplit} is paired with a dual 
split inclusion \eqref{eqn:strongsplit}.

It is useful to relate Eq.~\eqref{eqn:strongsplit} to Haag duality for cones, that is,  
$\mathcal{R}(\Lambda)' = \mathcal{R}(\Lambda^c)$.
In particular, Eq. \eqref{eqn:approxsplit} and Haag duality imply \eqref{eqn:strongsplit}. Let  $\Lambda \ll \widehat{\Lambda}$ 
and suppose there exists a type I factor $\mathcal{N}$ such that Eq.~\eqref{eqn:approxsplit} holds.
Then, taking commutants, we have $\mathcal{R}(\Lambda) \subset \mathcal{N} \subset \mathcal{R}(\widehat{\Lambda})
\iff
\mathcal{R}(\widehat{\Lambda})'  \subset \mathcal{N}' \subset \mathcal{R}(\Lambda)'$.
With Haag duality the right hand side is equivalent to 
$\mathcal{R}(\widehat{\Lambda}^c)  \subset \mathcal{N}' \subset \mathcal{R}(\Lambda^c).$
Thus, the type I factor $\mathcal{N}'$ gives a split inclusion for the cone algebras $ \mathcal{R}(\widehat{\Lambda}^c) 
\subset \mathcal{R}(\Lambda^c)$.
From this it follows that the strong approximate split property holds for the frustration-free ground state of Kitaev's abelian 
quantum double models~\cite{FiedlerN}. It should also be noted that the split property itself does not hold for cones, 
since $\mathcal{R}(\Lambda)$ is not a Type I factor~\cite{Naaijkens11}.

In our results we do not need to assume Haag duality, and therefore only assume that the strong approximate split property holds.
However, we do not know of any models for which the strong approximate property holds but Haag duality fails.

\subsection{Superselection structure}
Superselection sectors can be identified with equivalence classes of (irreducible) representations of the observable 
algebra $\calA_\Gamma$.  Physically, this means that it is not possible to make coherent superpositions of vector 
states in different representations,  or equivalently, it is not possible to obtain a vector state in one representation
by acting with observables on a vector state in an inequivalent representation.
This can be interpreted as a form of charge conservation: it is not possible to go from a charged state to a state with a different (total) charge.
Hence it stands to reason that one can learn something about these charges (or, in our case, anyons) by studying sets of representations of $\calA_\Gamma$ and their intertwiners.

The first issue that one encounters is that a general $C^*$-algebra has very many inequivalent irreducible representations, 
most of which are not physically relevant (for example because they describe states with infinite energy).
A \emph{superselection criterion} is a criterion to select the relevant representations.
For example, for models such as the toric code, the following criterion is natural~\cite{Naaijkens11}.
Consider representations $\pi$ such that for \emph{any} cone $\Lambda$, there is a unitary $U$ (that may depend on $\Lambda$) such that
\begin{equation}
	\label{eq:superselect}
	U \pi(A) U^* = \pi_0(A) \quad \text{for all} \quad A \in \calA_{\Lambda^c}.
\end{equation}
This criterion means that we select those representations $\pi$ of $\calA_\Gamma$ that are unitarily equivalent to the reference representation $\pi_0$ \emph{when restricted to observables outside an arbitrary cone $\Lambda$}.
The interpretation there is that charges can be localized in a cone, and that moreover they are transportable, in the sense 
that we can move the localization region around with unitary operators. In conformal or relativistic theories there are analogous criteria, depending on the type of charges one wants to describe~\cite{BuchholzF,DHR1,GabbianiF}.

The second issue is that the set of representations has relatively little structure, compared for example to endomorphisms 
of $\calA_\Gamma$, which can be composed.
It is therefore more convenient to work with endomorphisms.
For the quasi-local algebra $\calA_{ \Gamma}$, it is known that any cyclic representation 
is equivalent to one described by an asymptotically inner endomorphism \cite{Kishimoto}.
That is, let $\pi$ be a cyclic representation and $\pi_0$ the irreducible reference representation.
Then there is a continuous path $\{U_t\}_{t \in [0,\infty)}$ of unitaries in $\calA_\Gamma$ such that $\rho(A) = \lim_{t \to \infty} U_t A U_t^*$ for all $A \in \calA_\Gamma$ defines an endomorphism with $\pi$ unitarily equivalent to $\pi_0 \circ \rho$.
Thus, considering only endomorphisms may not be so restrictive.
The downside is however that one loses control of locality in this approach.
Alternatively, one can use Haag duality to pass from representations to endomorphisms which \emph{do} enjoy good locality properties~\cite{BuchholzF,DHR1,FiedlerN}.
Here we do not assume Haag duality, but bypass this step by first restricting to representations that are obtained by 
composing $\pi_0$ by an endomorphism that has suitable localization properties. 
It is possible that in some situations it may be of interest to compare the analysis 
of sectors obtained by endomorphisms to the potentially larger family obtained by representations,
but for our purposes all representations of interest can be obtained by suitable endomorphisms.

We now introduce the primary objects of interest.
Let $\Delta$ be a semi-group (with respect to composition) of $*$-endomorphisms on $\calA_\Gamma$.
Let $\mathcal{T}_\Delta$ denote the set of all intertwiners $T$ of pairs of endomorphisms in the semigroup,
i.e., $(\rho,\sigma)\in \Delta\times\Delta$:
$$
T\rho(A) = \sigma(A) T, \quad A\in \calA_\Gamma.
$$
Then $(\Delta,\mathcal{T}_\Delta)$ is a $C^*$-category whose objects are $\rho \in \Delta$ and arrows are 
intertwiners $\mathcal{T}_\Delta$.
In the following, we introduce a \emph{superselection criterion} on the endomorphisms of $\calA_\Gamma$,
to obtain the semi-group $\Delta$ which has objects that are \emph{almost localized} and \emph{transportable}
with respect to the reference state  $\omega_0$. 
In that case it is not sufficient to restrict to intertwiners in $\calA_\Gamma$ and it is necessary to redefine $\mathcal{T}_\Delta$ in terms of the reference representation $\pi_0$.

For our purposes it is necessary to have good control on the localization of the endomorphisms. To describe this,
let $ \mathcal{F}_\infty$ denote the family of functions $f: \RR^+\times \RR^+ \ra \RR^+$ such that 
$f_\epsilon(n) \equiv f(\epsilon,n)$ is non-increasing in both variables and 
$\lim_{n \ra \infty } n^{k} f_\epsilon(n)  = 0 $ for all $k \in \NN.$
Generally, we say that $ f_\epsilon \in O(n^{-\infty})$, since it goes to zero faster than any power of $n$.
Notice that $\mathcal{F}_\infty$ is closed under addition.
The family $\mathcal{F}_\infty$ will be used as a measure of locality for both operators and endomorphisms of $\calA_\Gamma$.

\begin{defn}\label{defn:alend}
	A $*$-endomorphism $\rho$ is said to be \emph{almost localized} in a cone $\Lambda_\alpha \in \mathcal{C}$ if 
	there exists a function  $f \in \mathcal{F}_\infty$  such that 
	\begin{equation}
	\sup_{A \in \calA_{\Lambda^c_{\alpha+\epsilon}+n}} \frac{\| \rho(A) - A\|}{\|A\|} \leq f_\epsilon(n)  
	\quad \mbox{for all  \quad } 0< \epsilon< \pi - \alpha, \quad n \in \NN.
\end{equation}
The function $f$ is called the \emph{decay function} for $\rho$ in $\Lambda_\alpha$.
\end{defn}
Note that the definition says that $\rho$ is close to the identity outside a shifted cone with a slightly wider opening angle
than $\alpha$, and the approximation improves the further we shift away the cone.
In applications, the localization cone $\Lambda_{ \alpha  }$ will be \emph{convex}, that is, $ 0<\alpha< \pi/2$.
By convention we do not consider the open half-space ($\alpha =\pi/2$) as a convex cone.

For a $*$-endomorphism $ \rho$, we consider the cyclic representation of the form $(\pi_0\circ \rho, \Omega_0, \calH_0)$.
Let $ \rho \cong \rho'$ denote unitary equivalence of the corresponding representations, $ \pi_0\circ \rho \cong \pi_0 \circ \rho'$.
Since $\pi_0$ is faithful we will often abuse notation to write $ \pi_0 \circ \rho$ as simply $ \rho$.

\begin{defn}\label{defn:transportable}
Let $\rho$ be almost localized in $\Lambda_\alpha$ with a decay function $f$.
	We say that $\rho$ is \emph{transportable} with respect to $\omega_0$ if for each cone 
	$\Lambda'_\beta \in \mathcal{C}$ there exists a $*$-endomorphism $\rho'_\beta$ almost localized in $\Lambda'_\beta$ 
	such that $\rho\cong \rho'_\beta$.
	Further, if $\beta \geq \alpha$ then $\rho'_\beta$ can be chosen to have decay function $f$ in $\Lambda'_\beta$.
	Finally, we assume that if $\beta < \alpha$, there is some decay function $g$ that works for all cones $\Lambda'_\beta \in \mathcal{C}_\beta$.
\end{defn}
The last condition guarantees that if we transport to a smaller cone, the resulting almost localized endomorphism is still transportable according to our definition.
That is, we can consider smaller cones, at the expense of having to choose a possibly worse decay function.
For our purposes it would be enough to require this only for angles above some small enough minimum angle $\beta_0 > 0$, but relaxing this condition complicates the proofs.

When the context is clear, we simply say that $ \rho$ is transportable.
In translation invariant models, one might expect the sectors to be translation covariant, that is $ T_{-x} \circ \rho\circ T_x \cong \rho$ for all $x \in \ZZ^\nu$. 
This would give transportability with respect to translations of cones.
In addition to this, however, we need to be able to rotate cones and decrease their opening angles as well.
These symmetries do not have a natural action on the lattice.
In any case, for our results we do not need covariance, and hence we do not make this assumption.

\begin{defn}\label{def:superselection}
	A semi-group $\Delta$ of endomorphisms of $\calA_\Gamma$ is said to satisfy the almost localized 
	and transportable \emph{superselection criterion} for $\omega_0$ (or for its GNS representation $\pi_0$) if
 for all $\rho \in \Delta$ there exists a cone $\Lambda_{ \alpha  }^\rho \in \mathcal{C}$ such that $\rho$ is almost localized in $\Lambda_{ \alpha  }^\rho$, and
 $\rho$ is transportable with respect to the state $\omega_0$.
\end{defn}

Note that transportability is defined with respect to the representation $\pi_0$.
Consequently, it is natural to also consider intertwiners with respect to this representation, and redefine $\mathcal{T}_\Delta$ to be the set of operators $T \in \mathcal{B}(\mathcal{H}_0)$ such that $T$ intertwines $\pi_0 \circ \rho$ and $\pi_0 \circ \sigma$ for some $\rho,\sigma \in \Delta$.
To emphasize the dependence on $\pi_0$ we will use the notation $T \in (\rho,\sigma)_{\pi_0}$, although if the reference representation is clear from the context we sometimes drop the subscript.
With this notation, $\rho \cong \sigma$ precisely if there is some unitary $T \in (\rho,\sigma)_{\pi_0}$.
An equivalence class of such endomorphisms $\rho$ such that $\pi_0 \circ \rho$ is irreducible is called a \emph{superselection sector} or simply a \emph{sector}.

Now let $\mathcal{A}_\Delta$ be the $C^*$-subalgebra of $\calB(\calH_0)$ generated by $\calA_\Gamma$ and $\mathcal{T}_\Delta$.
A technical difficulty in extending $\Delta$ to a tensor $C^*$-category, and thus obtaining fusion rules for $\Delta$, is that the 
endomorphisms $\rho \in \Delta$ are defined on $\calA_\Gamma$ but do not necessarily have a unique extension to $\calA_{ \Delta}$.
A key result in this paper is that this difficulty can be overcome. In particular, it shows that if $\Delta$ satisfies the almost localized and transportable superselection criterion 
then $\Delta$ is a tensor $C^*$-category.

\begin{thm}\label{thm:asymptopia}
    Suppose $\omega_0$ satisfies the strong approximate split property for cones.
    Let $\Delta$ be a semi-group of endomorphisms satisfying the almost localized and transportable superselection criterion~\ref{def:superselection}
    and suppose that each $\rho \in \Delta$ is almost localized in a convex cone $\Lambda_{ \alpha}^\rho$.
    Suppose	there exists a convex cone $K \in \mathcal{C}$ such that 
    $\Lambda_{ \alpha }^\rho \ll (K+n^\rho)^c$ for all $\rho$ and some $n^\rho \in\NN$.
	Then, each $\rho \in \Delta$ has a unique extension to a $*$-endomorphism $\widehat{\rho}$ on $\calA_{ \Delta}$
	and $\Delta$ is a tensor $C^*$-category. Furthermore, if $\nu \geq 2$ then $\Delta$ is a braided tensor $C^*$-category.
\end{thm}

This formulation of the theorem follows directly from Theorems~\ref{thm:Asymptopia} and~\ref{thm:biasymptopia} below.
The proof, which is inspired by~\cite{BuchholzAA}, can be divided into two main steps.
First, we prove norm locality estimates on the intertwiner maps $\mathcal{T}_\Delta$, which do not necessarily belong to the quasi-local algebra. 
Indeed, we show that if $T \in (\rho, \sigma)$ and $\rho$ and $\sigma$ are almost localized in a cone $\Lambda_{ \alpha  }$,
then $T$ is `almost localized' in the same cone.
Second, we construct an \emph{asymptopia}, as introduced in \cite{BuchholzAA}, for $\Delta$.
More specifically, for each $\rho, \rho', \sigma, \sigma' \in \Delta$ we construct sequences of unitaries such that
$\rho(A) = \lim_{n\ra\infty} U_n^* A U_n$ and $\sigma(A) = \lim_{n \ra \infty } V_n^* A V_n$,
and given $R \in (\rho, \rho')$ and $R' \in (\sigma,\sigma')$,
\begin{equation}\label{eqn:asympabelian}
\lim_{m,n\ra \infty} \| [ V_n R U_m^*, R'] \| = 0.
\end{equation}
Heuristically, the limit above describes a procedure in which the intertwiner $R$ is moved to  infinity in a direction disjoint from the support of $R'$,
and thus the interwiners are said to be \emph{asymptotically abelian}.
Eq. \eqref{eqn:asympabelian} is motivated  by the following formal calculation,
$$
	\| R \rho(R')  -  \sigma(R') R  \|  = \lim_{m,n \ra \infty} \| R U_m^* R' U_m - V_n^* R' V_n R \| = \lim_{m,n \ra \infty}  \|V_n R U_m^* R' -  R' V_n RU_m^* \|,
$$
where this is formal in the sense that $\rho(R')$ and $\sigma(R')$ may not be well defined, since the intertwiners need not be in $\calA_\Gamma$.
When $\nu \geq 2$ a similar construction leads to a \emph{bi-asypmtopia} for $\Delta$,
and further, the results of \cite{BuchholzAA} give a braided tensor $C^*$-structure for $\Delta$.

The cone $K$ can be interpreted of as a forbidden direction
and is used in the construction of the auxiliary $C^*$-algebra \cite{BuchholzF, Naaijkens11},
$\calB_K\equiv  \overline{ \bigcup_{x\in \ZZ^\nu} \mathcal{R}((K+x)^c)}^{\| \cdot\|} =  \overline{ \bigcup_{n \in \NN } \mathcal{R}((K+n)^c)}^{\| \cdot\|}$.
Indeed we will show that the intertwiners belong to the auxiliary algebra.
In principle, the tensor structure on $\Delta$ may depend on the choice of $K$.
However, we show in Lemma~\ref{lem:indK} that is not the case.

The use of the cone $K$ is necessary mainly for technical reasons.
In particular, we need to extend our endomorphisms to endomorphisms of a larger algebra which contains the intertwiners, in order to define the monoidal product.
The cone $K$ provides us with a canonical way to define this algebra (see equation~\ref{eqn:auxalg} below).
In the case that we consider, the choice also provides a canonical choice of direction in which we can move charges to infinity, in such a way that we can guarantee that no other charges are in the way.
This will be useful to prove various commutativity bounds, which are crucial in extending the approximately localized endomorphisms to act on the intertwiners.
At the same time, using the same auxiliary algebra as in the unperturbed case~\cite{Naaijkens11} makes it easier to relate these results to our present analysis.
And finally, in the 2D case that we are interested in, it can be used to define unambiguously when a cone is to the ``left'' of another one, which is necessary in order to define a braiding.
We will come back to this point later.
As a final remark we mention that it is possible to avoid the use of a forbidden direction by considering different ``coordinate patches'' (see for example~\cite{FroehlichGabbiani}).
This would however complicate our analysis for an end result which is essentially the same.

\subsection{Dynamics of quantum spin systems and Lieb-Robinson bounds}\label{sec:dynamics}
To discuss specific models, we have to define the dynamics of the model.
This allows us to talk about ground states, and more generally, about quantum phases.
For our problem ground states play a fundamental role: typically in the models of interest to us, there is a ``preferred'' 
ground state, for example because it is the only ground state that is invariant with respect to translations. This then gives a 
natural reference state $\omega_0$, and via the GNS construction also a reference representation, which is necessary for the discussion 
of superselection sectors. This is analogous to the role the vacuum plays in relativistic theories.
We will now state our assumptions on the dynamics.

To define dynamics for a quantum spin system we first recall the notion of interactions.
 An \emph{interaction} is a map $\Phi: \mathcal{P}_0(\Gamma) \ra \calA_{loc}$ such that
$\Phi(X) \in \calA_X $ and $ \Phi(X)^* = \Phi(X)$.
For a finite subset $\Lambda \subset \Gamma$, the local Hamiltonians and the Heisenberg dynamics corresponding to $\Phi$ are given, respectively, by
\[ H_\Lambda \equiv \sum_{X\subset \Lambda} \Phi(X) \quad \text{ and } \quad \tau^\Lambda_t (A) \equiv e^{i t H_\Lambda} A e^{-i t H_\Lambda}.\]
They describe the time evolution due to all interactions within the region $\Lambda$.

If the interaction is sufficiently local, that is, 
the norm $\| \Phi(X) \|$ decays sufficiently fast in $\operatorname{diam}(X)$,
it will be possible to define a dynamics on the infinite system.
The full details for the range of interactions we allow will be discussed in Sec.~\ref{sec:stability}.
For the purposes of the present discussion, one can simply assume $\Phi$ to be a finite range interaction. 
That is, suppose there is a $R>0$ such that $\Phi(X) = 0$ if $ \operatorname{diam}(X) >R$ and $\| \Phi(X) \|$ is uniformly bounded.
Let  $\Lambda_n \in \mathcal{P}_0(\Gamma)$, $n\geq 1$, be an increasing and exhausting sequence
in $\Gamma$, that is, if $\Lambda_n \subset \Lambda_{n+1}$ and $\Gamma = \bigcup_{n} \Lambda_n$.
Then, the norm limit $\tau_t(A) \equiv \lim_{n \ra \infty} \tau_t^{\Lambda_n}(A) $
exists for all $t\in \RR$ and $A \in \calA_\Gamma$ \cite{BratteliR2}.
The infinite volume dynamics $\tau_t$ defines a strongly continuous, one-parameter group of automorphisms on $\calA_\Gamma$.
Locality estimates for the infinite volume dynamics are described by Lieb-Robinson bounds \cite{BratteliR2,NachOS}.
In their simplest form, it says that for local observables $A \in \calA_X, B \in \calA_Y$, we have
\[
	\| [ \tau_t(A), B ] \| \leq 2 C_{A,B} \| A \| \|B\| e^{v|t| - d(X,Y)},
\]
where $C_{A,B}$ may depend on the interaction and the size of the supports and their distance, and $v \geq 0$ is called the Lieb-Robinson velocity.
It can be seen as an analogy to the speed of light in relativistic theories.
Under quite general assumptions, which we will discuss in detail in Section~\ref{sec:LRcones}, good bounds can be obtained for $C_{A,B}$.
In particular we will need a Lieb-Robinson type of bound for observables that are localized in \emph{infinite} regions, such as the cones above.

A consequence is that the infinite volume dynamics is \emph{quasi-local} in the sense that for $A \in \calA_X$ and $\abs{X}<\infty$,
the time-evolved observable $\tau_t(A)$ can be approximated exponentially well in the set $ X + v \abs{t} + l$ for $l>0$ and some $v>0$ \cite{NachOS}. 
This again has an analog for cones.
Such types of quasi-local maps appear in various contexts, and do not always come from some dynamics.
See~\cite{NachtergaeleSY} for an overview.
In light of this, we will sometimes call a family $s \mapsto \alpha_s$ of automorphism a \emph{quasi-local dynamics} if it satisfies a Lieb-Robinson bound as above.
In fact, it is often enough to restrict to some interval, say $s \in [0,1]$.
In any case, the key point to keep in mind is that if $A$ is strictly local, $\alpha_s(A)$ can be well approximated by strictly local observables in such a way that the support of the strictly local approximations does not grow too quickly.


\subsection{Stability under deformation by a quasi-local dynamics}
The category $\Delta$ that was introduced in the previous section depends on the choice of reference representation $\pi_0$.
If the superselection structure of $\Delta$ corresponds to a system of quasi-particle excitations then it is expected 
that this structure is an invariant of a gapped ground state phase~\cite{BravyiHM,Haah,KitaevQD}.
In particular, when the quasi-particle excitations are anyons, this stability is expected to play a crucial role 
in the classification of two dimensional topologically ordered phases~\cite{KitaevHC}.
Addressing these questions is the main motivation behind the present work.

Our stability results for the superselection structure of $\Delta$ are twofold.
First, we show that if $\Delta$ satisfies the almost localization and transportability criterion 
then the superselection structure is stable under any deformation of $\Delta$ described by a quasi-local dynamics, in a sense that we will explain below.
Second, we apply this to show that the superselection sectors of anyons in the Kitaev abelian quantum double models satisfy the almost localization and transportability criterion,
and prove the superselection structure of anyons is stable under any uniform and local perturbation that does not close the spectral gap in the ground state.
This provides a general framework to study the anyon structure of a topologically ordered system and prove its stability.
 
The main result can be paraphrased as follows: given a (family of) automorphisms $\tau_t$ that act sufficiently local, the sector structure $\Delta$ is invariant under conjugation with $\tau_t$.
The precise conditions on $\tau_t$ are introduced in Section~\ref{sec:stability}, but one could think for example of dynamics satisfying a suitable Lieb-Robinson type of bound.
We note that $\tau_t$ need not be directly related to the dynamics of the underlying system.
Rather, it would typically be obtained as a ``spectral flow''~\cite{BachmannMNS}.
The important things is that it should be sufficiently local, in the sense that for local observables $A$ and $B$, the commutator $\| [\tau_t(A), B ]\|$ should decay super-polynomially in the distance between the supports of $A$ and $B$.
Alternatively, this means that $\tau_t(A)$ can be well approximated by a strictly local observable $A'$, with error decaying super-polynomially in the size of the support of $A'$.
If this is the case, and $\Delta$ is a semi-group of endomorphisms satisfying the almost localized and transportable superselection criterion for $\omega_0$,
it follows that $\tau_t^{-1} \circ \Delta \circ \tau_t$ satisfies the same criteria for the state $\omega_0 \circ \tau_t$.

\begin{thm}\label{thm:stabheuristic}
	Let $\nu \geq 2$ and assume that we have a semi-group $\Delta$ and reference state $\omega_0$ satisfying all the assumptions of Theorem~\ref{thm:asymptopia}.
	Let $t \mapsto \tau_t$ for $t \in [0,1]$ be a quasi-local dynamics in the sense of Section~\ref{sec:dynamics}.
Then endomorphisms in $\tau_t^{-1} \circ \Delta \circ \tau_t$ are localized and transportable with respect to $\omega_0 \circ \tau_t$, and the semigroup forms a tensor category that is braided tensor equivalent to $\Delta$.
\end{thm}

The proof relies on estimates for the Lieb-Robinson bounds applied to two infinite disjoint cones, in particular on how they depend on their separation.
Estimates of this type were first proved by Schmitz in~\cite{Schmitz}.
For completeness, we include a proof and discussion of the Lieb-Robinson bounds for cones in Appendix~\ref{sec:LRcones}. 

These results can be applied to prove stability of abelian quantum double models.
The details can be found in Section~\ref{sec:applications}, but the result is essentially as follows.
Consider a path $H_\Lambda(s)$ of local Hamiltonians, with $H_\Lambda(0)$ the Hamiltonian of the abelian quantum double model for some finite abelian group $G$.
If the ``perturbation'' along the path is sufficiently small, it can be proven that the Hamiltonians stay gapped~\cite{BravyiHM,MichalakisZ}.
If in addition the Hamiltonians are sufficiently local, and the path is piecewise $C^1$, it is possible to relate the ground state spaces along the path, at least those that are weak-$*$ limits of finite volume ground states.
We will denote these sets by $\mathcal{S}(s)$.
This can be done by adapting Hasting's quasi-adiabatic continuation technique (or spectral flow) to the infinite setting~\cite{BachmannMNS}.
The result is that there are automorphisms $\tau_s$ (with $s \in [0,1]$) such that for the weak-$*$ limits $\mathcal{S}(s) = \mathcal{S}(0) \circ \tau_s$.
Moreover, it can be shown that $\tau_s$ can seen as resulting from some time-dependent dynamics, and consequently the $\tau_s$ satisfy a corresponding Lieb-Robinson bound.
This puts us in a position to apply Theorem~\ref{thm:stabheuristic} to yield the following stability result:

\begin{thm}
	Let $H_\Lambda(s)$ be as described above, and write $\Delta(s)$ for the category of almost localized and transportable endomorphisms with respect to $\omega_s$, the infinite volume ground state of the local dynamics $\Lambda \mapsto H_\Lambda(s)$.
	Assume moreover that the object is $\Delta(0)$ satisfy an additional ``finite energy'' criterion.
	Then $\Delta(s)$ is braided tensor equivalent to $\operatorname{Rep}(\mathcal{D}(G))$, the representation category of the quantum double $\mathcal{D}(G)$, for all $s \in [0,1]$.
\end{thm}

The finite energy criterion says that for every irreducible endomorphism $\rho$, $\omega_0 \circ \rho$ is equivalent to one of the known states from~\cite{ChaNN}.
This criterion ensures that (for the unperturbed model) relaxing strict localization to approximate localization does not introduce new sectors.
We conjecture that this criterion is redundant (at least for the abelian quantum double model), but are not aware of a proof.

\section{Almost localized and transportable superselection sectors}\label{sec:superselect}

We discuss some properties of the almost localized endomorphisms of Definition~\ref{defn:alend}.
Notice that if $\rho$ is almost localized in $\Lambda_\alpha$ 
then by increasing the opening angle of the cone we get that
$ \rho$  is almost localized in $ \Lambda_{ \alpha + \delta}$ for all $0\leq \delta < \pi - \alpha $ with the same decay function.
On the other hand, if we decrease the opening angle then $\rho$ is no longer guaranteed to be almost localized in a cone with smaller opening angle.
In particular,  as $\epsilon \ra 0$ the function $f_\epsilon$ could diverge.

Perhaps the simplest example of an almost localized endomorphism is an exactly localized endomorphism, for which by definition $\rho(A) = A$ for $A \in \calA_{\Lambda^c}$ for some cone $\Lambda$.
Such an endomorphism could be obtained as follows.
Let  $U_n \in \calA_{ \Lambda_\alpha}$ be a sequence of unitary operators each supported on the cone $\Lambda_{ \alpha }$
and suppose the limit $\rho(A) \equiv \lim_{n\ra \infty} U_n^* A U_n$ exists for all $A \in \calA_{ \Gamma}$.
Then, $\rho$ is a $*$-endomorphism with $\rho(B) - B = 0$ for all $ B \in \calA_{ \Lambda_{\alpha}^c}$, by locality,
and it follows that $\rho$ is exactly localized in $\Lambda_\alpha$, and hence also almost localized in $\Lambda_{ \alpha }$.
It is well known that $\rho$ is an inner automorphism if and only if the sequence $U_n$ converges in $\calA_\Gamma$, but we are mainly interested in examples where this is \emph{not} the case, since by definition such automorphisms belong to the trivial (reference) sector.

In Definition~\ref{defn:alend} the class of decay functions $\mathcal{F}_\infty$ could in principle be weakened, 
however it is crucial that composition of endomorphisms preserve the almost localized property, as in the following lemma.
\begin{prop}\label{prop:aloccomp}
	Let $\rho$ and $\sigma$ be almost localized in $\Lambda_\alpha$ with decay functions $f$ and $ g$, respectively. 
	Then, $\rho \circ \sigma$ is almost localized in $\Lambda_\alpha$ with decay function $ f+ g$.
\end{prop}

\begin{proof}
	This follows from the triangle inequality, since
	\begin{align*}
	\sup_{A \in \calA_{\Lambda^c_{\alpha+\epsilon}+n}}\frac{\| \rho\circ \sigma (A) - A \| }{\|A\|} 
	& \leq \sup_{A \in \calA_{\Lambda^c_{\alpha+\epsilon}+n}} \frac{\|\rho\| \| ( \sigma(A) - A) \|}{\|A\|} + \frac{\| \rho(A) - A \|}{\|A\|}  \\
	&\leq f_\epsilon(n) + g_\epsilon(n),
	\end{align*} 
	where we use that $\|\rho\| = 1$.
\end{proof}

There are some simple but useful properties that follow readily from the definition of almost localized endomorphisms.
Let $\rho$ be almost localized in an infinite cone $\Lambda_\alpha$ with decay function $f$.
Then, for any fixed $n \in \NN$ 
\[  \sup_{A \in \calA_{\Lambda_{\alpha+\epsilon}^c+(n+n') } }\frac{\| \rho(A) - A\| }{\| A\| } \leq f_\epsilon(n+n') \qquad \text{ for all } \quad n' \in \ZZ_{\geq 0}. \]
Thus, $\rho$ is almost localized in $\Lambda_\alpha - n$ with the decay function $f_\epsilon(n+n')$.
If $f$ is submultiplicative in the second variable then $f_\epsilon(n + n') \leq f_\epsilon(n) f_\epsilon(n')$,
 thus recovering the same decay function up to a factor.

Consider translations of the cone $\Lambda_\alpha+x$
and the corresponding translated endomorphism $\rho_x \equiv T_{x} \circ \rho \circ T_{-x}$, with $x \in \ZZ^\nu$.
Then,
\begin{align*}
\sup_{A \in \calA_{\Lambda_{\alpha+\epsilon}^c -(x-n)}}\frac{\| \rho_x(A) - A\| }{\| A\| } 
& =  \sup_{A \in \calA_{\Lambda_{\alpha+\epsilon}^c+n } }\frac{\|\rho_x(T_{x}(A)) - T_{x}(A)\| }{\| T_{x}(A)\| } \\
& = \sup_{A \in \calA_{\Lambda_{\alpha+\epsilon}^c+n } }\frac{\|T_{x}( \rho(A) - A)\| }{\| T_{x}(A)\| } \\
& \leq f_\epsilon(n)
\end{align*}
Thus, the shifted endomorphism $\rho_x$ is almost localized in $\Lambda_\alpha+{x}$ with the same decay function.
In particular, for any sequence $x_n$ such that $d( \Lambda_{ \alpha  } + x_n, 0) \ra \infty$ as $n\ra \infty$ then
\begin{equation}\label{eqn:shifttoid}
\lim_{\substack{ n \ra \infty \\ d(\Lambda_{ \alpha  }+x_n,0)\ra \infty}} \| \rho_{x_n}(A) - A \| = 0  \quad \text{ for all } \quad A \in \calA_\Gamma.
\end{equation}

Notice that almost localized property of $\rho$ is defined on the level of the algebra of observables
and does not depend on the reference state $\omega_0$.
In contrast, the transportability property (see Definition~\ref{defn:transportable}) has a clear dependence on the reference state.
In quantum spin models, like the toric code model, transportability can typically be proven from the properties of the ground state such as
path independence of string-like operators and translation invariance.

Generally, any cyclic representation of the quasi-local algebra $\calA_\Gamma$ can be obtained by 
composing $\pi_0$ with an asymptotically inner endomorphism~\cite{Kishimoto}.
For localized and transportable endomorphisms, the transportability property leads to natural choices for the sequence of unitaries.
Here we give an explicit construction for 
cyclic representations of the form $\pi_0 \circ \rho$ where $\rho$  is almost localized and transportable.
Similar constructions have been considered in models for the electromagnetic charge \cite{BuchholzEM}.

\begin{lemma}\label{lem:asympinner}
	Let $\rho$ be an almost localized and transportable $*$-endomorphism on $\calA_\Gamma$.
	Then $\rho$ is asymptotically inner in $\calB(\calH_0)$. That is,
	there exists a sequence of unitary intertwiners $U_n \in \calB(\calH_0)$ such that 
	\begin{equation}
	\rho(A) = \lim_{n \ra \infty } U_n^* A U_n \quad \text{ for all } \quad A \in \calA_\Gamma.
	\end{equation}
	Note that we identify $\calA_\Gamma$ with $\pi_0(\calA_\Gamma)$ again.
\end{lemma}

\begin{proof}
	Let $\Lambda_{ \alpha  }$ be a cone such that $\rho$ is almost localized in $\Lambda_{ \alpha  }$ with decay function $f$.
	By transportability, there exists a sequence of $*$-endomorphisms $ \rho_n$ almost localized in $\Lambda_{ \alpha  } + n$ with decay function $f$ and 	unitary operators $U_n \in (\rho, \rho_n)$ such that 
	\begin{equation}
	\rho_n(A) =  U_n \rho(A) U_n^* \quad \text{ for all } \quad A \in \calA_\Gamma.
	\end{equation}
	Let $\epsilon >0$ be given and $A \in \calA_{loc}$.
	Then, there exists $N>0$ such that $ A \in \calA_{\Lambda_{\alpha +\epsilon}^c -N}$. 
	It follows that 
	\begin{align*}
	\| \rho(A) - U_n^* A U_n\| &=  \| U_n^* ( \rho_n(A) - A ) U_n \|  = \| \rho_n(A) - A \| \\
	& \leq f_\epsilon(n - N) \| A\| \quad \ra 0 \quad \mbox{ as } \quad n \ra \infty.
	\end{align*}
	In the last inequality, transportability implies that each $ \rho_n$ can be assigned the same decay function $f$.
	Since $\calA_{loc}$ is dense in $\calA_\Gamma$, we have $\rho(A) = \lim_{n \ra \infty } U_n^* A U_n$ for all $A \in \calA_\Gamma$.
\end{proof}

\subsection{Locality structure for intertwiners}
\label{sec:intertwiners}
Let $\rho$ and $ \sigma$ be $*$-endomorphisms.
Recall that the space of interwiners is given by 
\begin{equation*}
(\rho, \sigma)_{\pi_0}\equiv \{ T \in \calB(\calH_0) : T \pi_0( \rho (A) )= \pi_0(\sigma(A)) T, A \in \calA_\Gamma\},
\end{equation*}
where we use the subscript $\pi_0$ if we want to emphasize the dependence on the reference representation.
When the context is clear, we drop the subscript $\pi_0$ and write $(\rho, \sigma) = (\rho,\sigma)_{\pi_0}$, as we have done before.
Notice that if $R \in (\rho,\sigma)$ and $S \in (\sigma, \tau)$ then
\begin{equation*}
SR \rho(A) = S \sigma(A) R = \tau (A) SR,
\end{equation*} 
so that $ SR \in (\rho, \tau)$ intertwines $\rho$ and $\tau$.
Hence the set of $*$-endomorphisms of $\calA_\Gamma$ with intertwiners has the structure of a category.
We analyze the locality structure for intertwiners between almost localized endomorphisms.

Since the intertwiners generally do not reside in the quasi-local algebra, 
the strong approximate split property for $\omega_0$ will be crucial to establish a suitable $C^*$-algebra for the set of intertwiners.
More precisely, it will allow us to decompose the Hilbert space of the reference representation in a natural way into tensor products.
We can then approximate the intertwiner by an observable that acts on only one of the tensor factors~\cite{NachtergaeleSW}.
Using the strong approximate split property, this observable can be seen to be localized in one of the cone algebras $\mathcal{R}(\Lambda)$.
It should be noted that in general writing $\calA_\Gamma \cong \calA_\Lambda \otimes \calA_\Lambda^c$ does \emph{not} lead to a tensor product decomposition of the corresponding von Neumann algebras (cf.~\cite[Sect. 5]{Naaijkens11}).

Unlike in the situation where we have strictly localized endomorphisms and Haag duality, intertwiners can no longer be strictly localized either, 
and it is necessary to look at \emph{almost localized} operators.
In the context of local quantum physics, almost local observables were first introduced in \cite{ArakiH} and further studied 
in the context of quantum spin systems in \cite{BachmannDN,Schmitz}.
Here, we study almost localized observables of $ \calB(\calH_0)$, 
which by irreducibility coincides with the weak operator closure $\overline{\pi_0(\calA_\Gamma)}^w$,
with respect to a cone region.
The underlying idea is that if $R$ is strictly localized in a cone $\Lambda$, it commutes with all observables localized outside the cone.
The definition here is an approximate version of that.

\begin{defn}\label{def:alocoperator}
	An operator $A \in \calB(\calH_0)$ is said to be \emph{almost localized} in a cone $\Lambda_{\alpha} \in \mathcal{C}$ if 
	there exists a function $f\in \mathcal{F}_\infty$ such that 
	\begin{equation}\label{eqn:alocobs}
	\sup_{B\in \mathcal{R}(\Lambda^c_{\alpha+\epsilon}+n)} \frac{ \| AB-BA \| }{\| B \| } \leq f_\epsilon(n) = O(n^{-\infty})
	\quad \mbox{ for all } \quad 0 < \epsilon < \pi - \alpha. 
	\end{equation}
	The function $f$ is called the \emph{decay function} for $A$ in $\Lambda_{\alpha}$.
\end{defn}
Hence, just as in applications of Lieb-Robinson bounds, locality is expressed by the property that commutators with operators outside 
of the localization region are small.
Note that the supremum is over a von Neumann algebra.
The following lemma shows that instead it can be taken over the $C^*$-subalgebra of observables $\calA_{\Lambda^c_{\alpha+\epsilon}+n} \subset  \mathcal{R}(\Lambda^c_{\alpha+\epsilon}+n)$.
\begin{lemma}\label{lem:alocoperator}
	An operator $A \in \calB(\calH_0)$ is almost localized in a cone $\Lambda_{\alpha}$ if and only if
	there exists a  function $f\in \mathcal{F}_\infty$ such that 
	\begin{equation*}
	\sup_{B\in \calA_{\Lambda^c_{\alpha+\epsilon}+n}} \frac{ \| AB-BA \| }{\| B \| } \leq f_\epsilon(n) = O(n^{-\infty})
	\quad \mbox{ for all } \quad 0 < \epsilon < \pi - \alpha.
	\end{equation*}
\end{lemma}

\begin{proof}
	The forward direction of the lemma  follows directly from the fact $\calA_{\Lambda^c_{\alpha+\epsilon}+n} \subset \mathcal{R}(\Lambda^c_{\alpha+\epsilon}+n)$.
	
	Suppose there exists a function $f\in \mathcal{F}_\infty$ such that 
	\begin{equation*}
	\sup_{B\in \calA_{\Lambda^c_{\alpha+\epsilon}+n}} \frac{ \| AB-BA \| }{\| B \| } \leq f_\epsilon(n) = O(n^{-\infty})
	\quad \mbox{ for all } \quad 0 < \epsilon < \pi - \alpha.
	\end{equation*}
	Let $ B \in \mathcal{R}(\Lambda^c_{\alpha+\epsilon}+n)$ be such that $\|B\| =1$.
	By Kaplansky's density theorem the unit ball of $\calA_{\Lambda^c_{\alpha+\epsilon}+n}$ is dense (in the strong operator topology) in the unit ball of $\mathcal{R}(\Lambda^c_{\alpha+\epsilon}+n)$.
	Hence there is a net $B_\lambda \in \calA_{\Lambda^c_{\alpha+\varepsilon}+n}$ such that $\|B_\lambda\| \leq 1$ and $ \slim B_\lambda = B$.
	By strong operator continuity of multiplication on bounded sets, we have that $AB - BA = \slim_\lambda A B_\lambda - B_\lambda A$.
	Let $\xi \in \calH_0$.
	We have that
	\[
		\| [A,B] \xi \| = \lim_\lambda \| A B_\lambda - B_\lambda A\ \xi \| \leq \| \xi \| \lim_\lambda \| A B_\lambda - B_\lambda A \| \leq \| \xi \| f_\epsilon(n) \lim_\lambda \|B_\lambda\|.
	\]
	It follows that $\| AB - BA \| \leq \|B\| f_\epsilon(n)$, from which the result follows.
\end{proof}

As a consequence, we get the following locality property for intertwiners. 
\begin{cor}\label{cor:intertwineraloc}
	Let $\rho$ and $\sigma$ be almost localized in a cone $\Lambda_\alpha$
	with decay functions $f$ and $g$, respectively.
	If $R \in (\rho, \sigma)$ then $R$ is almost localized in $\Lambda_{\alpha}$ with decay function $\|R\|(f + g)$.
\end{cor}

\begin{proof}
	For all $B \in \calA_\Gamma$ we have that
	\begin{align*}
	\| R B - B R \| 
	& \leq  \| R B- R\rho(B)  \| + \| R\rho(B) - \sigma(B)R\| + \|  \sigma(B) R - B R\| \\
	&\leq  \| R\| ( \| \rho(B) - B \| + \|\sigma(B) -B \| ).
	\end{align*}
	Therefore,
	\begin{equation*} 
	\sup_{B\in \calA_{\Lambda^c_{\alpha+\epsilon}+n}} \frac{ \| RB-BR\| }{\| B \| }  \leq \|R\|( f_\epsilon(n) + g_\epsilon(n)),
	\end{equation*}
	from which the claim follows using the lemma.
\end{proof}
 
The almost local property for operators $A \in \calB(\calH_0)$ suggest that 
if $\omega_0$ satisfies certain locality conditions then $A$ may be well approximated in some cone algebra $\mathcal{R}(\Lambda)$.
Following~\cite{BuchholzF,Naaijkens11}, we introduce a convex cone $K_\kappa$ with an arbitrary but small opening angle  $\kappa<\pi$ 
determining a  forbidden direction.
For convenience, we abuse notation to simply write
$K \equiv K_\kappa$ and $K_{\epsilon} = K_{\kappa+\epsilon}$.
Define the \emph{auxiliary algebra} as 
\begin{equation}\label{eqn:auxalg}
\calB_{K}\equiv  \overline{ \bigcup_{x\in \ZZ^\nu} \mathcal{R}((K+x)^c)}^{\| \cdot\|} =  \overline{ \bigcup_{n \in \NN } \mathcal{R}((K+n)^c)}^{\| \cdot\|},
\end{equation}
where the second equality follows from the fact that for every $ x\in \ZZ^\nu$ there is an $n \in \NN$ such that $ K+n \subset K+x$
and $ (K+x)^c \subset (K+n)^c$.

Up to this point we have not used any of the assumptions on the reference state $\omega_0$.
Now we will assume $\omega_0$ satisfies the strong approximate split property
and use it to establish operator norm estimates for the set of intertwiners.
These estimates are key in extending the almost localized endomorphisms to intertwiners.
\begin{lemma}\label{lem:alocauxalg}
	Suppose $\omega_0$ satisfies the strong approximate split property.
	Let $\Lambda_\alpha$ be a cone such that for some $\epsilon>0$ and $x\in \ZZ^\nu$ it holds that  $\Lambda_{\alpha+\epsilon} \ll (K+x)^c$.
	If $A\in \calB(\calH_0)$ is almost localized in $\Lambda_\alpha$  then $A \in \calB_{K}$.
\end{lemma}

\begin{proof}
	Notice that for all $n \in \NN$ we have that $ K+x+n \subset \Lambda_{ \alpha  +\epsilon}^c +n $.
	
	Let $\delta >0$ be given and  $A\in \calB(\calH_0)$ be almost localized in $\Lambda_{ \alpha  }$. Then, there exists an $N>0$ such that if $n>N$ then 
	\begin{equation}\label{eqn:alo}
	\|A B - B A\| \leq f_\epsilon(n) \| B\|  < \delta \|A\|\|B\| \quad \text{for all} \quad B \in \mathcal{R}(K+x+n).
	\end{equation}
	Let $n' \in \NN$ be such that $ \Lambda_{\alpha+\epsilon} \ll (K+x+n)^c \ll (K+x+n')^c$.
	By the strong approximate split property for $\omega_0$, if $n'$ is chosen large enough then there is a type I factor $\mathcal{N}$ such that 
	\begin{align}
	\mathcal{R}((K+x+n)^c) \ &\subset \  \mathcal{N} \  \subset \  \mathcal{R}((K+x+n')^c)\\
	\mathcal{R}(K+x+n')\  &\subset \ \mathcal{N}'  \subset  \ \mathcal{R}(K+x+n). \label{eqn:split2}
	\end{align}
	From \eqref{eqn:alo} and \eqref{eqn:split2} it follows that  
	\begin{equation*}
	\| A B - B A \| \leq \delta \|A\| \|B\| \quad \text{for all} \quad B \in \mathcal{N}'.
	\end{equation*}
	Applying Lemma 2.1 of~\cite{NachSW} with $\calB(\calH_1) \cong \mathcal{N}$ and $\calB(\calH_2)\cong \mathcal{N}'$,
	there exists an operator $A' \in \mathcal{N}$ such that 
	\begin{equation*}
	\| A - A' \| \leq \delta \| A\|.
	\end{equation*}
	In other words, $A$ is arbitrarily well approximated  in norm by an operator $A'$ in the cone algebra $\mathcal{R}((K+y)^c)$ for some $y \in \ZZ^\nu$.
	Therefore, $A \in \calB_K$.
\end{proof}

In the setting above, we do not restrict the angle of the cone $\Lambda_{ \alpha  }$ but simply require that it be epsilon bounded away from the cone $(K+x)^c$.
In what remains, we will consider endomorphisms almost localized on convex cones.  
It is clear that if $\Lambda_{ \alpha  }$ is a  convex cone then we could weaken the assumptions in Lemma~\ref{lem:alocauxalg} to 
simply the assumption that $\Lambda_{ \alpha  } \ll (K+x)^c$ for some $x\in\ZZ^\nu$.

Consider now two endomorphisms $\rho$ and $\sigma$. 
If $\rho$ and $\sigma$ are almost localized in the same convex cone $\Lambda_{ \alpha  }$ with $ \Lambda_\alpha \subset (K+x)^c$ for some $x$,
then Corollary~\ref{cor:intertwineraloc} and Lemma~\ref{lem:alocauxalg} immediately yield that $(\rho,\sigma) \subset \calB_K$.
However, if $\rho$ and $\sigma$ are almost localized on potentially different convex cones $\Lambda_{ \alpha  }^\rho$ and $\Lambda_{ \beta}^\sigma$, respectively,
the argument as above must be slightly modified.

\begin{cor}\label{cor:interwinerauxalg}
	Suppose $\omega_0$ satisfies the strong approximate split property.
	Let $\rho$ and $\sigma$ be almost localized endomorphisms in convex cones $\Lambda_\alpha^\rho$ and $\Lambda_{ \beta }^\sigma$, respectively.
	If there is a convex cone $K$ such that  $\Lambda_{\alpha}^\rho \ll (K+n^\rho)^c$ and $ \Lambda_{ \beta}^\sigma \ll (K+n^\sigma)^c$
	for some $ n^\rho, n^\sigma \in \ZZ$ then $(\rho,\sigma) \subset \calB_K$.
\end{cor}

\begin{proof}
	By convexity, there exist $\epsilon>0$ such that 
	$\Lambda_{\alpha+\epsilon}^\rho \ll (K+n)^c$ and $ \Lambda_{\beta+\epsilon}^\sigma \ll  (K+n)^c$
	where $n = \max\{n^\rho,n^\sigma\}  $.
	In other words, $\Lambda_{ \alpha  }^\rho \cup \Lambda_{ \beta}^\sigma \subset ( K_{\epsilon}+n )^c \subset (K+n+1)^c$.
	It follows that $\rho$ and $\sigma$ are almost localized endomorphisms on the cone $(  K_{\epsilon}+n)^c$.
	By Corollary~\ref{cor:intertwineraloc}, $R\in(\rho,\sigma)$ is almost localized in $(  K_{\epsilon}+n)^c$
	and $( K_{\epsilon}+n)^c \subset ( K_{\epsilon/2}+n)^c \ll (K+n')^c$ for some $n'\in\NN$.
	Thus, the result is obtained by applying Lemma~\ref{lem:alocauxalg}.
\end{proof}

It is possible to give a more precise characterization of the localization properties of $R \in (\rho,\sigma)$, but we will not need this.

\subsection{Superselection structure}
We assume throughout this section that $\omega_0$ satisfies the strong approximate split property.

Let $\rho, \rho', \sigma, \sigma' \in \Delta$ and $R\in(\rho, \rho')$, $R'\in(\sigma,\sigma')$.
Sequences $\{U_n\}$ and $\{V_n\}$ of unitary operators implementing the asymptotically inner property for $\rho$ and $\rho'$, respectively, as in Lemma~\ref{lem:asympinner} are not unique.
Motivated by the following formal calculation, that should hold independent of the choice of sequences, 
\begin{align*}
\| R \rho(R')  -  \sigma(R') R  \| & = \lim_{m,n \ra \infty} \| R U_m^* R' U_m - V_n^* R' V_n R \| = \lim_{m,n \ra \infty}  \|V_n R U_m^* R' -  R' V_n RU_m^* \|,
\end{align*}
the notions of \emph{asymptopia} and \emph{bi-asymptopia} were introduced in \cite{BuchholzAA}.
We will recall their definitions later.
The calculation above is formal in the sense that $\rho(R')$ and $\sigma(R')$ may not be well defined. 

We consider a mapping $\mathcal{U}: \Delta\ni\rho \mapsto \mathcal{U}_\rho$ where each $\mathcal{U}_\rho$ is a family of unitary sequences
implementing the asymptotically inner property for $\rho$.

\begin{defn}[\cite{BuchholzAA}]\label{defn:asymptopia}
	An \emph{asymptopia} for $\Delta$ is a mapping $\mathcal{K} : \rho \mapsto \mathcal{K}_\rho$, with $\rho \in \Delta$, 
	where each $\mathcal{K}_\rho$ is a \emph{stable} family of sequences of unitary operators in $\calB(\calH_0)$, that is, 
	closed under taking subsequences,
	such that for each $\{U_m\} \in \mathcal{K}_\rho$ we have
	\begin{equation*}
	\rho(A) = \lim_{m \ra \infty} U_m^* A U_m \quad \mbox{ for all } \quad A \in \calA_\Gamma,
	\end{equation*}
	and for any $\rho, \rho', \sigma,\sigma' \in \Delta$ given $R \in (\rho, \rho')$ and $R' \in (\sigma,\sigma')$, and $ \{U_m\} \in \mathcal{K}_\rho$ and $ \{V_n\} \in \mathcal{K}_{\rho'}$,
	\begin{equation*}
	\lim_{m,n\ra \infty} \| [ V_n R U_m^*, R'] \| = 0.
	\end{equation*}
\end{defn}
We will later need to consider different asymptopia simultaneously, and will generally use the notation $\mathcal{K}, \mathcal{U}$ or $\mathcal{V}$.
The last condition will allow us to make sense of expressions like $\rho(T)$, for $T$ an intertwiner.
This is essential in defining the tensor product (and hence, fusion rules) on $\Delta$.
More precisely, consider $S \in (\rho,\rho')$ and $T \in (\sigma,\sigma')$.
We will define extensions $\rho_\mathcal{K}$ (and similarly for the others) to an algebra that contains $S$ and $T$.
Then, we can define
\begin{equation}
	\label{eq:tensorprod}
	(\rho \otimes \sigma)(A) \equiv \rho_{\mathcal{K}} \circ \sigma(A), \quad \quad S \otimes T \equiv S \rho_{\mathcal{K}}(T).
\end{equation}
A formal calculation shows that indeed $S \otimes T \in (\rho \otimes \sigma, \rho' \otimes \sigma')$.

In our case, $\Delta$ will be a family of the almost localized and transportable endomorphisms.
The first goal is to show that these naturally lead to a choice of asymptotia.
Moreover, they allow us to extend the localized endomorphisms of $\calA_\Gamma$ to an algebra containing also the intertwiners.
To this end, let $\mathcal{A}_\Delta$ be the $C^*$-subalgebra of $\calB(\calH_0)$ generated by $\calA_\Gamma$ and $\mathcal{T}_\Delta$.
We then have the following result.
\begin{thm}\label{thm:Asymptopia}
	Let $\Delta$ be a semi-group of endomorphisms of $\calA_\Gamma$ satisfying the almost localized and transportable superselection criterion of Definition~\ref{def:superselection},
	and suppose that each $\rho \in \Delta$ is almost localized in a convex cone $\Lambda_{ \alpha}^\rho$.
	Suppose	there exists a convex cone $K \in \mathcal{C}$ such that 
	for all $\rho \in \Delta$ we have 
	$\Lambda_{ \alpha }^\rho \ll (K+n^\rho)^c$ for some $n^\rho \in \NN$. 
	Then, there exists an asymptopia $\mathcal{K}: \rho \mapsto \mathcal{K}_\rho$ for $\Delta$.
	Furthermore, each $\rho \in \Delta$ has a unique extension to a $*$-endomorphism $\rho_{\mathcal{K}}$ of $\calA_{ \Delta}$ such that
	for each $\{ U_m\} \in \mathcal{K}_\rho$ we have 
	\begin{equation*}
	\rho_{\mathcal{K}}(A) = \lim_{m \ra \infty} U_m^* A U_m \quad \mbox{ for all } \quad A \in \calA_\Delta,
	\end{equation*}
	and $(\rho, \sigma) = (\rho_{\mathcal{K}}, \sigma_{\mathcal{K}} )$ for all $\rho, \sigma \in \Delta$.
	Moreover, $\Delta$ is a tensor $C^*$-category, where the tensor operation is defined as in equation~\eqref{eq:tensorprod}.
\end{thm}

\begin{proof}
	We construct an asymptopia $\mathcal{K}$ for $\Delta$ based on the forbidden region $K$.
	Let $\rho \in \Delta$ be given.
	Define $\mathcal{K}_\rho$ as the set of sequences of unitaries $\{ U_m\}$ with the following properties:
	\begin{enumerate}
		\item $\rho(A) = \lim_{m \ra \infty } U_m^* A U_m$ for all $A \in \calA$;
		\item there is an $\epsilon(\{U_m\})>0$ (i.e., only depending on the sequence), an increasing sequence $k_m \in \NN$ and a sequence of $*$-endomorphisms $\rho_m$ 
		such that $\rho_m$ is almost localized in $K_{-\epsilon(\{U_m\})} + k_m$ with decay function $f^\rho$ independent of $m$ (but which may depend on the sequence);
		\item $U_m$ intertwines $\rho$ and $\rho_{m}$, that is, $U_m \in (\rho, \rho_m)$.
	\end{enumerate} 
We first show that $\mathcal{K}_\rho$ is non-empty, using an argument similar to the proof of Lemma~\ref{lem:asympinner}.
	Choose some $0 < \epsilon < \kappa$.
	Then by the transportability assumption, given an increasing sequence $k_m$, there is a sequence of endomorphisms $\rho_m$ almost localized in $K_{-\epsilon} + k_m$, with some decay function $f^\rho$ which can be chosen independently of $m$, since the opening angle of the cones is the same.
	For each $m$ we can then choose a unitary $U_m \in (\rho, \rho_m)$, and we set $\epsilon(\{U_m\}) \equiv \epsilon$.
	An argument similar to that in the proof of Lemma~\ref{lem:asympinner} shows that the sequence $\rho_m$ will converge pointwise to the identity,
	\begin{equation}\label{eqn:pointwiseidentity}
	\lim_{ m\ra \infty} \| \rho_m(A) - A \| = \lim_{m \ra \infty } \| \rho(A) -U_m^* A U_m \| = 0,
	\end{equation}
	and hence $\mathcal{K}_\rho$ is non-empty.
	By construction, the family $\mathcal{K}_\rho$ is stable, that is, any subsequence of $\{ U_m\}$ will also belong to $\mathcal{K}_\rho$.
	
	Let  $\rho, \rho', \sigma, \sigma' \in \Delta$,
    $R \in (\rho, \rho')$ and $S \in (\sigma, \sigma')$.
    Consider $\{ U_m\} \in \mathcal{K}_\rho$ and $ \{V_{m'}\} \in \mathcal{K}_{\rho'}$.
	It follows that  
	\begin{equation*}
	V_{m'}R U_m^* \rho_m(A) = V_{m'} R \rho(A) U_m^*= V_{m'} \rho'(A) R U_m^* = \rho'_{m'}(A) V_{m'}RU_m^*.
	\end{equation*}
		so that $V_{m'}R U_m^* \in ( \rho_m , \rho'_{m'})$ is an intertwiner.
		By Corollary~\ref{cor:intertwineraloc},  $V_{m'}R U_m^* $ is an almost localized operator in $K_{-\epsilon} + \min\{ k_m,k'_{m'}\}$ with decay function $g \equiv \|R\|(f^\rho + f^{\rho'})$ for $\epsilon = \min\{\epsilon(\{U_n\}),\epsilon(\{V_{m'}\})\}$.
		In particular, the following inequality holds whenever $\min\{k_m,k'_{m'}\}> n > 0$,		 
	\[ \sup_{A \in \mathcal{R}((K+n)^c)}  \frac{\|[V_{m'}R U_m^*,A]\|}{\|A\|}  \leq g_\epsilon( \min\{k_m,k'_{m'}\} - n). \]

	Let $\epsilon'>0$ be given.
	By assumption, if $n > \max\{n^\sigma, n^{\sigma'}\}$ then $ \Lambda_{ \alpha  }^\sigma, \Lambda_{ \beta}^{\sigma'}  \subset (K+n)^c$.
	Thus, Corollary~\ref{cor:interwinerauxalg} implies that  $S  \in \mathcal{B}_K$. 
	Therefore,  there exists $N>0$ such that for $n\geq N$  there is an operator $S_n\in \mathcal{R}((K + n)^c)$ with $\|S_n\| \leq \|S\|$ and the approximation
	$\| S - S_n \| < \epsilon' \| S\|$.
	It follows that 
	\begin{align*}
		\| [ V_{m'} R U_m^*, S] \| & \leq   \| [ V_{m'} R U_m^*, S_n] \| + 2 \| R\| \| S- S_n\|\\
			&\leq  \| S\|  g_{\epsilon}( \min\{k_m,k_{m'}'\} - n) + 2 \epsilon' \| R \| \|S\|.
	\end{align*}
	Taking the limit as $m,m' \ra \infty$ and since $\epsilon' >0$ was arbitrary we have that 
	\begin{equation*}
	\lim_{m,m' \ra \infty}  \| [ V_{m'} R U_m^*, S] \| = 0.
	\end{equation*}
	Therefore, $\mathcal{K}$ is an asymptopia for $\Delta.$
	Applying Theorem 4 of \cite{BuchholzAA} we recover that $\rho$ extends uniquely to a $*$-endomorphism $\rho_{\mathcal{K}}$ on  $\calA_{\Delta}$ 
	and satisfies the properties as listed in the theorem.
	
	The tensor product on $\Delta$ was defined in equation~\eqref{eq:tensorprod}.
	A calculation shows that for all $ A \in \calA_\Gamma$,
	\begin{align*}
	(R \otimes S )\rho \circ \sigma (A) & = R \rho_{\mathcal{K}}(S) \rho ( \sigma(A))  
	= R \rho_{\mathcal{K}}(S \sigma(A) ) 
	= \rho'(\sigma' (A)) R \rho_{\mathcal{K}}(S) \\
	&= \rho' \circ \sigma' (A) (R \otimes S),
	\end{align*}
	and thus $R \otimes S \in ( \rho\circ \sigma, \rho' \circ \sigma')$. It follows that $\Delta$ is a tensor $C^*$-category.
\end{proof}

If every object $\rho\in \Delta$ can be almost localized in a fixed convex cone $\Lambda_\alpha$ 
then any convex cone $K$ satisfying $ K \subset( \Lambda_{ \alpha  } + 1)^c $
will satisfy the assumptions of Theorem~\ref{thm:Asymptopia}.

For a fixed reference state $\omega_0$ the asymptopia $\mathcal{K}$, defined in Theorem~\ref{thm:Asymptopia},
 are determined by a forbidden region $K$.
In principle, a choice of different cones could lead to different asymptopia, and in particular, a different tensor structure.
However, for the tensor product to be related to the physical procedure of fusing charges
it is important that the tensor structure on $\Delta$ is independent of the forbidden region.
This is the case.
To see this, consider a \emph{path} of cones.
More precisely, choose cones $K^i \in \mathcal{C}$, with $i = 1, \dots, n$, such that for each $i = 1, \dots n-1$ we have either $K^i \subset K^{i+1}$ or $K^{i+1} \subset K^i$.
This allows for an interpolation argument (compare for example with~\cite[Sect. 4]{DHR1}).

\begin{lemma}\label{lem:indK}
Let $\Delta$ be a $C^*$-category satisfying the almost localized and transportable superselection criterion~\ref{def:superselection}
and suppose that each $\rho \in \Delta$ is almost localized in a convex cone $\Lambda_{ \alpha}^\rho$.
Suppose that there exists a path of cones $K^1, \dots,  K^n \in \mathcal{C}$ such that for all $\rho \in \Delta$ 
we have $\Lambda_{ \alpha}^\rho \ll (K^i+n_i^\rho)^c$ for some $n_i^\rho\in\NN$ and $ i = 1, \dots, n$.
Then, the asymptopias $\mathcal{K}_1$ and $ \mathcal{K}_n$ (as defined in Theorem~\ref{thm:Asymptopia}) 
determine the same tensor structure on $\Delta$.
\end{lemma}

\begin{proof}
	It is enough to show that if $\widetilde{K}_1 \subset \widetilde{K}_2$ (with $\Lambda_\alpha \ll (\widetilde{K}_2 + n^\rho)^c$ for some $n^\rho$), then the corresponding asymptopias $\widetilde{\mathcal{K}}_i$ from Theorem~\ref{thm:Asymptopia} yield the same tensor structure.
	For this it is enough to show that the $\widetilde{\mathcal{K}}_i$ are contained in a common asymptopia~\cite{BuchholzAA}.
	But since $\widetilde{K}_1 \subset \widetilde{K}_2$, from the construction in the theorem it is clear that every sequence in $\widetilde{\mathcal{K}}_1$ is also in $\widetilde{\mathcal{K}}_2$.
\end{proof}

Given a tensor $C^*$-structure on $\Delta$,
a braiding is a unitary intertwiner $\epsilon_{\rho,\sigma} \in (\rho\otimes \sigma, \sigma\otimes \rho)$ satisfying some natural additional conditions.
Physically, a braiding describes an interchange of the charges labeled by $\rho$ and $\sigma$.
Intuitively, if $\rho$ and $\sigma$ are almost localized in $\Lambda_{ \alpha  }$ 
a braiding could be performed by first transporting 
$\rho$ to a far away cone $\Lambda_{ \alpha  }^\mathcal{U}$ 
and $\sigma$ in the opposite direction to a cone $\Lambda_{ \beta}^\mathcal{V}$, see Figure~\ref{fig:cones}.
By the almost localized property, after this transportation the charges would approximately commute, 
in the sense that for the transported endomorphisms we have $\rho' \circ \sigma'$ is close to $\sigma' \circ \rho'$.
Transporting the charges back would complete the braid.
In the following we make this argument precise.

\begin{defn}[\cite{BuchholzAA}]\label{defn:biasymptopia}
	Suppose  $\Delta$ is a tensor $C^*$-category of endomorphisms on $\calA_\Gamma$.
	A \emph{bi-asymptopia} for $\Delta$ is a pair of mappings  
	$\mathcal{U} : \rho \mapsto \mathcal{U}_\rho$ and $ \mathcal{V}: \rho \mapsto \mathcal{V}_\rho$ 
	where $\mathcal{U}_\rho$ and $\mathcal{V}_\rho$ are stable families of unitary sequences
	such that for the following hold:
	\begin{enumerate}
		\item for each $ \{ U_m\} \in \mathcal{U}_\rho$ and $\{ V_n\} \in \mathcal{V}_\rho$ we have 
		\begin{equation*}
		\rho(A) = \lim_{m \ra \infty} U_m^* A U_m = \lim_{n \ra\infty} V_n^* A V_n \quad \mbox{ for all } \quad A \in \calA_\Gamma,
		\end{equation*}
		\item for $ \rho, \rho' \in \Delta$ there exist $\{U_m\} \in \mathcal{U}_{\rho}$, $ \{U'_m\} \in \mathcal{U}_{\rho'}$, $\{V_n\} \in \mathcal{V}_{\rho}$, and $ \{V'_n\} \in \mathcal{V}_{\rho'}$ such that $\{U_m \otimes U'_m\} \in \mathcal{U}_{\rho \circ \rho'}$ and $ \{V_n \otimes V'_n\} \in \mathcal{V}_{\rho\circ \rho'}$,
		\item  and for  $R \in (\rho,\rho')$ and $ S \in ( \sigma, \sigma')$ we have that
		\begin{equation*}
		\lim_{m,m',n,n' \ra \infty} \| U'_{m'} R U_{m}^* \otimes V'_{n'} S V_{n}^* - V'_{n'} S V_{n}^* \otimes U'_{m'} R U_{m}^* \| =  0
		\end{equation*}
		for all $U_m \in \mathcal{U}_\rho$, $U_{m'}' \in \mathcal{U}_{\rho'}$, $ V_m \in \mathcal{V}_\sigma$, and $ V_{m'}' \in \mathcal{V}_{\sigma'}$.
		\end{enumerate}
		The tensor product is that of Theorem~\ref{thm:Asymptopia}.
	\end{defn}

\begin{figure}
	\begin{center}
	\includegraphics[scale=.5]{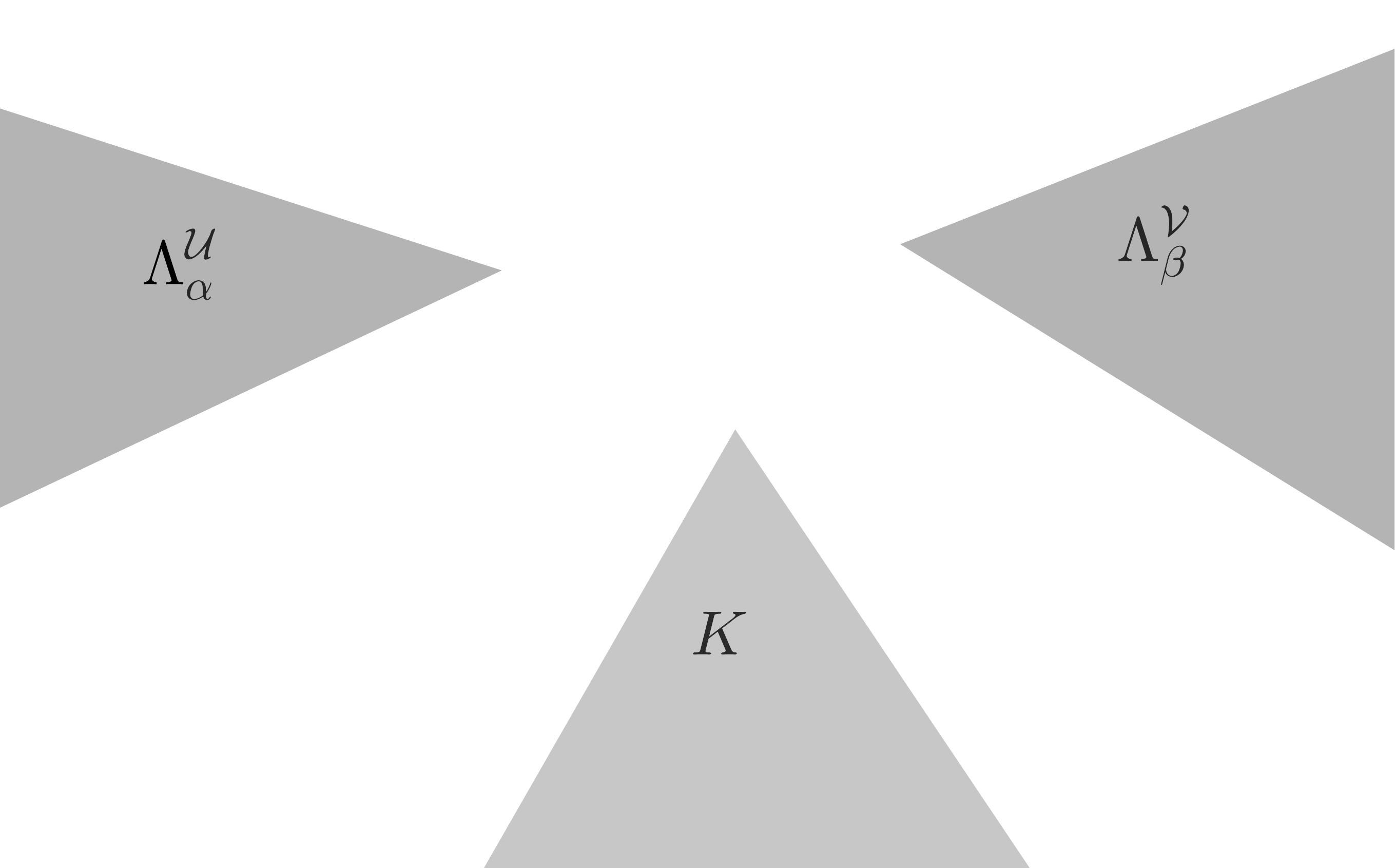}
	\end{center}
	\caption{A sample configuration of cones $\Lambda_{ \alpha}^{\mathcal{U}}$, $\Lambda_\beta^{\mathcal{V}}$ and $K$ satisfying the conditions set for the bi-asymptopia $\{\mathcal{U}, \mathcal{V}\}$ of $\Delta$. Here $\Lambda_{\alpha}^{\mathcal{U}}$ is to the left of $\Lambda_\beta^{\mathcal{V}}$ with respect to $K$.}
	\label{fig:cones}
\end{figure}

The bi-asymptopia can be used to define a braiding, as we will do below.
Before we continue it should be noted that bi-asymptopia generally are not unique, and consequently different choices may lead to different braidings.
One can order bi-asymptopia by inclusion, and then it can be seen that two bi-asymptopia yield the same braiding if they lie in the same path component.
It should be noted that this phenomenon already appears in the case of strict cone localization in $\nu = 2$.
There one has to make a choice if for $\epsilon_{\rho,\sigma}$ one has to move $\rho$ to the ``left'' and back, or to the ``right'' and back.
The notion of left and right can be defined in terms of an auxiliary cone $K$, which can be regarded as a ``forbidden'' direction.
To do this, consider the auxiliary cone $K$ and two cones $\Lambda_1$ and $\Lambda_2$, such that they are mutually disjoint.
Now consider a circle large enough such that it contains the apexes of all three cones.
Then the intersection of the cones with the circle gives us three disjoint intervals, which we denote $I_1, I_2$ and $I_K$.
We say that $\Lambda_1$ is to the left of $\Lambda_2$, if we can rotate the interval $I_1$ counter clockwise along the circle in such a way that it will first intersect with $I_K$ before it intersects with $I_2$.
The choices of moving left or moving right generally lead to \emph{distinct} (but nevertheless related) braidings (the argument is similar to the compact case described in for example~\cite{Halvorson}).

\begin{rem}
If the spatial dimension is $\nu > 2$, it is no longer possible to define ``left'' and ``right'' in this way.
This however is no issue: in this case one can continuously deform one cone into another without either crossing the auxiliary cone $K$ or a fixed localization cone $\Lambda$.
As a consequence the definition of the braiding does not depend on the relative position of the cones, and the braiding is in fact a symmetry, in the sense that interchanging two charges twice is the trivial operation.
In other words, we only have bosons and fermions.
\end{rem}

We will not attempt to classify all possible braidings.
Rather, the goal is to recover the braiding defined in Ref.~\cite{Naaijkens11} for strictly localized sectors.
To this end, consider $\rho \in \Delta$.
In the construction of the bi-asymptopia, we have to choose cones $\Lambda_\alpha^{\mathcal{U}}$ and $\Lambda_\beta^{\mathcal{V}}$.
We choose $\Lambda_\alpha^{\mathcal{U}}$ to be to the left of $\Lambda_\beta^{\mathcal{V}}$.
Here ``left'' is defined as in~\cite{Naaijkens11} and explained above.
This is the situation depicted in Figure~\ref{fig:cones}.

\begin{thm}\label{thm:biasymptopia}
	Suppose $\nu\geq 2$ and $\Delta$  satisfies the assumptions of Theorem~\ref{thm:Asymptopia}.
	Then, there exists a bi-asymptopia $\{ \mathcal{U}, \mathcal{V}\}$ for $\Delta$
	and given $\rho, \sigma \in \Delta$ the following limit exists:
	\begin{equation}\label{eqn:braid}
	\epsilon_{\rho,\sigma} \equiv \lim_{m, n \ra \infty } ( V_{n}^* \otimes U_{m}^*) (  U_{m} \otimes  V_n ).
	\end{equation}
	The limit is independent of the choice of $\{ U_m \} \in \mathcal{U}_\rho$ and $ \{ V_n \} \in \mathcal{V}_\sigma$, and we have $ \epsilon_{\rho,\sigma} \in ( \rho \otimes \sigma, \sigma \otimes \rho)$.
	
	Furthermore, if $ R \in (\rho, \rho')$ and $ S \in (\sigma, \sigma')$ then 
	\begin{equation}\label{eqn:braid1}
		\epsilon_{\rho',\sigma'} (R\otimes S) = (S \otimes R) \  \epsilon_{\rho, \sigma}
	\end{equation}
	and if $\tau \in \Delta$  then 
	\begin{align}
	\epsilon_{ \rho \circ \sigma, \tau}  &= ( \epsilon_{ \rho, \tau} \otimes 1_\sigma )( 1_\rho \otimes \epsilon_{\sigma, \tau}) \label{eqn:braid2}\\
	\epsilon_{ \rho, \sigma \circ \tau} &= ( 1_\sigma \otimes \epsilon_{ \rho, \tau}) ( \epsilon_{ \rho, \sigma} \otimes 1_\tau) \label{eqn:braid3}.
	\end{align}
	Here $1_\sigma \equiv I$, seen as the trivial intertwiner from $\sigma$ to itself.
	That is, the braiding is natural in both variables, and satisfies the braid equations.
	Hence $\Delta$ is a braided tensor $C^*$-category.
\end{thm}

\begin{proof}
	Let $\Lambda_\alpha^{\mathcal{U}}$ and $\Lambda_\beta^{\mathcal{V}}$ be convex cones in $\mathcal{C}$ 
	such that for some $\epsilon >0$ and $x \in \ZZ^\nu$ we have
	$ \Lambda_{\alpha+\epsilon}^{ \mathcal{U}}\ll  (K+x)^c$, 
	$ \Lambda_{\beta+\epsilon}^{ \mathcal{V}}\ll  (K+x)^c$
	and $\Lambda_{ \alpha + \epsilon}^{\mathcal{U}} \ll (\Lambda_{\beta}^{ \mathcal{V}})^c$.
	Moreover we choose $\Lambda_\alpha^{\mathcal{U}}$ to the left of $\Lambda_\beta^{\mathcal{V}}$, as explained above.
	For instance, see Figure~\ref{fig:cones} for a sample configuration on $\ZZ^2$.
	A similar configuration exists for all $\nu \geq 2$, where for $\nu > 2$ the choice for ``left'' or ``right'' is not important any more.
	We will freely use Lemma~\ref{lem:indK} on the independence of the cone in the definition of the tensor product.
	
	Construct families of unitary sequences $\mathcal{U}_\rho$  and $\mathcal{V}_\rho$ 
	analogously to $\mathcal{K}_\rho$ as defined in Theorem~\ref{thm:Asymptopia}, 
	except in the construction the role of the cone $K$ is played by the cones $  \Lambda_{\alpha}^{ \mathcal{U}}$ and 
	$ \Lambda_{\beta}^{ \mathcal{V}}$, respectively.
	If $ \{ U_m\} \in \mathcal{U}_\rho $ then  
	$\rho(A) = \lim_{m \ra \infty} U_m^* A U_m$ for all $ A \in \calA_\Gamma$
	and  there exists an increasing sequence $k_m \in \NN$ and a sequence of $*$-endomorphisms $\rho_m$  such that 
	$\rho_m$ is almost localized in $\Lambda^{\mathcal{U}}_{\alpha- \epsilon(\{U_m\})}+k_m$  with decay function $f$, chosen independent of $m$, 
	and $U_m \in (\rho, \rho_m)$.
	$\mathcal{U}_\rho$ is stable, non-empty, and the sequence $ \rho_m$  converges pointwise to the identity, see~\eqref{eqn:pointwiseidentity}.
	We define $ \mathcal{V}_\rho$ similarly where  $ \Lambda_\beta^\mathcal{V}$ will play the analogous role for $\Lambda_\alpha^\mathcal{U}.$
	
	Let $\rho,\rho' \in \Delta$ and $\{ U_m\} \in \mathcal{U}_\rho$ and $ \{U_m'\} \in \mathcal{U}_{\rho'}$.
	For all $A\in \calA_\Gamma$ we have
	\begin{align*}
	\lim_{m \ra \infty} (U_m \rho( U_m') )^* A U_m \rho( U_m') & = 	\lim_{m \ra \infty} U_m^* \rho_m( U_m'^*) A  \rho_m( U_m') U_m = \rho \circ \rho' (A)
	\end{align*}
	where $\rho(U_m')$ exists by Theorem \ref{thm:Asymptopia}  and  for the second equality we used \eqref{eqn:pointwiseidentity}.
	Moreover, $U_m \rho(U_m') \in (\rho \circ \rho', \rho_m \circ \rho_m')$.
	Using Proposition~\ref{prop:aloccomp} one can find $\epsilon(\{ U_m \rho(U_m') \})>0$ and appropriate integers $k_m$ as in the proof of Theorem~\ref{thm:Asymptopia}.
	It follows that $ \{ U_m \otimes U_m'\} \in \mathcal{U}_{\rho \circ \rho'}$.
	
	Let $\epsilon'>0$ be given.
	For $\rho, \rho', \sigma, \sigma' \in \Delta$ and $ \{U_m\}\in \mathcal{U}_\rho$, $\{U_{m'}'\} \in \mathcal{U}_{\rho'}$, $  \{ V_n\} \in \mathcal{V}_\sigma$ and $\{ V_{n'}'\} \in \mathcal{V}_{\sigma'}$ 
	we have that 
	$ U'_{m'} R U_{m}^* \in ( \rho_m, \rho'_{m'} )$ and $ V'_{n'} S V_{n}^* \in (\sigma_n, \sigma_{n'}')$.
	By Corollary~\ref{cor:intertwineraloc} and the construction of $\mathcal{U}$ and $ \mathcal{V}$,
	the intertwiner $U'_{m'} R U_{m}^*$ is almost localized in $\Lambda_{ \alpha-\epsilon}^{\mathcal{U}} + \min\{k_m,k_{m'}\}$ with decay function $2 \| R\| f$
	and $ V'_{n'} S V_{n}^* $ is almost localized in $\Lambda^{\mathcal{V}}_{\beta-\epsilon} + \min\{k_n, k_{n'} \}$ with decay function $2 \| S\| g$,
	where $\epsilon = \min\{ \epsilon(\{U_m\}),\epsilon(\{U_{m'}'\}), \epsilon(\{V_n\}),\epsilon(\{V_{n'}'\})  \}$.
	Thus, there exists $M, N>0$ and  $ R_M \in \mathcal{R}( (\Lambda_{ \alpha }^{\mathcal{U}} + M)^c)$ and $S_N \in \mathcal{R}( (\Lambda_{ \beta }^{\mathcal{V}} + N)^c)$ such that for all $m,m',n,n'$ sufficiently large we have that 
	$\| U'_{m'} R U_{m}^* - R_M\| < \epsilon' \|R\|$ and $\| V'_{n'} S V_{n}^* - S_N\| < \epsilon' \|S\|$.
	It follows that
	\begin{align*}
	& \| U'_{m'} R U_{m}^* \otimes V'_{n'} S V_{n}^* - V'_{n'} S V_{n}^* \otimes U'_{m'} R U_{m}^* \| \\
	& \qquad = \| U'_{m'} R U_{m}^* \rho_m( V'_{n'} S V_{n}^* )- V'_{n'} S V_{n}^* \sigma_n(  U'_{m'} R U_{m}^*) \|  \\
	&  \qquad\leq  3 \| R\| \| V'_{n'} S V_{n}^*  -  S_N\| 
	+ 3\| S \| \|  U'_{m'} R U_{m}^* - R_M \| \\
	& \qquad \qquad +\| R\| \| \rho_m( S_N ) - S_N\| 
	+ \| S \| \| \sigma_n( R_M) -  R_M\|  + \|R_M S_N - S_NR_M \|\\
	& \qquad < 6 \epsilon' \| R\| \|S\| + 2 \|R\| \|S\|(  f_\epsilon(N + m) + g_\epsilon(M +n)) . 
	\end{align*}
	Taking the limit as $m,m',n,n' \ra \infty$ and since $\epsilon'>0$ was arbitrary, 
	gives that $\mathcal{U}$ and $ \mathcal{V}$ form a bi-asymptopia for $\Delta.$
	The result follows from Theorem 8 of \cite{BuchholzAA}.
	
	Here we include a calculation showing that $\epsilon_{ \rho,\sigma}$ intertwines $ \rho \otimes \sigma$ with $\sigma \otimes \rho$.
	A crucial point is that  $\rho_m $ and $\sigma_n$, being almost localized in far removed disjoint cones, commute asymptotically:
	\begin{align*}
		\epsilon_{ \rho, \sigma} (\rho \otimes \sigma) (A) & = \lim_{m,n \ra \infty} ( V_{n}^* \otimes U_{m}^*) (  U_{m} \otimes  V_n ) \rho (\sigma(A))
	& = \lim_{m,n \ra \infty} V_n^* \sigma_n(U_m^*) U_m \rho( V_n) \rho(\sigma(A))\\
	& = \lim_{m,n \ra \infty} V_n^* \sigma_n(U_m^*) U_m \rho( \sigma_n(A)V_n)
	& = \lim_{m,n \ra \infty} V_n^* \sigma_n(U_m^*)  \rho_m( \sigma_n(A)V_n) U_m\\
	& = \lim_{m,n \ra \infty} V_n^* \sigma_n(U_m^*)  \sigma_n( \rho_m(A))\rho_m(V_n) U_m
	& = \lim_{m,n \ra \infty} V_n^* \sigma_n(U_m^* \rho_m(A))\rho_m(V_n) U_m\\
	& = \lim_{m,n \ra \infty} V_n^* \sigma_n( \rho(A) U_m^*)\rho_m(V_n) U_m 
	& = \lim_{m,n \ra \infty}  \sigma(\rho(A))V_n^* \sigma_n(U_m^*) U_m \rho( V_n)\\
	& = \lim_{m,n \ra \infty} \sigma (\rho(A)) ( V_{n}^* \otimes U_{m}^*) (  U_{m} \otimes  V_n )\\
	& = (\sigma\otimes \rho)(A) \epsilon_{ \rho,\sigma}
	\end{align*}
	Calculations along the same lines as above hold to show properties \eqref{eqn:braid1}, \eqref{eqn:braid2}, and \eqref{eqn:braid3}.
\end{proof}

In case one has Haag duality for cones and strictly cone-localized charges one can do the usual sector theory, as we discussed earlier.
For sector theory of almost localized sectors to be a proper generalization of this case, at least the tensor product defined above should coincide with the original definition.
Fortunately, this turns out to be the case.
\begin{prop}
Suppose that Haag duality for cones hold. Then on strictly localized and transportable endomorphisms, the tensor product coincides.
\end{prop}
\begin{proof}
Recall that in the analysis of strictly cone-localized sectors one first chooses an auxiliary cone $K$ and define an auxiliary algebra $\calA^K \equiv \overline{\bigcup_{x \in \mathbb{Z}^2 } \mathcal{R}((K+x)^c)}^{\| \cdot \|}$, see~\cite{Naaijkens11}.
Choose $K$ to be the cone of Theorem~\ref{thm:Asymptopia}.
Let $\rho$ be strictly localized in some cone $\Lambda$ that is disjoint from $K+x$ for some $x$.
Then one can show that $\rho$ has a unique weakly continuous extension $\rho^K$ to $\calA^K$.
This is done by choosing a unitary $V$ such that $V \rho(A) V^*$ is localized in the auxiliary cone $K+x$ for some $x$.
Then, by localization, it follows that $\rho(A) = V^* A V$ for all $A \in \calA_{(K+x)^c}$.
Hence $\rho^K(A) \equiv V^* A V$, now with $A \in \mathcal{R}((K+x)^c)$, is a weakly continuous extension of $\rho$.

Note that the choice of $V$ does not matter: for another choice $W$, $VAV^* = WAW^* = \rho(A)$ on $\calA_{(K+x)^c}$, and hence by weak continuity their extensions coincide on $\mathcal{R}((K+x)^c)$.
In particular, one can choose the unitaries $U_m$ from the asymptopia $\mathcal{K}_\rho$, as defined in the proof of Theorem~\ref{thm:Asymptopia}.
It is then clear that on $\calA_\Delta \cap \mathcal{R}( (K+x)^c)$ the extensions $\rho^K$ and $\rho_{\mathcal{K}}$ coincide.
Moreover, by Haag duality one can show that $T \in (\rho,\sigma)$ is in $\mathcal{R}(\Lambda)$ for any cone $\Lambda$ containing the localization regions of both $\rho$ and $\sigma$, hence $\rho_\mathcal{K}(T) = \rho^K(T)$, and both tensor products coincide.
\end{proof}

Since the morphisms in the category are the intertwiners in both cases, the following follows immediately.
\begin{cor}\label{cor:monoidalinclusion}
There is a full and faithful monoidal functor from the category of strictly localized and transportable endomorphisms to $\Delta$. 
\end{cor}
This still leaves open the possibility that there are sectors in $\Delta$ that are not equivalent to a strictly localized sector.
In other words, the inclusion might not be an equivalence of braided tensor categories.
We will show later that -- under the additional assumption that a natural energy criterion is satisfied -- the categories are indeed equivalent.

\section{Stability of the superselection structure}\label{sec:stability}
The results in the previous section show that we can define a braided tensor $C^*$-category $\Delta$ of almost localized and transportable endomorphisms.
It is important to remember that the notion of transportability, and hence $\Delta$ itself, depends on the choice of reference representation $\pi_0$.
In this section we study what happens if we change the reference representation $\pi_0$.

In general, there does not need to be a relation between the categories $\Delta$ for two reference representations $\pi_0$ and $\pi_0'$.
To give an example, consider two finite abelian groups $|G_1| = |G_2|$ which are not isomorphic.
Then one can consider the corresponding abelian quantum double models~\cite{KitaevQD}.
In the thermodynamic limit, both models have the same underlying quasi-local algebra $\calA_\Gamma$, and both have a unique, translation invariant pure ground state $\omega_0^{G_1}$ and $\omega_0^{G_2}$, with corresponding ground state representations.
Since automorphisms act transitively on the pure states of $\calA_\Gamma$, see for example~\cite[Thm. 12.3.4]{KR2}, it follows that there is an automorphism $\alpha \in \operatorname{Aut}(\calA_\Gamma)$ such that $\pi_0^{G_1} \circ \alpha = \pi_0^{G_2}$.
Nevertheless, the corresponding categories of localized and transportable endomorphisms are not equivalent~\cite{FiedlerN}.
As a corollary, our results show that $\alpha$ must be highly non-local.
A similar result has been obtained by Haah using different methods~\cite[Thm. 4.2]{Haah}.

Recall that two ground states of gapped local Hamiltonians $H(0)$ and $H(1)$ respectively are said to be in the same phase if there is a continuous path of gapped local Hamiltonians $t \mapsto H(t)$ with $H(0)$ and $H(1)$ as given.
One way to get examples of such paths is to perturb a given local Hamiltonian.
Under suitable conditions, one can show that the perturbed ground state is related to the unperturbed one via an automorphism $\alpha_s$~\cite{BachmannMNS},
see Appendix~\ref{app:spectral} for the main points.
This automorphism is local, in the sense that it satisfies a Lieb-Robinson type of bound.
This is the proper generalization in the thermodynamic limit of \emph{local unitary circuits}~\cite{ChenGW}.
This is the type of automorphisms that we will be interested in.
In the remainder of this section we will work this out in detail.

As before we restrict to the case $\Gamma = \ZZ^\nu$, with distance function $d(x,y) = \abs{x-y}$.
The automorphisms that we are interested in are typically obtained as the dynamics generated by local dynamics,
where the strength of the local interactions should decay fast enough for the dynamics to be defined at all.
To make this precise, we follow~\cite{NachOS} and first introduce $\mathcal{F}$-functions.
A function $F: \RR^{\geq 0} \ra \RR^{\geq 0}$ is called an $\mathcal{F}$-function for $\ZZ^\nu$ if it is monotone decreasing and satisfies
\begin{equation}\label{eqn:unifint}
	\|F \|_0 = \sup_{x \in \Gamma} \sum_{y \in \Gamma} F(d(x,y)) < \infty  \quad \quad \text{ (uniform integrability),}
\end{equation}
\begin{equation}
	C_F  = \sup_{x,y\in \Gamma} \sum_{z\in \Gamma} \frac{ F(d(x,z)) F(d(z,y))}{F(d(x,y))} < \infty \qquad \text{ (convolution identity)}.
\end{equation}
It can be checked that the function $F(r) = \frac{1}{(1+r)^{\nu+\epsilon}}$ is an $\mathcal{F}-$function for all $\epsilon>0$.
Let $ b>0$ and $g:\RR^{\geq 0} \ra \RR^{\geq 0}$ be uniformly continuous, non-decreasing and sub-additive, that is, $g(x+y) \leq g(x) + g(y)$.
If $F$ is an $\mathcal{F}$-function, then
\begin{equation}
F_{bg}(r)\equiv e^{-bg(r)} F(r)
\end{equation}
also is an $\mathcal{F}$-function.
These properties are satisfied by the following functions:
\begin{align}
g(r) &= r^\alpha  \qquad \quad \mbox{ for } 0<\alpha\leq 1, \label{eqn:expg}\\
g(r) &= \left\{ \begin{array}{ll}
\frac{r}{\ln^p(r)} & \mbox{ if }  x> e^p\\
\left( \frac{e}{p}\right)^p & \mbox{ if } x\leq e^p
\end{array} \right. .\label{eqn:subexpg}
\end{align}
When the context is clear, we will abuse notation and denote $ F_{b}(r) \equiv e^{-b r }F(r) $.

One measure of locality in quantum spin system is by commutator bounds.  
Let $X,Y \in \mathcal{P}_0(\Gamma)$.
If $d(X,Y) >0$ then $[A,B] = 0$ for all $A\in\calA_X$ and $ B \in\calA_Y$.  
To measure the diffusion or spreading of a dynamics $\tau_t^\Lambda$, 
Lieb and Robinson \cite{LiebR} considered bounding the commutator $[ \tau_t^\Lambda(A), B]$.
Since the Heisenberg dynamics are non-relativistic, 
it is generally expected that the commutator norm is non-zero at any finite time $t>0$.
In the relativistic case,  the commutator is non-zero only if $\tau_t(A)$ and $B$ are not spacelike seperated.
The Lieb-Robinson bounds motivate the following definition.

\begin{defn}
	A one-parameter family of automorphism $\tau_t \in \operatorname{Aut}(\calA_\Gamma)$ is called a \emph{quasi-local dynamics} if 
there exist an $\mathcal{F}-$function $F$ and constants $v,C_F >0$ such that 
for any $A \in \calA_X$ and $B \in \calA_Y$  with $d(X,Y) >0$, 
\begin{equation}\label{eqn:qlbound}
\| [ \tau_t (A) ,B ] \| \leq \frac{2 \|A \| \| B\|}{C_F} (e^{ v \abs{t}} - 1) \sum_{x\in X}\sum_{y\in Y}F(d(x,y))
\end{equation}
for any $t \in \RR$.
\end{defn}

Quasi-local dynamics are typically generated by a short-ranged interaction map.
Here we consider interactions $\Phi$ satisfying a finite $F$-norm, that is,
	\begin{equation}
	\| \Phi\|_F \equiv \sup_{x,y \in \Gamma} \frac{1}{F(d(x,y))} \sum_{\substack{ X \subset \Gamma: \\x,y \in X}} \| \Phi(X) \| < \infty.
	\end{equation}
	Let $\Phi$ be an interaction with a finite $F$-norm.
	Then, for any increasing and exhausting sequence $\Lambda_n$, that is, if $\Lambda_n \subset \Lambda_{n+1}$ and $\Gamma = \bigcup_{n} \Lambda_n$,
	 the norm limit 
	$\tau_t(A) \equiv \lim_{n \ra \infty} \tau_t^{\Lambda_n}(A) $
	exists for all $t\in \RR$ and $A \in \calA_\Gamma$~\cite{NachOS}.
	Here $\tau_t^{\Lambda_n}$ is the dynamics generated by the local Hamiltonian $H_{\Lambda_n} \equiv \sum_{X \subset \Lambda_n} \Phi(X)$.
	The limiting dynamics $\tau_t$ defines a strongly continuous, one-parameter group of automorphisms on $\calA_\Gamma$.
	The convergence is uniform for $t$ in compact sets and 
	is independent of the sequence $\Lambda_n$.
Moreover, if $\Phi$ satisfies a finite $F$-norm, then $\tau_t$ satisfies a Lieb-Robinson bound~\cite{NachOS} and is a quasi-local dynamics with the same $\mathcal{F}$-function. 
The estimate of the form (\ref{eqn:qlbound}) given in \cite[Theorem 2.1]{NachOS} is
stated for observables with finite supports $X$ and $Y$ and for the dynamics generated by the interactions in a finite set.
However, since the LHS of (\ref{eqn:qlbound}) is continuous as a function of $A$ 
and $B$ on $\calA_{loc}$, and the size of the region for which the dynamics are defined do not appear in the bound,
the bound extends directly to $A\in\calA_X$ and $B\in\calA_Y$ for infinite sets $X$ and $Y$, and gives a non-trivial estimate as long as the RHS is finite.
Indeed, showing that it is finite in interesting cases of quasi-local observables is a key technical tool in our analysis.

\subsection{Lieb-Robinson bound for cones}
In the definition of almost localized endomorphisms, we consider operators in two cone regions.
However, if $X, Y \subset \Gamma$ are infinite regions 
then the bound in \eqref{eqn:qlbound} may not be better than the trivial bound $2\|A\| \|B\|$.
For $X, Y\in \mathcal{C}$  and  $Y\subset X^c$ two far separated cone regions,
we show that the bound \eqref{eqn:qlbound} recovers a good approximation.
Results of this type were first discussed and proved in the unpublished thesis of Schmitz~\cite{Schmitz}. 
For completeness, and because access to \cite{Schmitz} is not readily available, we present the results here.
The bounds here are crucial in proving stability of the superselection structure.

We first state our assumptions.
\begin{assumption}\label{asp1}
	Let $g:\RR^{\geq 0} \ra \RR^{\geq 0}$ be uniformly continuous, non-decreasing and sub-additive.
	We assume that for all $b>0$ and $k \in \NN$ there exists a $t_0 > 0$ such that for all $t > t_0$,
	\begin{equation}
	\int_{t}^\infty r^k e^{-b g(r)} dr \leq K_{b,k}t^{l(k)} e^{-b g(t)}.
	\end{equation}
	for some positive function $ K_{b,k}>0$ and affine function $l(k)$.
\end{assumption}

By uniform continuity, if $g$ satisfies Assumption~\ref{asp1} then $\lim_{r\ra \infty} r^k e^{ - b g(r)}  = 0$ for all $k\in \NN$. 
From the following 	inequalities
\begin{align*}
\int_{t}^\infty r^k e^{-b r} dr &\leq \frac{k+1}{b} t^{k} e^{- b t} \quad \mbox{ for } t>k,\\
\int_{t}^\infty r^k e^{-b \frac{r}{\ln^2 r}} dr &\leq \frac{2k+3}{b} t^{2k+2} e^{- b \frac{r}{\ln^2 r}} \quad \mbox{ for } t>e^4,
\end{align*}
(see Lemma 2.5 of~\cite{BachmannMNS}) we see that for example the functions
\begin{align*}
g(r) = r \qquad\mbox{and}\qquad g(r) = \frac{r}{\ln^2r}
\end{align*}
satisfy Assumption \ref{asp1}.

The main result in this section is the following Lieb-Robinson bound for cones.
\begin{thm}\label{thm:LRcone}
	Suppose $g$ satisfies Assumption~\ref{asp1} and 
	$\Phi$ has a finite $F_{bg}$-norm for some $b>0$.  
	Let $X \in \mathcal{C}_\alpha$ (i.e., $X$ is some cone $\Lambda_\alpha$) and define $Y_{\epsilon,n} \equiv \left( \Lambda_{\alpha+ \epsilon} - n\right)^c $.
	Then there exists an affine function $\tilde{l}$ and $v_{bg}>0$ such that
	for all $0\leq \alpha < \pi$, $ 0< \epsilon< \pi - \alpha$,
	and observables $A\in \calA_X$ and $B\in \calA_{Y_{\epsilon,n}}$, 
	\begin{equation}
		\| [ \tau_t(A),B] \|\leq 2 \|A \| \| B\| C_\epsilon d(X,Y_{\epsilon,n})^{\tilde{l}(\nu)} e^{v_{bg} \abs{t} - b g( d(X,Y_{n,\epsilon})\sin\epsilon)},
	\end{equation}
	where 
	\begin{equation}
	n \sin(\alpha + \epsilon) \leq d(X,Y_{\epsilon,n}) \leq n \sin(\alpha + \epsilon)+2
	\end{equation}
	and $C_\epsilon$ is non-increasing in $\epsilon$ and only depends on $\nu$, $b$ and $\alpha$.
\end{thm}

\begin{proof}
	Substituting the bound found in Lemma \ref{lem:doubImpInt} into the Lieb-Robinson bound \eqref{eqn:qlbound} gives the result.
\end{proof}

The quasi-locality of the dynamics $\tau_t$ can be interpreted by measuring the growth of the support of a time-evolved observable.
Let $\Lambda \in \mathcal{P}_0(\Gamma)$ be a finite  subset and take $X\subset \Lambda$.
For any observable $A \in \calA_\Lambda$ consider the conditional expectation
\begin{equation}\label{eqn:partialtrace}
\langle A \rangle_{X^c}  \equiv \int_{\mathcal{U}(X^c\cap \Lambda)} U^* A U \mu(dU),
\end{equation}
where $\mu$ is the normalized Haar measure on the family of unitary operators $\mathcal{U}(X^c\cap\Lambda) \subset \calA_\Lambda$, that is,
$\int_{\mathcal{U}(X^c\cap\Lambda)}  \mu(dU)= 1$.
If $ A \in \calA_X$ then $A$ commutes with $\mathcal{U}(X^c)$ so that $\langle A\rangle_{X^c} = A$.
Moreover, invariance of the Haar measure implies that $\langle A \rangle_{X^c}$ commutes with any unitary localized in $\calA_{\Lambda \cap X^c}$, and hence it is an element of $\mathcal{A}_X$.
In particular, $\langle \langle A \rangle_{X^c} \rangle_{X^c}=\langle A \rangle_{X^c} $ for all $A \in \calA_\Lambda$.

We first recall how Lieb-Robinson bounds can be used to obtain a local approximation of the time evolution of a local operator.
Let $ n>0$ and denote
$B_t(X,n) = \{ y \in \Gamma :  d(y,X) \leq v_{bg} \abs{t} + n \}.$
Suppose $ A \in \calA_X$.
Applying the Lieb-Robinson bound \eqref{eqn:qlbound}, it follows that 
\begin{align}
\left\| \tau_t(A) - \langle \tau_t(A) \rangle_{B^c_t(X,n)} \right\| 
&= \left\| \int_{\mathcal{U}(B_t^c(X,n))} [ \tau_t(A), U] \  \mu(dU) \right\| \\
&\leq  \sup_{U \in\mathcal{U}(B_t^c(X,n)) } \| [ \tau_t(A), U] \| \\
& \leq K_a \| A \| \abs{X} (1 - e^{ - v_{gb} \abs{t}}) e^{- b g(n) }, \label{eqn:fvQLdyn}
\end{align}
where $K_a$ is a constant (which can be calculated explicitly by doing the double summation in the bound~\eqref{eqn:qlbound}).

Now let $X\subset \ZZ^\nu$ be a potentially infinite subset, e.g., $X = \Lambda_\alpha$ an infinite cone.
Denote the following subsets of $\ZZ^\nu$ by 
\[ \Lambda_L \equiv [-L,L]^\nu \cap \ZZ^\nu, \quad  X_L  \equiv X \cap \Lambda_L, \quad \text{ and } \quad X_L^c \equiv \Lambda_L \setminus X_L. \]
\begin{defn}
	We define $ \langle \ \cdot \  \rangle_{X^c} : \calA_\Gamma \ra \calA_{X} \cong \calA_X \otimes I \subset \calA_\Gamma$ by
	\begin{equation}\label{eqn:conePT}
	\langle A \rangle_{X^c} \equiv \lim_{L \ra \infty} \langle A \rangle_{X_L^c}
	= \lim_{L\ra \infty} \int_{\mathcal{U}(X_L^c)} U^* A U \mu(dU) \quad \text{ for all } \quad A \in \calA_\Gamma,
	\end{equation}
	where the limit is in the norm sense.
\end{defn}

\begin{lemma}\label{lem:coneProj}
	The operator $\langle A \rangle_{X^c}$ is well defined: the limit~\eqref{eqn:conePT} exists and is unique, and $\langle A \rangle_{X^c} \in \calA_{X}$.
\end{lemma}

\begin{proof}
	Let $\epsilon>0$ be given.  
	By density of $\calA_{loc}$ in $\calA_\Gamma$ there exists $L>0$ and an operator $A_L \in \calA_{\Lambda_L}$ such that 
	\[ \| A - A_L \| < \epsilon. \]
	For any $n,m > L$ we have that 
	\begin{align*}
	\| \langle A \rangle_{X_m^c} - \langle A \rangle_{X_n^c} \| 
	& \leq \| \langle A \rangle_{X_m^c} - \langle A_L \rangle_{X_m^c} \| 
	+  \| \langle A \rangle_{X_n^c} - \langle A_L \rangle_{X_n^c} \| 
	+ \| \langle A_L \rangle_{X_m^c} - \langle A_L \rangle_{X_n^c}\|\\
	& = \left\| \int_{\mathcal{U}(X_m^c)} U^* (A - A_L) U \mu (dU) \right\|
	+ \left\| \int_{\mathcal{U}(X_n^c)} U^* (A - A_L) U \mu (dU) \right\|\\
	& \leq 2 \epsilon,
	\end{align*}
	where we use that $\langle A_L \rangle_{X_m^c} = \langle A_L \rangle_{X_n^c} =  A_L$ for the last term. 
	Thus, the sequence $ \langle A \rangle_{X_L^c}$ is Cauchy in $\calA_X$ and its limit is defined to be the operator $\langle A \rangle_{X^c} \in \calA_X$.
\end{proof}

This result can be combined with the Lieb-Robinson bound for cones, to approximate time-evolved observables localized outside cones (compare with~\cite[Satz II.8]{Schmitz}).
\begin{cor}\label{cor:quasiloc}
	Suppose $g$ satisfies Assumption \ref{asp1} and $\Phi$ has a finite $F_{bg}$-norm for $b>0$.  
	Let  $X$ and $ Y_{\epsilon,n}$ as from Lemma \ref{lem:doubImpInt}.
	Denote $ Y_{\epsilon,n}(t) \equiv (\Lambda_{\alpha+ \epsilon} - (\lceil \frac{v_{bg} \abs{t}}{a} +  n \rceil))^c$.
	Then, there exists an affine function $\tilde{l}$ such that 
	for all $0\leq \alpha < \pi$, $ 0< \epsilon< \pi - \alpha$,
	and $ A \in \calA_{X}$,
	\begin{equation}
	\| \langle \tau_t(A) \rangle_{Y_{\epsilon,n }(t)} - \tau_t(A) \| \leq C_\epsilon \| A\|  d(X,Y_{\epsilon,n}(t))^{\tilde{l}(\nu)} e^{- b g( d(X,Y_{\epsilon,n}(t)) \sin\epsilon )},
	\end{equation}
	where 
	\[ \Big\lceil v_{bg} \abs{t} +  n \Big\rceil \sin(\alpha + \epsilon)
	\leq d(X,Y_{\epsilon,n}(t)) 
	\leq \Big\lceil v_{bg} \abs{t}+  n \Big\rceil  \sin(\alpha + \epsilon) + 2, \]
	and $C_\epsilon$ is non-increasing in $\epsilon$ and only depends on $\nu$, $b$ and $\alpha$.
	In particular, 
	\begin{equation}\label{eqn:coneql}
	\lim_{n\ra \infty}  n^{k} \| \langle \tau_t(A) \rangle_{Y_{\epsilon,n}(t)} - \tau_t(A) \| = 0 \quad \text{ for all } \quad k\in \NN.
	\end{equation}
\end{cor}

\begin{proof}
	The proof is similar to the argument given for the finite volume case~\eqref{eqn:fvQLdyn} where we take the limit as in Lemma~\ref{lem:coneProj}. 
	Equation~\eqref{eqn:coneql} comes from a  similar argument to the proof of Corollary~\ref{cor:fsatpolydecay}.
\end{proof}

The bounds in Theorem~\ref{thm:LRcone} and Corollary~\ref{cor:quasiloc} have exact generalizations for quasi-local dynamics.

\subsection{Stability}
Let $\tau_t$ be a quasi-local dynamics for some $F_{bg}$ with $b>0$.
We will consider endomorphisms of the form $\tau_t^{-1}\circ \rho \circ \tau_t$.
In a later section, we will argue why if $\rho$ describes a charge in an unperturbed system, the new endomorphism corresponds to a charge in a perturbed system.
The main result of this section is that the evolution by a quasi-local dynamics will preserve the defining property of almost-localized endomorphisms. 
The main tool in the proof will be the Lieb-Robinson bound for cones that was established in the previous section.

Suppose $g$ satisfies Assumption~\ref{asp1}.
Then, there exist an $n_0 > 0 $ such that for constants $r, b, a >0$ and $C_\epsilon$ non-increasing in $\epsilon$, the function
\begin{equation}
h_\epsilon(n) \equiv 
\Bigg\{ \begin{array}{ll}
C_\epsilon n_0^r e^{ -b g( a n_0 ) \sin\epsilon} &\mbox{ if } n \leq n_0\\
C_\epsilon n^r e^{ -b g( a n ) \sin\epsilon}  & \mbox{ if } n > n_0
\end{array}
\end{equation}
is in the class $\mathcal{F}_\infty$.
By Corollary~\ref{cor:quasiloc}, there are constants $r,b,a>0$  such that 
\begin{equation}\label{eqn:quasiloc}
\| \langle \tau_t(A) \rangle_{Y_{\epsilon,n }(t)} - \tau_t(A) \| \leq h_\epsilon(n) \|A\|
\end{equation}
for $A$ localized in a cone $X$ and $Y_{\epsilon,n}(t)$ as defined above.

\begin{lemma}\label{lem:stabaloc}
Suppose that $g$ satisfies Assumption~\ref{asp1} and $\tau_t$ is a quasi-local dynamics for some  $F_{bg}$-function and $b>0$.
If $ \rho$  is an almost-localized endomorphism in $\Lambda_\alpha$
with decay function $f$ 
then for all $ t \in \RR$, $ \tau_t^{-1} \circ \rho \circ \tau_t$ is an
almost-localized endomorphism in $\Lambda_{\alpha}$
with decay function $f_{\epsilon/2}(n/2) + 2 h_{\epsilon/2}(v_{bg} \abs{t}+ n/2)$.
\end{lemma}

\begin{proof}
Let $\epsilon >0$ be given.
By Corollary \ref{cor:quasiloc} and \eqref{eqn:quasiloc} we have that there exists constants $C, r >0$ such that  if $n$ is even then
\begin{align}
\sup_{A \in \calA_{\Lambda^c_{\alpha+\epsilon}-n}} \frac{  \| \tau_t^{-1} \circ \rho \circ \tau_t (A) - A \| }{\|A\|}
&\leq  \sup_{A \in \calA_{\Lambda^c_{\alpha+\epsilon}-n}}  \frac{ \| \rho( \langle \tau_t(A) \rangle) -  \langle \tau_t(A) \rangle \|}{\|A\|}
+ 2 \frac{\|  \langle \tau_t(A) \rangle - \tau_t(A)\| }{\|A\|}\\
& \leq f_{\epsilon/2}(n/2) +2 h_{\epsilon/2}(v_{bg} \abs{t}+ n/2),
\end{align}
where $\langle \ \cdot \  \rangle = \langle \ \cdot \ \rangle_{ \Lambda_{ \alpha + \epsilon/2} + n/2}$.
\end{proof}

Hence the above result says that almost localized endomorphisms are mapped to almost localized endomorphisms. This is in fact true for the complete tensor-$C^*$ structure on $\Delta$, so that quasi-local dynamics can be used to change the reference representation.

\begin{thm}\label{thm:stabsectorstructure}
Let  $\Delta$ is a semi-group of almost localized and transportable endomorphisms satisfying the superselection criterion, Definition~\ref{def:superselection}, 
for a reference state $\omega_0$.
If $\tau_t$ is a quasi-local dynamics then for all $ t\in[0,1]$ the semi-group $\tau_t^{-1} \circ \Delta\circ \tau_t $ satisfies the superselection criterion
for the reference state $\omega_0 \circ \tau_t$.
Furthermore, for all $\rho, \sigma \in \Delta$ we have that
 \begin{equation} \label{eqn:stabintertwiner}
( \rho, \sigma)_{\pi_0} = (\tau_t^{-1}\circ \rho \circ \tau_t, \tau_t^{-1} \circ \sigma \circ \tau_t)_{\pi_0\circ \tau_t}.
\end{equation}
If in addition $\Delta$ satisfies the assumptions of Theorem~\ref{thm:biasymptopia}, 
then  $\tau_t^{-1} \circ \Delta \circ \tau_t$ with reference state $\omega_0 \circ \tau_t$ is a braided tensor $C^*$-category
and is braided equivalent to $\Delta$.
\end{thm}

\begin{proof}
Let $\rho\in \Delta$ be almost localized in a cone $\Lambda_{ \alpha }$. 
Then, by Lemma~\ref{lem:stabaloc}, $\tau_t^{-1} \circ \rho \circ \tau_t$ is almost localized in $\Lambda_{ \alpha }$.
Thus, the first part of the superselection criterion is satisfied.

Let $T \in (\rho,\sigma)_{\pi_0}$.
It follows that $T\in (\rho \circ \tau_t, \sigma\circ \tau_t)_{\pi_0}$ and 
\begin{align*}
	T (\pi_0\circ \tau_t) \circ  \tau_t^{-1}  \rho \circ \tau_t (A)  = T  \pi_0 \circ \rho \tau_t (A) 
=  \pi_0 \circ \sigma \circ \tau_t(A) T
= (\pi_0\circ \tau_t) \circ  \tau_t^{-1}  \sigma \circ \tau_t (A) T.
\end{align*}
Thus, $T \in (\tau_t^{-1}\circ \rho \circ \tau_t, \tau_t^{-1} \circ \sigma\circ \tau_t)_{\pi_0 \circ \tau_t}$.
A similar argument shows the reverse inclusion, and hence that $(\rho,\sigma)_{\pi_0} = (\tau_t^{-1}\circ \rho \circ \tau_t, \tau_t^{-1}\circ\sigma\circ \tau_t)_{\pi_0 \circ \tau_t}$.

We now show that $\rho_t \equiv \tau_t^{-1} \circ \rho \circ \tau_t \in \tau_t^{-1} \circ \Delta \circ \tau_t$, where $\rho \in \Delta$, is transportable with respect to $\pi_0 \circ \tau_t$.
Let $\Lambda_{\beta}' \in \mathcal{C}$ be a cone. 
Since $\rho$ is transportable there is a $ \rho'$ almost localized in $\Lambda_{ \beta}'$ and a unitary $U$ such that $U  \pi_0\circ \rho(A) = \pi_0\circ\rho'(A) U$ for all $A\in\calA_\Gamma$.
Let $\rho_t' \equiv \tau_t^{-1} \circ \rho' \circ \tau_t$. 
By Lemma~\ref{lem:stabaloc} $\rho_t'$, is almost localized in $\Lambda_{ \beta}'$
and by~\eqref{eqn:stabintertwiner} we have that $( \rho, \rho')_{\pi_0} = (\rho_t, \rho_t')_{\pi_0\circ \tau_t}$.
Therefore, $\rho_t$ is transportable for $\omega_0 \circ \tau_t$.	

If in addition, $\Delta$ satisfies the assumptions of Theorem~\ref{thm:biasymptopia} then $\Delta$ is a braided tensor $C^*$-category.
Thus, $\tau_t^{-1} \circ \Delta \circ \tau_t$ is a braided tensor $C^*$-category with braiding defined by $\epsilon_{\tau_t^{-1} \circ \rho \circ \tau_t , \tau_t^{-1}\circ \sigma \circ \tau_t} \equiv \epsilon_{ \rho, \sigma}$.
Define a functor $F$ by $F(\rho) = \tau^{-1}_t \circ \rho \circ  \tau_t$ and $F(T) = T$ on the intertwiners.  It follows that 
\begin{align*}
F (\rho \otimes \sigma ) &= \tau_t^{-1} \circ \rho \circ \sigma \circ \tau_t 
= \tau_t^{-1} \circ \rho \circ \tau_t \circ \tau_t^{-1} \circ \sigma \circ \tau_t\\
&= F(\rho) \otimes F(\sigma).
\end{align*}
Hence $F$ is a $\otimes$-functor.
From the first part of the proof it easily follows that $F$ is a braided equivalence of tensor $C^*$-categories.
\end{proof}

\section{Application: stability of abelian quantum double models}\label{sec:applications}
We are now in a position to apply the above analysis to a prototypical example of a topologically ordered quantum spin system: Kitaev's quantum double model.
More precisely, we consider the model for an \emph{abelian} group $G$, and consider the quantum spin system and dynamics introduced by Kitaev~\cite{KitaevQD}.
It is well known that its anyonic excitations are described by the representation theory of the quantum double $\mathcal{D}(G)$ of the group algebra of $G$~\cite{BombinMD,KitaevQD}.
The category of such representations is a modular tensor category (see e.g.~\cite[Ch. 2]{BakalovKirillov}).
In the thermodynamic limit on the plane, the case that we are interested in here, this category has been obtained using the sector analysis outlined above~\cite{FiedlerN,Naaijkens11}.

The model has a gap above the ground state, which is stable under local perturbations \cite{BravyiHM}.
Here we apply the superselection criterion and stability results of the previous section
combined with the techniques of spectral flow~\cite{BachmannMNS} to in addition prove the stability of anyons in the abelian quantum double models.
The key idea is to consider the set $\mathcal{S}(s)$ of low-energy states.
Here $s$ parametrizes a path of perturbations, with $s=0$ being the unperturbed model.
The set $\mathcal{S}(s)$ is essentially the set of weak$^*$ limits of states of the finite model with energy below some fixed (volume-independent) threshold.
This set can be understood using techniques developed in previous work by the authors on classifying the infinite volume ground states of the unperturbed model~\cite{ChaNN}.
We then use the spectral flow~\cite{BachmannMNS} (or ``quasi-adiabatic continuation''~\cite{HastingsLSM,HastingsW}) to relate $\mathcal{S}(s)$ to $\mathcal{S}(0)$.
More precisely, this will yield a family of automorphisms $\alpha_s$ that satisfy a Lieb-Robinson type of bound, see Appendix~\ref{app:spectral}.
Applying Theorem~\ref{thm:stabsectorstructure} will then lead to the desired conclusion.

We begin by recalling the main properties of the quantum double model.
In the final subsection this will be combined with the spectral flow to prove our main result of this section.

\subsection{Kitaev's abelian quantum double model}
We recall the family of planar quantum double models for finite groups $G$ as defined by Kitaev~\cite{KitaevQD}.
Although the definition is the same for non-abelian groups, we will restrict to abelian $G$.
We do expect that the stability result is equally valid for non-abelian $G$, but a proof would require an extension of the sector analysis to non-abelian groups.

Let $\mathcal{B}$ be the bond set of the planar square lattice $\ZZ^2$.
To each $e \in \calB$ we assign a $\abs{G}$-dimensional Hilbert space, $\CC^{\abs{G}}$, with an orthonormal basis denoted $\ket{g}$, with $g \in G$.
To define the model we specify the local Hamiltonians and the Heisenberg dynamics on the quasi-local algebra $\calA_\Gamma$.
The interaction terms of the local Hamiltonian are non-trivial only on certain subsets of $\calB$, called stars and plaquettes.
A \emph{star} $v$ is the set of four edges sharing a vertex and similarly, a \emph{plaquette} $f$ is the set of four edges forming a unit square in the lattice.
Interaction terms are defined for each star and plaquette by
\begin{equation}
A_v \equiv \frac{1}{\abs{G}} \sum_{g\in G} A_v^g, \qquad \text{ and } \qquad  B_f \equiv B_f^e,
\end{equation}
where the terms $A_v^g$ and $B_f^h$ are defined by
\begin{figure}[h!]
	\centering
	\includegraphics[width=0.45\textwidth]{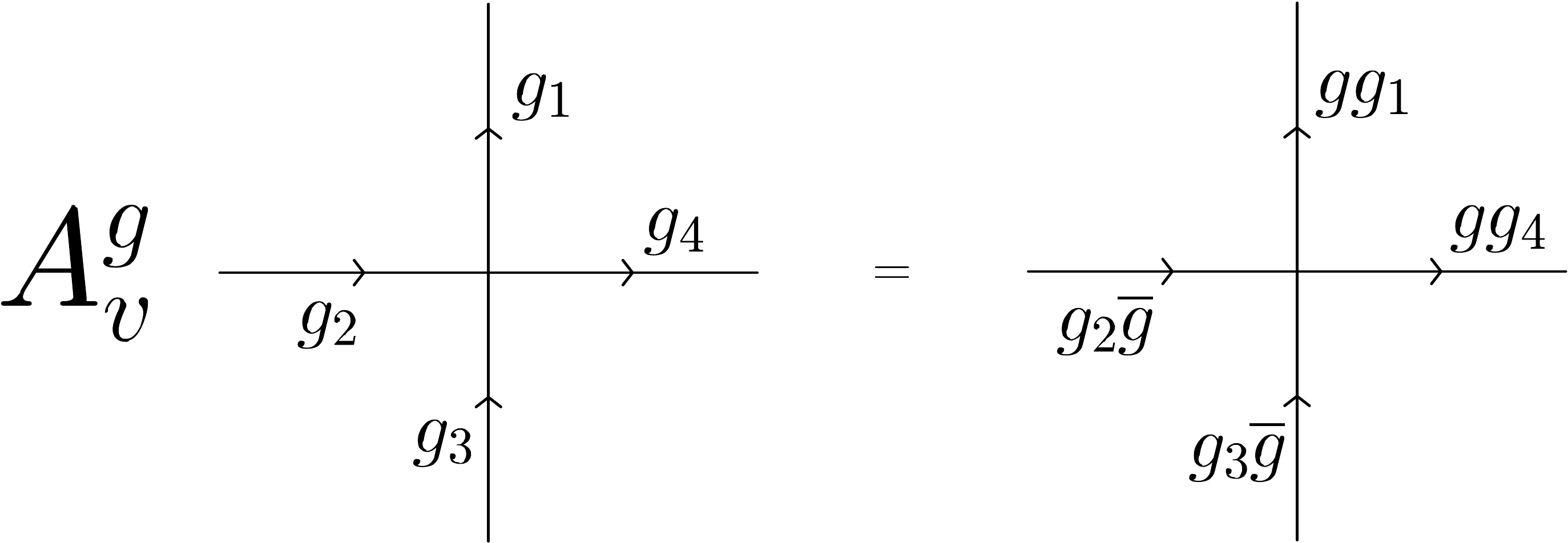}
	\hfill
	\includegraphics[width=0.45\textwidth]{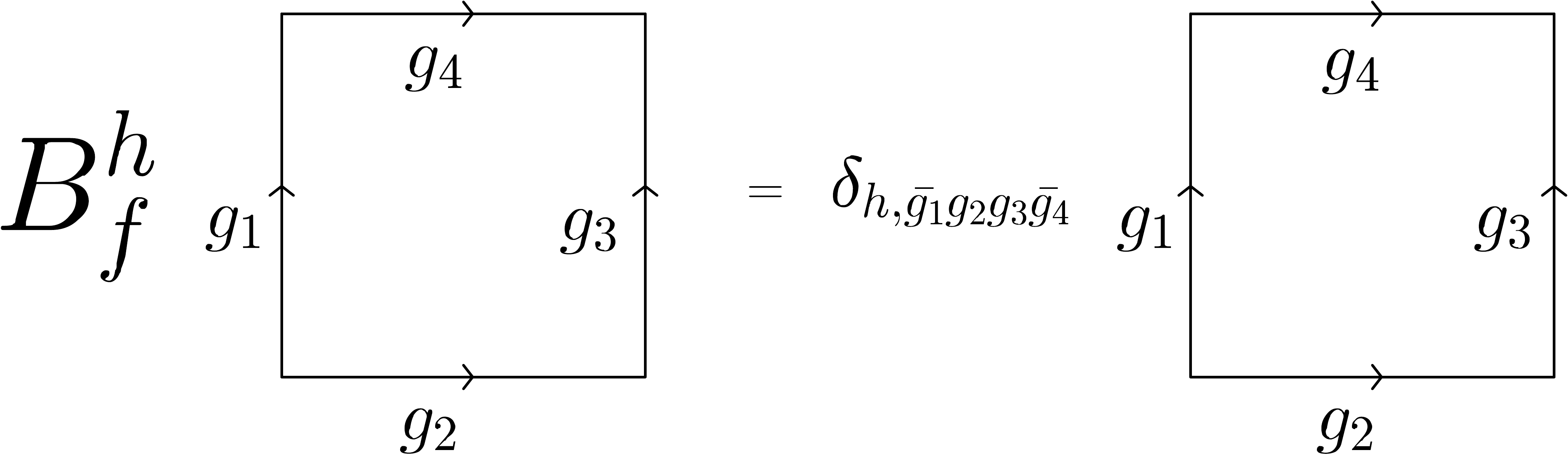}
\end{figure}

\noindent Here $\overline{g}$ is the inverse of $g$.
Note that $A_v$ and $B_f$ commute and are both projectors.
The operator $A_v^g$ can be seen as gauge transformations, with $A_v$ projecting onto the invariant subspace, while $B_e$ projects onto the states with trivial ``flux'' through a plaquette.

For $\Lambda \subset \calB$, denote the subset of stars and plaquettes contained in $\Lambda$ as
$\mathcal{V}_\Lambda = \{ v \subset \Lambda: v \text{ is a star} \}$, and $\mathcal{F}_\Lambda = \{ f \subset \Lambda: f \text{ is a plaquette} \}.$
If $\Lambda \in \mathcal{P}_0(\calB)$, the local Hamiltonians for the quantum double models defined by Kitaev \cite{KitaevQD} are given by
\begin{equation}\label{eqn:qdHam}
\sum_{v\in \mathcal{V}_\Lambda} (I - A_v) + \sum_{f\in \mathcal{F}_\Lambda} (I - B_f) \equiv H_\Lambda \in \calA_\Lambda.
\end{equation}
Since the interaction terms are uniformly bounded and of finite range, 
the existence of global dynamics $t \mapsto \tau_t \in \operatorname{Aut}(\calA)$ is readily established.

For our analysis, following~\cite{ChaNN}, it will be enough to consider square regions $\Lambda_L\subset \calB$ consisting of all edges in $[-L,L]^2$.  
We will denote $H_L = H_{\Lambda_L}$, $\calH_L = \calH_{\Lambda_L}$,
and $H_L^{per}$ the Hamiltonian with periodic boundary conditions
and likewise, $\mathcal{F}_L = \mathcal{F}_{\Lambda_L}$, $\mathcal{V}_{L} = \mathcal{V}_{\Lambda_L}$, and 
$\mathcal{F}_L^{per}$ and $\mathcal{V}_L^{per}$ the set of stars and plaquettes in the case of periodic boundary conditions. 
The generator of the dynamics is the closure of the operator
\begin{equation}
\delta(A) = \lim_{L \ra \infty} [H_L, A]  = \lim_{L \ra \infty } [H_L^{per}, A]\quad \mbox{ for all } \quad A \in \calA_{loc}.
\end{equation}
It follows that the dynamics are given by $\tau_t(A) = e^{i t \delta}(A)$ for all $A \in \calA$. 

The ground state space $\mathcal{G}_L$ of the finite volume Hamiltonians $H_L$ 
simultaneously minimizes the energy of each interaction term, that is, $\mathcal{G}_L = \ker H_L$.
Equivalently, $\Omega_L \in \calH_L$  is a ground state if and only if $( I -A_v) \Omega_L =(I- B_f) \Omega_L =0$ for all $v \in \mathcal{V}_L$ and
$f \in \mathcal{F}_L$.
Consider a family of states $\{\omega_L\}_{L=2}^\infty$ as $L\ra \infty$, 
where $ \omega_L$ is an arbitrary extension
of the state $\langle \Omega_L , \cdot \ \Omega_L  \rangle$ for $\Omega_L \in \mathcal{G}_L$ to the quasi-local algebra $\calA$.
In particular, we could choose a product state.
It can be shown that the sequence $\omega_L(A)$ is eventually constant for any local observable $A$~\cite{FiedlerN}.
This is essentially due to the local topological order condition, which amounts to local observables being unable to distinguish ground states.
Thus, the limit $\omega_0 \equiv \lim_{L \ra \infty } \omega_L$ exists, and by the local topological order condition turns out to unique.
For any $v$ and $f$, choose $L$ large enough such that $v \in \mathcal{V}_{L}$ and $f \in \mathcal{F}_{L}$.
Since $\omega_L$ is a ground state for the finite model it follows that $\omega_0(I - A_v) = \omega_L(I - A_v) = 0$ and $\omega_0(I - B_f) = \omega_L(I - B_f) = 0$.
Thus, we call $\omega_0$ the frustration-free ground state. 
Note that it is translation invariant.

This ground state, which will play the role of the reference state, is pure, gapped and has a simple ground state eigenvector in the GNS representation.
Because of the spectral gap we also have exponential decay of correlations~\cite{NachSLR}.
Finally, $\omega_0$ satisfies the strong approximate split property, as in Definition~\ref{def:strongsplit}, for cones~\cite{FiedlerN}.
To summarize, we have the following properties:
\begin{prop}\label{prop:qdffgs}\label{prop:qdsplit}\cite{AlickiFH,FiedlerN}
	Let $\omega_0$ be the frustration-free ground state of the quantum double model for an abelian group $G$, obtained as above.
	Then the following is true:
	\begin{enumerate}
		\item $\omega_0$ is a pure state;
		\item if $\omega$ is a frustration-free ground state then $\omega = \omega_0$;
		\item let $(\pi_0, \Omega_0, \calH_0)$ be a GNS-representation for $\omega_0$ and $H_0$ be the GNS Hamiltonian.
		Then, $spec(H_0) = 2 \ZZ^{\geq 0}$ with a simple ground state eigenvector $\Omega_0$;
		\item $\omega_0$ satisfies the strong approximate split property for cones.
	\end{enumerate}
\end{prop}
Hence $\omega_0$ fulfills the assumptions as stated in Theorems~\ref{thm:biasymptopia} and~\ref{thm:stabsectorstructure}.

\subsection{Stability of anyons for Kitaev's abelian quantum double models}
Recall that excitations of the quantum double model are formed by violating one or more of the frustration-free conditions.
Excitations live at \emph{sites}.
A site $x = (v,f)$ is a pair of a vertex $v$ and a face $f$ of $\Lambda$ such that $v$ is at the tip of one of the edges surrounding $f$.%
\footnote{In the abelian case one can treat the vertex and face case separately, but it is more convenient to consider them together.}
Write $\mathcal{F}_L$ for all the faces of $\mathcal{B}$ that are contained in the square $[-L,L]^2$ with $L > 0$.
Let $S_L$ denote the set of all sites $x=(v,f)$ such that $v \in \ZZ^2 \cap [-L,L]^2$ and the corresponding face $f \in \mathcal{F}_L$.
We say that a site $x=(v,f)$ is on the boundary of $\Lambda_L$ if $v \in \ZZ^2 \cap [-L,L]^2$ and 
the corresponding face $f \in \mathcal{F}_{L+1} \setminus \mathcal{F}_{L}$.

Excitations live at the end of ribbons.
A ribbon $\rho$ is an ordered sequence of adjacent sites (see~\cite{BombinMD} or Section 3 of~\cite{ChaNN} for a more precise definition and more details).
Ribbons carry an orientation, and we write $\partial_0 \rho$ and $\partial_1\rho$ for the starting (respectively ending) site of $\rho$.
Write $\widehat{G}$ for the group of characters or $G$.
For each $ (\chi,c) \in \widehat{G}\times G$ and ribbon $\rho$ there exists a ribbon operator $F_\rho^{\chi,c}$
such that if $\rho$ is an open ribbon, i.e. $\partial_0 \rho \neq \partial_1 \rho$, then 
\begin{equation}\label{eqn:ribenergy}
H_L F_\rho^{\chi,c} \Omega = C_\rho ( 2- \delta_{\chi, \iota} - \delta_{c, e}  ) F_{\rho}^{\chi,c}  \Omega.
\end{equation}
	where 
\[ C_\rho = 
\left\{ \begin{array}{ll}
2 & \mbox{ if } \partial_i \rho \in S_L \text{ for }   i = 0,1 \\
1 & \mbox{ if } \partial_i \rho \in S_L, \  \partial_{i+1}\rho \notin S_L \\
0 & \mbox{ if } \partial_i \rho \notin S_L \text{ for }   i=0,1.
\end{array}\right. \]
Hence we can create energy eigenstates using the ribbon operators.
A precise inspection shows that the energy increase is precisely due to the frustration-freeness condition being violated at the sites $\partial_i \rho$.

In the infinite volume, we consider a half-infinite ribbon $\rho$ such that $\partial_0 \rho = x \in S_L$ for some $L$ and write $\partial_1 \rho =\infty$ to indicate that the ribbon goes off to infinity.
We also denote $\rho_L = \rho \cap \Lambda_L$ for the part that is contained in a finite box of width $2L$.
It was shown in \cite{Naaijkens11} that the following limit exists (in norm) for all $A \in \calA_\Gamma$:
\begin{equation}\label{eqn:singleexcitation}
\rho^{\chi,c} (A) \equiv \lim_{L\ra \infty}  F_{\rho_L}^{\chi,c *} A   F_{\rho_L}^{\chi,c }.
\end{equation}
It defines an automorphism of $\calA_\Gamma$ which can be interpreted as describing a single excitation of type $(\chi,c)$.
This is vindicated by the following properties, where we use the notation $\omega^{\chi,c}_x \equiv \omega_0 \circ \rho^{\chi,c}$.
\begin{prop}[\cite{FiedlerN,Naaijkens11}]\label{prop:singleexc}
	Let $ (\pi_x^{\chi,c}, \Omega_x^{\chi,c}, \calH_x^{\chi,c})$ be the GNS triple for $\omega_x^{\chi,c}$.
	Then:
	\begin{enumerate}
		\item $\pi_x^{\chi,c}$ are irreducible representations satisfying the selection criterion
		\begin{equation}\label{eqn:conecrit}
		\pi_0 \upharpoonright \calA_{\Lambda^c} \cong \pi \upharpoonright \calA_{\Lambda^c} \quad \mbox{ for all } \quad \Lambda \in\mathcal{C},
		\end{equation}
		\item for all sites $x,x'$ we have $\pi_{x}^{\chi,c} \cong \pi_{x'}^{\chi,c} $,
		\item if $(\chi,c) \neq (\chi',c')$ then $\pi_x^{\chi,c} $ and $\pi_x^{\chi',c'}$  are inequivalent representations,
		\item if $\pi$ is irreducible and satisfies \eqref{eqn:conecrit} then there exists a pair $(\chi, c)$ such that $\pi \cong \pi_x^{\chi,c}$.
	\end{enumerate}
\end{prop}
Thus, we say that the state $\omega_x^{\chi,c}$ is a single excitation state of type $(\chi,c)$.
Since these states satisfy the superselection criterion, and the abelian quantum double models satisfy Haag duality for cones, it is possible to analyze the sector theory, and one finds that it is given by the category $\operatorname{Rep}(\mathcal{D}(G))$.

The goal is to use the spectral flow dynamics to relate charged states of the perturbed model to the unperturbed one.
Hence as a first step we have to identify a suitable set $\mathcal{S}^{qd}$ of weak$^*$ limits of states.
To this end, note that elementary excitation have energy bounded by four as shown in~\eqref{eqn:ribenergy}. 
We consider the set of all states in the infinite systems that have the same energy threshold as follows.
Let $\mathcal{S}_L^{qd}$ be the set of mixtures of eigenstates of $H_L^{per}$ with energy in $[0, 4]$
and $\mathcal{S}^{qd}$ be the set of all weak$^*$ limit points of the sets $S^{qd}_L$.

The following classification theorem follows directly from Lemma 4.2 and Theorem 4.7 of~\cite{ChaNN}.
\begin{thm}\label{thm:aloctoloc}
The state $\omega_x^{\chi,c} \in \mathcal{S}^{qd}$ for all $(\chi,c) \in \widehat{G}\times G$ and sites $x$.
Furthermore, if $\omega \in \mathcal{S}^{qd}$ is a pure state then $\omega \cong \omega_x^{\chi,c}$ for some $(\chi,c) \in \widehat{G} \times G$.
Here $\widehat{G}$ is the group of characters of $G$.
\end{thm}

The energy threshold for $\mathcal{S}^{qd}$ was chosen such that each single excitation state is represented as a state in the  $\mathcal{S}^{qd}$.
In principle, this energy threshold could be increased without changing the sector theory:
the additional states one obtains can be obtained by local operations from states in $\mathcal{S}^{qd}$, and hence will be equivalent to them.
There is also a natural physical interpretation of the criterion.
Note that the most interesting states are the single excitation states, which are weak$^*$ limits of states with uniformly bounded energy.
These states could be interpreted as having a \emph{pair} of excitations, with one of them being moved off to infinity.
If limits of such states are in $\mathcal{S}^{qd}$, that means that the energy cannot exceed a fixed upper bound as we move one of them away.
In other words, the energy criterion means that we restrict to charges which are not confined.

To treat the perturbed and unperturbed model on equal footing, we have to modify the selection criterion~\eqref{eq:superselect} slightly.
Based on the results and discussion so far, we make the following definition.
\begin{defn}\label{defn:qdoubcat}
Define the $C^*$-category of endomorphisms $\Delta^{qd}$ with reference ground state $\omega_0$ as follows:
its objects are $*$-endomorphism $\rho \in \Delta^{qd}$  satisfying
\begin{enumerate}
	\item $\omega_0 \circ \rho \cong \omega$ for some $ \omega \in \mathcal{S}^{qd}$,
	\item $\rho$ is almost localized and transportable (see Definition~\ref{def:superselection}).
\end{enumerate}
We further assume that $\Delta^{qd}$ satisfies the assumptions of Theorem~\ref{thm:Asymptopia}.
The arrows are the intertwiner spaces $(\rho,\sigma)_{\pi_0}$ for each $ \rho,\sigma \in \Delta^{qd}$.
\end{defn}
It should be stressed that the first criterion is about equivalence of representations.
Applying local operators to a state does not change the equivalence class of the corresponding GNS representation.
Hence we can create as many anyon pairs from the ground state as we want without violating the ``energy criterion''.
Ultimately this is because we are only interested in the \emph{total} charge of the representation, which is not changed by local operations, since they always create pairs of conjugate charges.

Note that the irreducible representations satisfying the selection criterion~\eqref{eq:superselect} satisfy both criteria of Definition~\ref{defn:qdoubcat}, by Proposition~\ref{prop:singleexc} and Theorem~\ref{thm:aloctoloc}.
Hence the question is if relaxing strict localization to almost localization (together with the energy condition) will give rise to new sectors.
This turns out not to be the case, and we recover the original category, $\operatorname{Rep}(\mathcal{D}(G))$.
We refer to Appendix~\ref{app:braided} for more on equivalences of braided tensor categories.

The idea behind the proof is to construct a braided tensor functor from the category of the unperturbed model with strict localization to the category of almost localized sectors.
This category is known to be equivalent to $\operatorname{Rep}(\mathcal{D}(G))$.
For the two categories to be the same, this functor should be an equivalence of categories, and moreover be braided.

\begin{thm}\label{thm:stabqd}
The category $\Delta^{qd}$ is a braided tensor $C^*$-category
and is braided tensor equivalent to the category of finite dimensional representations of the quantum double of $G$, $\operatorname{Rep}(\mathcal{D}(G))$.
\end{thm}
\begin{proof}
From Proposition~\ref{prop:qdffgs}, the frustration-free ground state $\omega_0$ is a pure state and has the strong approximate split property.
By construction, we may apply Theorem~\ref{thm:biasymptopia} directly to $\Delta^{qd}$ to show that it can be given the structure of a braided tensor $C^*$-category.

We now show that $\Delta^{qd}$ is equivalent to the category constructed in Refs.~\cite{FiedlerN,Naaijkens11} (and hence equivalent to $\operatorname{Rep}(\mathcal{D}(G))$).
Write $\Delta$ for this category.
From Corollary~\ref{cor:monoidalinclusion} it follows that there is a full and faithful monoidal functor $F: \Delta \to \Delta^{qd}$.
It is also essentially surjective.
To prove this, let $\Lambda \in \mathcal{C}$ be a fixed convex cone.
Consider the superselection criterion selecting $*$-representations satisfying equation~\eqref{eqn:conecrit}.
Let $ \rho \in \Delta^{qd}$ be such that $\pi_0 \circ \rho$ is irreducible. 
Then, by Theorem~\ref{thm:aloctoloc}, we have that $\pi_0 \circ \rho \cong \pi_x^{\chi,c}$ for some site $x$ and $(\chi,c) \in \widehat{G}\times G$.
By Proposition~\ref{prop:singleexc}, $\pi_x^{\chi,c}$ satisfies the criterion~\eqref{eqn:conecrit} from which it follows that 
$\pi_0 \upharpoonright \calA_{\Lambda^c} \cong \pi_0 \circ \rho \upharpoonright \calA_{ \Lambda^c}$.
Recall that for each site $x$ of $\Lambda$ and $(\chi,c) \in \widehat{G} \times G$, the representation $\pi_x^{\chi,c}$ is given by $\pi_0 \circ \rho^{\chi,c}_x$
where the automorphisms $ \rho^{\chi,c}_x$ can be constructed from an infinite ribbon supported entirely in $\Lambda$.
It follows that $ \rho^{\chi,c}_x \in \Delta^{qd}$.
Therefore, the irreducible objects of $\Delta^{qd}$ and the irreducible representations satisfying~\eqref{eqn:conecrit} are exactly equal (up to equivalence).
It follows that $F$ is essentially surjective, and hence an equivalence of categories.

It remains to be shown that $F$ is braided. 
Let $\Lambda^{\mathcal{U}}_\alpha$ and $\Lambda^{\mathcal{V}}_\beta$ be as in Theorem~\ref{thm:biasymptopia}.
Suppose that $\rho$ and $\sigma$ are localized in some cone $\Lambda$ that is to right of $\Lambda_\alpha^{\mathcal{U}}$ and to the left of $\Lambda_\beta^{\mathcal{V}}$.
The general case can be obtained by transporting the sectors and using naturality of the braiding.
Write $\epsilon^\Delta_{\rho,\sigma}$ for the braiding of $\Delta$.
Let $U_m$, $V_n$, and the transported charges $\rho_m, \sigma_n$  be as in Theorem~\ref{thm:biasymptopia}.
From localization and Haag duality, it follows that $\sigma_n(U_m^*) = U_m^*$, and hence
\begin{equation}
	\label{eq:braidstrict}
	(V_n \otimes U_m)^*(U_m \otimes V_n) = V_n^* \rho(V_n) = \epsilon^\Delta_{\rho,\sigma}.
\end{equation}
Note that by Lemma 4.2 of~\cite{Naaijkens11} the braiding $\epsilon^\Delta_{\rho,\sigma}$ does not depend on the specific choice of charge transporters, but only on the relative position of $\rho_n$ and $\sigma$.
Hence equation~\eqref{eq:braidstrict} does not depend on $n$.
Taking the image under $F$ of equation~\eqref{eq:braidstrict}, it follows with equation~\eqref{eqn:braid} that $F(\epsilon^\Delta_{\rho,\sigma}) = \epsilon_{F(\rho),F(\sigma)}$.
The result then follows from the discussion in Appendix~\ref{app:braided} (see in particular Lemma~\ref{lem:eqbraidcat}, together with Theorem 6.3 of~\cite{Naaijkens11} and its generalization to finite abelian groups in \cite{FiedlerN}), and indeed we have a braided tensor equivalence of categories $\Delta^{qd} \ra \operatorname{Rep}(\mathcal{D}(G))$.
\end{proof}

The theorem shows that we might just as well use almost localized endomorphisms, even for the unperturbed model.
As argued earlier, almost localization is the appropriate notion for the perturbed model.
Hence, we now consider perturbations of the quantum double models with periodic boundary conditions of the form discussed in Section~\ref{ass:diff}.
Let $I_0 = [0,4]$ and $ I_1=[5,\infty)$.
Then, by stability of the spectral gap~\cite{BravyiHM}, for some $\epsilon>0$  and  all $ 0\leq s <\epsilon$ 
there are intervals $I_0(s)$ and $ I_1(s)$ with endpoint depending continuously on $s$ 
such that $I_0(0) = I_0$ and $ I_1(0) = I_1$, 
there is a $\gamma>0$ such that $ d( I_0(s), I_1(s) ) >\gamma$ and 
the spectrum of $H_L(s)$ splits into two disjoint sets: $\operatorname{spec}(H_L(s)) = \Sigma_L^0(s) \cup \Sigma_L^1(s)$ with 
$ \Sigma_L^0(s) \subset I_0(s)$ and $\Sigma_L^1(s) \subset I_L^1(s)$ for all $L$.
By equation~\eqref{eq:spectrallr}, for each $s$ the dynamics defined by $H_L(s)$ will satisfy a Lieb-Robinson bound.

Consider the set of elementary excitations for the perturbed quantum double model.
Naturally, these correspond to states with energy supported in the interval $I_0(s)$.
Let $\mathcal{S}_L^{qd}(s)$ be the set of mixtures of eigenstates with energy in $I_0(s)$
and $\mathcal{S}^{qd}(s)$ be the set of all weak$^*$ limit points of the sets $\mathcal{S}^{qd}_L(s)$.
Recall that by Theorem~\ref{thm:autoeq} the spectral flow dynamics $\alpha_s$ satisfies $\mathcal{S}^{qd}(s) = \mathcal{S}^{qd} \circ \alpha_s$.

Define the $C^*$-category $\Delta^{qd}(s)$ of $*$-endomorphisms with respect to the reference ground state $\omega_0 \circ \alpha_s$ as before.
That is, $\Delta^{qd}(s)$ has as objects all $\rho$ satisfying:
\begin{enumerate}
\item $\omega_0 \circ \alpha_s\circ   \rho  \cong \omega$ for some $ \omega \in \mathcal{S}(s)$, 
\item $\rho$ is almost localized and transportable with respect to $\omega_0 \circ \alpha_s$ and further satisfies the assumptions of Theorem~\ref{thm:biasymptopia}.
\end{enumerate}
The intertwiners are defined as before.

With this definition, we can now prove stability of the abelian quantum double models.
\begin{thm}\label{thm:stabqds}
The category $\Delta^{qd}(s)$ is a braided tensor $C^*$-category
and is braided tensor equivalent to the category of finite dimensional representations of the quantum double of $G$, $\operatorname{Rep}(\mathcal{D}(G))$.
\end{thm}
\begin{proof}
By construction, $\Delta^{qd}(0) = \Delta^{qd}$ and thus is braided equivalent to $\operatorname{Rep}(\mathcal{D}(G))$ by Theorem~\ref{thm:stabqd}.

Now consider the $C^*$-category $\alpha_{s}^{-1} \circ \Delta^{qd} \circ  \alpha_s$ with reference state $\omega_0 \circ \alpha_s$.
Theorems~\ref{thm:stabsectorstructure} and~\ref{thm:autoeq} give $\alpha_{s}^{-1} \circ \Delta^{qd} \circ  \alpha_s \subset \Delta^{qd}(s)$.
We claim that $\Delta^{qd}(s) = \alpha_{s}^{-1} \circ \Delta^{qd}\circ  \alpha_s$.
For each irreducible $ \rho \in \Delta^{qd}(s)$ there is a pure state $\omega(s) \in \mathcal{S}(s)$ 
such that $ \omega_0 \circ \alpha_s \circ \rho \cong \omega(s)$.
By Theorem~\ref{thm:autoeq}, $\omega(s) = \omega \circ \alpha_s$ for $ \omega \in \mathcal{S}$.
The state $\omega$ must be a pure state since $\alpha_s$ is an automorphism.
By purity and Theorem~\ref{prop:singleexc}, $\omega \cong \omega^{\chi,c}_x$ for some single excitation state.
It follows that 
\begin{align}
\omega_0 \circ \alpha_s \circ \rho & \cong \omega(s) \\
& = \omega \circ \alpha_s\\
& \cong \omega^{\chi,c}_x \circ \alpha_s \\
& = \omega_0 \circ \alpha_s \circ \alpha_s^{-1} \circ \rho^{\chi,c}_x \circ \alpha_s,
\end{align}
leading to the equivalence $\rho \cong \alpha_s^{-1} \circ \rho^{\chi,c}_x \circ \alpha_s$
from which the claim follows.

By Theorem~\ref{thm:stabsectorstructure}, $\alpha_{s}^{-1} \circ \Delta^{qd}\circ  \alpha_s $ is braided equivalent to $\Delta^{qd}(0)$, which concludes the proof.
\end{proof}

\section{Concluding remarks}
We have presented a general framework to describe charges in infinite quantum spin systems that can be approximately localized in cones.
This is necessary if one wants to perturb, for example, frustration-free models of topologically ordered systems, such as the toric code.
Hence it is relevant for the classification of topological phases.
Using this framework we have shown that the full superselection structure of abelian quantum double models is invariant under perturbations of the dynamics that do not close the gap.
We conclude with a brief discussion of some open problems remaining that we leave to future work.

Theorem~\ref{thm:stabsectorstructure}, establishes stability under quasi-local deformations by, in a sense,
pushing forward the superselection structure from the base system to the deformed system.
An interesting question is whether one can compute a superselection structure for the deformed system independent of the base system.
If, for instance, we could establish the strong approximate split property independently at each deformed level then 
we could simply apply the construction as in Theorem~\ref{thm:asymptopia} and compare the results at each level.
The problem is closely related to the stability of the strong approximate split property in gapped phases.
More generally, in the unperturbed case von Neumann algebraic aspects play an important role, as discussed in the introduction.
It would be very useful to get better control of what happens with the von Neumann algebras after perturbing the dynamics, and hence changing the reference representation.

The almost localized and transportable superselection criterion can be applied quite generally to study the superselection structure of charges in any dimension.
However, the case of two dimensional systems is most interesting, because here the braiding can be a proper braiding, and not a symmetry.
That is, the charges can be genuine anyons.
It is an interesting open problem to find criteria in which this indeed happens.
More generally, there are many intriguing candidates for such models, for example the Levin-Wen string-net models~\cite{LevinW}.
But it is far from clear what the necessary conditions are to obtain a non-trivial superselection theory.
For example, given an arbitrary ground state, say, of a gapped local Hamiltonian, it is not at all clear if there even are any non-trivial almost localized and transportable endomorphisms.

Two dimensional topologically ordered systems are an interesting class for which we expect an interesting superselection theory.
It is unclear under what circumstances we can expect the resulting theory to be described by a modular tensor category.
On the other hand, an anyon theory is often defined by giving the data of a modular tensor category, see for example~\cite{Wang}.
In particular this would mean that there are only finitely many sectors, and it is expected that in topologically ordered models this indeed is the case~\cite{FlammiaHKK}. 
The analysis we presented here gives a way to obtain the braiding, once our assumptions hold true.
For the resulting category to be modular, this braiding should then be completely non-degenerate.
It is an interesting open problem to find sufficient modularity conditions for $\Delta$ determined entirely by the Hamiltonian.
Related results in this direction have been found in the case of localized and transportable endomorphisms~\cite{NaaijkensKL}, but these rely on the specific structure of strictly localized models, and appear to be difficult to verify in practice.

\appendix
\section{Lieb-Robinson bounds for cones}\label{sec:LRcones}
In this appendix we prove the Lieb-Robinson bound for observables localized in cones, Theorem~\ref{thm:LRcone}. This result is a variation of a result obtained by Schmitz~\cite{Schmitz}. Since this reference is not easily accessible, and this is an important technical tool for our results, we provide a proof here for the convenience of the reader. Throughout this appendix we assume that we are given a function $g: \mathbb{R}^{\geq 0} \to \mathbb{R}^{\geq 0}$ satisfying Assumption~\ref{asp1}.

\begin{lemma}\cite{Schmitz}\label{lem:doubImpInt}
	Suppose $g$ satisfies Assumption \ref{asp1} and let $F$ be an $\mathcal{F}$-function.
Define the following sets as $X = \Lambda_\alpha \in \mathcal{C}$ and $ Y_{\epsilon,n} = \left( \Lambda_{\alpha+ \epsilon} - n\right)^c$.
Then, there exists an affine function $\tilde{l}$ such that 
for all $0\leq \alpha < \pi $, $ 0< \epsilon< \pi - \alpha$ and $b>0$ 
\begin{equation}\label{eqn:LRdoubsum}
\sum_{x\in X}\sum_{ y\in Y_{\epsilon,n}} F_{bg} (d(x,y)) \leq C_\epsilon d(X,Y_{\epsilon,n})^{\tilde{l}(\nu)} e^{ - b g(d(X,Y_{\epsilon,n})\sin\epsilon )}
\end{equation}
where 
\begin{equation*}
n \sin(\alpha + \epsilon) \leq d(X,Y_{\epsilon,n}) \leq n \sin(\alpha + \epsilon)+2
\end{equation*}
and  $C_\epsilon$ is non-increasing in $\epsilon$ and only depends on $\nu$, $b$ and $\alpha$.
\end{lemma}

\begin{proof}
Without loss of generality, suppose $\Lambda_{ \alpha  }$ is a cone based at the origin.
Suppose that  $0\leq \alpha < \pi/2$, that is, $\Lambda_{ \alpha  }$ is a convex cone.
Let $\Lambda^*_\alpha = - \Lambda_{\pi/2	 - \alpha}$ be the polar cone for $\Lambda_\alpha$.
Let $Y^I_{\epsilon,n} = Y_{\epsilon,n} \cap \Lambda^*_\alpha$ and $Y^{II}_{\epsilon,n} = Y_{\epsilon,n} \setminus Y_{\epsilon,n}^I$.
It follows that
\begin{equation*}
\sum_{x\in X}\sum_{ y\in Y_{\epsilon, n}} F_{bg}(d(x,y)) =   \sum_{y \in Y_{\epsilon,n}^I} \sum_{x\in X}F_{bg} (d(x,y)) + \sum_{y \in Y_{\epsilon,n}^{II}}  \sum_{x\in X}F_{bg} (d(x,y)).
\end{equation*}
We proceed by bounding the first sum.
By geometry, if $y \in C^*_\alpha$ and $ x \in C_\alpha$ then $ d(y,0) \leq d(x,y)$.
It follows that
\begin{align*}
\sum_{y \in Y^I_{\epsilon,n}} \sum_{x\in X}  F_{bg} (d(x,y)) 
& \leq  \sum_{ y \in Y^I_{\epsilon,n}} \sum_{x\in B_{y}^c(\abs{y} ) }  F_{bg}( d(x,y)))\\
& = \sum_{ y \in Y^I_{\epsilon,n}} \sum_{x\in B_{0}^c(\abs{y}) ) }  F_{bg}( \abs{x}) .
\end{align*}
Bounding the sums by integration and using the explicit form $F_{bg}(r) = e^{-bg(r)} F(r)$ along with Assumption~\ref{asp1} gives
\begin{align*}
\sum_{ y \in Y^I_{\epsilon,n}} \sum_{x\in B_{0}^c(\abs{y}) ) }  F_{bg}( \abs{x}) 
& \leq V  \sum_{ y \in Y^I_{\epsilon,n}} F(\abs{y}) \int_{\abs{y}}^\infty  d r \  r^{\nu-1} e^{- b g(r)} \\
& \leq  V K_{\nu - 1} \sum_{ y \in B_0^c( d(X,Y_{\epsilon,n}))} F(\abs{y}) \abs{y}^{l(\nu -1)} e^{ - b g(\abs{y}) }  \\
& \leq V^2 K_{\nu-1} F(d(X,Y_{\epsilon,n})) \int_{d(X,Y_{\epsilon,n})}^\infty dr  r^{l(\nu - 1) +\nu - 1} e^{-b g(r)}  \\
&\leq C_1  F(d(X,Y_{\epsilon,n}))   d(X,Y_{\epsilon,n})^{\tilde{l}(\nu)} e^{-b g(d(X,Y_{\epsilon,n}))},
\end{align*}
where the constant $V$ is proportional to the volume of the unit sphere in $\RR^\nu$,
$C_1 = V^2 K_{\nu-1}K_{l(\nu-1)+nu-1}$, and $\tilde{l}(\nu) = l(l(\nu - 1) + \nu-1)$.
One can compute that $ n \sin( \alpha + \epsilon) \leq  d(X,Y_{\epsilon,n}) \leq  n \sin( \alpha + \epsilon)+2 $.  
Thus,
\begin{equation}\label{eqn:boundI}
\sum_{y \in Y^I_{\epsilon,n}} \sum_{x\in X}  F_{bg} (d(x,y))  \leq C_1  F(d(X,Y_{\epsilon,n}))   d(X,Y_{\epsilon,n})^{\tilde{l}(\nu)} e^{-b g(d(X,Y_{\epsilon,n}))}.
\end{equation}

Next, we bound the second term in the sum.
For each $y \in Y_{\epsilon,n}^{II}$, the inclusion $X \subset B_y^c(d(y,X))$ holds.
It follows that
\begin{align*}
\sum_{y \in Y_{\epsilon,n}^{II}} \sum_{x \in X} F_{bg}( d(x,y)) & \leq \sum_{y \in Y_{\epsilon,n}^{II}} \sum_{x \in B^c_y(d(y,X) )} F_{bg}( d(x,y))\\
& =  \sum_{y \in Y_{\epsilon,n}^{II}} \sum_{x \in B^c_0(d(y,X) )} F_{bg}( \abs{x}).
\end{align*}

In spherical coordinates, label $y = (r, \phi_1, \phi_2 \ldots, \phi_{\nu -1}) \in \RR^\nu$
where $\phi_1$ labels the angle from the positive $x_1$-axis, $\phi_i \in [0,2 \pi)$ for $ 1 \leq i \leq \nu -2$ and $ \phi_{\nu-1} \in [0,\pi).$
For convenience we will denote $\phi \equiv \phi_1$.
If $y = (r, \phi, \ldots, \phi_{\nu -1}) \in Y_{\epsilon,n}^{II}$ then $\alpha+\epsilon < \phi \leq \pi/2 + \alpha$ 
and for a fixed $\phi$ the values of the radius $r$ range from $ r_{\phi} < r < \infty$, 
where 
\[  r_{\phi} \equiv \frac{ n \sin(\alpha + \epsilon)}{\sin(\phi - \alpha - \epsilon)}.\] 
For $y \in Y_{II}$ a simple calculation gives that $d(y, X) = r \sin(\phi - \alpha) $.
Bounding the sum by integration gives,
\begin{align*}
& \sum_{y \in Y_{\epsilon,n}^{II}} \sum_{x \in B^c_0(d(y,X) )} F_{bg}( d(x,0)) \\
& \qquad \qquad\leq V \sum_{y \in Y_{\epsilon,n}^{II}} F(d(y,X)) \int_{d(y,X)}^\infty d r  \  r^{\nu -1} e^{-b g(r)} \\
& \qquad \qquad\leq  V K_{\nu-1}\sum_{y \in Y_{\epsilon,n}^{II}} F(d(y,X)) r^{l(\nu-1)} e^{- b g(d(y,X))} \\
& \qquad \qquad\leq  V^2 K_{\nu-1} F(d(X,Y_{\epsilon,n})) \int_{\alpha+ \epsilon}^{\pi/2 + \alpha} d\phi \sin \phi
\int_{r_{\phi}}^\infty  dr \ r^{l(\nu -1) +\nu -1} e^{-b g( r \sin(\phi - \alpha))} \\
&\qquad \qquad\leq C_1 F(d(X,Y_{\epsilon,n})) \int_{\alpha+ \epsilon}^{\pi/2 + \alpha} d\phi \sin \phi \sin(\phi - \alpha)^{l(\nu-1) - \nu}
r_\phi^{l(l(\nu -1) +\nu -1)} e^{-b g( r_\phi \sin(\phi - \alpha))}.
\end{align*}
For all $ \phi \in (\alpha+ \epsilon, \pi/2 + \alpha)$, the bound $ d(X,Y_{\epsilon,n})\geq  r_{\phi} \geq (n-1)\sin(\alpha + \epsilon) + \frac{ \sin(\alpha+\epsilon)}{\sin(\phi  - \alpha - \epsilon )}$ holds. 
It follows from the properties of $g$ that
\begin{equation*}
e^{-b g(r_{\phi}\sin(\phi - \alpha))} \leq e^{ -b g( (n-1) \sin(\alpha + \epsilon) \sin(\phi - \alpha))}e^{- b g\left(\frac{ \sin(\alpha+\epsilon) \sin(\phi-\alpha)}{\sin(\phi  - \alpha - \epsilon )}\right) }.
\end{equation*}
Substituting back into the bound for $Y_{\epsilon,n}^{II}$ and simplifying gives
\begin{align}\label{eqn:boundII}
& \sum_{y \in Y_{\epsilon,n}^{II}} \sum_{x \in X} F_{bg}( d(x,y)) \\
& \qquad \qquad \leq 	C_1 F(d(X,Y_{\epsilon,n})) d(X,Y_{\epsilon,n})^{\tilde{l}(\nu)} e^{ -b g(d(X,Y_{\epsilon,n}) \sin\epsilon)}
\int_{0}^{\pi/2 - \epsilon} d\phi  \sin(\phi)^{-\nu} e^{- a \frac{ \sin(\alpha+\epsilon) \sin\epsilon}{\sin\phi} }
\end{align}
The $\phi$-integral is bounded for fixed $ 0< \epsilon < \pi/2$, is non-increasing in $\epsilon$, and only depends on $\nu$, $a$ and  $\alpha$.
Combining the bounds \eqref{eqn:boundI} and \eqref{eqn:boundII} gives the result.

The double sum in \eqref{eqn:LRdoubsum} is  symmetric in the interchange of $x$ and $y$
and the opening angle for $Y_{\epsilon,n}$ is $ \beta = \pi - (\alpha+ \epsilon)$.
Since $ 0\leq  \alpha \leq \pi/2 $ and $ 0<\epsilon< \pi/2 -\alpha $ this implies that $  \pi/2 < \beta < \pi$.
Thus, exchanging the roles of $ X$ and $ Y_{\epsilon, n}$  gives the result.
\end{proof}


\begin{rem}
	The double sum $\sum_{x\in X} \sum_{y\in Y_{\epsilon,n}} F_{bg} (d(x,y))$ is indeed divergent as $\epsilon \ra 0$ for all $n$ and $b>0$. 
	By comparing the sets $Y_{\epsilon, n}$ and $ Y_{0, n+2}$ one achieves the lower bound
	\begin{equation}
	\left\lfloor\frac{3 \tan \alpha \tan (\alpha+ \epsilon)}{\tan(\alpha+\epsilon) - \tan\alpha} \right\rfloor \sum_{j= 0 }^\infty F_{bg}( j+ n+3)
	\leq \sum_{x\in X} \sum_{y\in Y_{\epsilon,n}} F_{bg} (d(x,y)),
	\end{equation}
	where the term $\left\lfloor\frac{3 \tan \alpha \tan (\alpha+ \epsilon)}{\tan(\alpha+\epsilon) - \tan\alpha} \right\rfloor   = O\left(\frac{1}{\epsilon}\right)$ is a lower bound for the number of lattice points in $ Y_{\epsilon,n}\setminus Y_{0,n+2}$.	
	This agrees with the upper bound found in Theorem \ref{thm:LRcone} as the $\phi$-integral in \eqref{eqn:boundII} diverges as $\epsilon \ra 0$.
\end{rem}

\begin{cor}\label{cor:fsatpolydecay}
	Suppose $g$ satisfies Assumption \ref{asp1}.  
	Let $X$ and $ Y_{\epsilon,n}$ be defined as in Lemma \ref{lem:doubImpInt}.
	Then, for all $0\leq \alpha < \pi/2$, $ 0< \epsilon< \pi/2 - \alpha$, $b>0$ and $k \in \NN$ we have that
	\begin{equation}
	\lim_{n \ra \infty } \bigg[ n^{k} \sum_{x\in X} \sum_{y\in Y_{\epsilon,n}} F_{bg} (d(x,y)) \bigg] \ra 0.
	\end{equation}
\end{cor}
\begin{proof}
	
	Comparing the bound in Lemma \ref{lem:doubImpInt} with the limit $\lim_{n \ra \infty} n^{k} e^{-b g(n)} = 0$
	for all $b>0$ and $k \in \NN$ gives the result.
\end{proof}

With these results Theorem~\ref{thm:LRcone} follows by substituting the bound found in Lemma~\ref{lem:doubImpInt} into the Lieb-Robinson bound~\eqref{eqn:qlbound}.

\section{Spectral flow}\label{app:spectral}
To study the perturbed model we apply the spectral flow theory~\cite{BachmannMNS}.
We recall the necessary assumptions, which limit the type of perturbations we can study, and the main result of this theory.
Let $\Gamma = \mathcal{B}$ be the bond set of the $\ZZ^\nu$ lattice.
We will be interested in the case $\nu = 2$, but the results in this section hold for any natural number $\nu$.
Note that there is a natural action of $\ZZ^\nu$ (and correspondingly on the bond set) by translations, inducing a group of automorphisms $x \mapsto T_x \in \operatorname{Aut}(\calA_\Gamma)$.

Consider a family of interactions $\Phi_s : \mathcal{P}_0(\Gamma) \ra \calA$ for $0\leq s \leq 1$.
These will typically be given by, say, a frustration-free Hamiltonian plus an $s$-dependent perturbation.
We will make the following assumptions.
\begin{assumption}[Differentiability]\label{ass:diff}
	The family is differentiable in the parameter $s$ and short-range such that
	for some $a, M>0$
	\begin{equation}
		\sup_{x,y \in \Gamma} e^{ a d(x,y)} \sum_{\substack{ X\subset \mathcal{P}_0(\Gamma):\\ x,y\in X} }\sup_s \| \Phi_s (X) \| + \abs{X} \| \partial_s \Phi_s (X) \| \leq M.
	\end{equation}
\end{assumption}

The following assumption allows us to split the spectrum in a ``low-energy'' part and the rest of the spectrum, separated by a gap.
\begin{assumption}[Spectral gap]\label{ass:gap}
	There exists an increasing and exhaustive sequence $\Lambda_n \in \mathcal{P}_0(\Gamma)$ such that the
	local Hamiltonians  $H_{\Lambda_n}(s) = \sum_{X \subset \Lambda_n} \Phi_s(X)$ have a uniform spectral gap. 
	Let $\Sigma_\Lambda(s) \equiv \operatorname{spec}(H_\Lambda(s))$ and assume that the spectrum decomposes into two
	non-empty disjoint  sets, $\Sigma_\Lambda(s) =\Sigma_\Lambda^{0}(s)  \cup \Sigma_\Lambda^{1}(s) $ and there are disjoint intervals $I_i(s) = [a_i(s), b_i(s)]$ for $i = 0,1$  
	such that $a_i(s), b_i(s)$ are continuous functions, $a_0(s) \leq  b_0(s) < a_1(s) \leq b_1(s)$ and
	$\Sigma_\Lambda^i(s) \subset I_i(s)$.
	The spectrum is gapped in the sense that 
	there exists $\gamma>0$ with $ d( \Sigma_{\Lambda_n}^1(s),\Sigma_{\Lambda_n}^0(s) )\geq \gamma >0$
	where $\gamma$ is uniform  in $s\in [0,1]$ and $n$.
\end{assumption}

\begin{assumption}[Lieb-Robinson bound]\label{ass:lrb}
	Consider the finite volume dynamics for $\Lambda \in \mathcal{P}_0(\Gamma)$, given by $\tau_t^{H_\Lambda(s)} (A) = e^{i t H_\Lambda(s)} Ae^{-i t H_\Lambda(s)}$.  We  assume there exist $a>0$ and $ v_a, K_a >0$ such that
	the following exponential Lieb-Robinson bound holds:
	\[  \left\|\left[\tau_t^{H_\Lambda(s)}(A), B\right]\right\| \leq K_a \|A\| \|B\| e^{v_a t} \sum_{x \in X}\sum_{y\in Y} e^{-a d(x,y)}. \]
\end{assumption}

\begin{assumption}[Translation invariance]\label{ass:transinv}
	We assume the perturbation is translation invariant, that is, $\Phi_s(T_{x}(X)) = T_x( \Phi_s(X))$ for all $X\subset \Gamma$ and $x \in \ZZ^\nu$.
\end{assumption}

The idea is to view $s \mapsto \Phi_s$ as an $s$-dependent ``time''-evolution, which can be smeared against a test function $w_\gamma$.
More precisely, let $w_\gamma \in L^1(\RR)$ be a function satisfying:
\begin{enumerate}
	\item $w_\gamma$ is real-valued and $ \int dt w_\gamma(t) = 1$,
	\item the Fourier transform $\widehat{w}_\gamma$ is supported in the interval $[-\gamma, \gamma]$.
\end{enumerate}
For the existence of such a function see Lemma 2.6 of~\cite{BachmannMNS}.

Consider, for a finite subset $\Lambda \subset \Gamma$, the self-adjoint operator
\begin{equation}
D_\Lambda(s) = \int_{-\infty}^{\infty} dt w_\gamma(t) \int_0^t du e^{i  u H_\Lambda(s)} H'_\Lambda(s) e^{-i u H_\Lambda(s)}.
\end{equation}
This operator will be the generator of the spectral flow, which allows one to relate the spectral subspaces $\Sigma^0_\Lambda(s)$ with each other.
More precisely:
\begin{prop}[Spectral flow~\cite{BachmannMNS}]\label{prop:spectralflow}
	There is a norm-continuous family of unitaries $U_\Lambda(s)$ such that the spectral projections $P_\Lambda(s)$ onto the subset $\Sigma_\Lambda^0(s)$ are given by
	\begin{equation}
	P_\Lambda(s) = U_\Lambda(s) P_\Lambda(0) U_\Lambda(s)^*. 
	\end{equation} 
	The unitary family is given by the unique solution to
	$-i \frac{d}{ds} U_\Lambda(s) = D_\Lambda(s) U_\Lambda(s)$ and $U_\Lambda(0) = I$.
\end{prop}

The infinite volume limit is obtained analogously to the transition from local to infinite volume dynamics.
Define the \emph{spectral flow} automorphism family as
\begin{equation}
\alpha_s^\Lambda(A) = U^*_\Lambda(s) A U_\Lambda(s).
\end{equation}
Since the $U_\Lambda(s)$ are solutions to a time-dependent Schr{\"o}dinger's equation, 
the spectral flow will have a cocycle property.
More precisely, solving the equation will yield a unitary propagator $U_\Lambda(s_1, s_2)$, with $U_\Lambda(s) \equiv U_\Lambda(s,0)$. 
The corresponding automorphisms satisfy
\begin{equation}
\alpha^\Lambda_{s_1,s_2}(A) = \alpha_{s,s_2}^\Lambda\circ \alpha_{s_1,s}^\Lambda (A).
\end{equation}
Again, we use the convention that $\alpha^\Lambda_s \equiv \alpha_{s,0}^\Lambda$.

To specify an interaction norm on the perturbations,
let $g(x)$ denote the function given by equation~\eqref{eqn:subexpg} for $p=2$. That is,
\begin{equation}
g(r) = \left\{ \begin{array}{ll}
\frac{r}{\ln^p(r)} & \mbox{ if }  x> e^p\\
\left( \frac{e}{p}\right)^p & \mbox{ if } x\leq e^p
\end{array}. \right. 
\end{equation}
Furthermore, let $F$ be an $\mathcal{F}$-function such that there exists $ 0< \delta < 2/7$ with
\begin{equation}
\sup_{r\geq 1} \frac{e^{-\delta g(r)}}{F(r)} < \infty.
\end{equation}
Define the $\mathcal{F}$-function 
\begin{equation}
F_\Psi(r) \equiv e^{ - \mu g\left(\frac{\gamma}{8 v_a} r\right)} F\left( \frac{\gamma}{8 v_a} r \right).
\end{equation}
The main property of the spectral flow is that it can be seen as a quasi-local dynamics connecting the spectral subspaces.

\begin{thm}\cite{BachmannMNS}
	Let Assumptions~\ref{ass:diff},~\ref{ass:gap}, and~\ref{ass:lrb} hold.
	Then there exists a time-dependent volume-dependent interaction $\Psi_\Lambda(s)$ 
	such that 
	\begin{equation}\label{eqn:specflowFnorm}
	\| \Psi_\Lambda\|_{F_\Psi} \equiv 
	\sup_{x,y \in \Gamma} \frac{1}{F_\Psi(d(x,y))}\sum_{\substack{Z \subset \Lambda \\x,y \in Z}} \sup_{0\leq s \leq 1} \| \Psi_\Lambda(Z,s)\| <\infty
	\end{equation} 
	and 
	\begin{equation}
	D_\Lambda(s) = \sum_{Z \subset \Lambda } \Psi_\Lambda(Z,s).
	\end{equation}
	Further, if $\Lambda_n \in \mathcal{P}_0(\Gamma)$ is an increasing and exhausting sequence as given in Assumption~\ref{ass:gap} and there exist positive constants $ b_1, b_2$ and $ p$ such that 
	\begin{equation}
	d(\Lambda_m, \Lambda_n^c) \geq b_1 (n-m), \quad \text{and} \quad \abs{\Lambda_n}\leq b_2 n^p,
	\end{equation}
	then \eqref{eqn:specflowFnorm} holds uniformly in $\Lambda_n$.
\end{thm}

We can use similar techniques as for local dynamics to argue the following.
There exists a strongly continuous cocycle of automorphisms $\alpha_s$ on $\calA_\Gamma$ as the strong limit of $ \alpha_s^{\Lambda_n}$ 
for an increasing and exhausting sequence $\Lambda_n$.
In addition, for all $A\in \calA_X$ and $ B\in \calA_Y$ such that $ d(X,Y) >0$ and $ 0\leq s \leq 1$ we have that 
\begin{equation}\label{eq:spectrallr}
\| [\alpha_s(A), B] \| \leq 2 \frac{\|A\| \|B\|}{C_{F_\Psi}} \left(e^{2 \|\Psi\|_{F_\Psi} \abs{s}} -1\right)\sum_{x \in X}\sum_{y\in Y} F_\Psi(d(x,y)) .
\end{equation}
That is, the spectral flow obeys a Lieb-Robinson bound.

Finally, translation invariance is preserved if the local interactions are translation invariant: 
\begin{lemma}\label{lem:transinv}
	If $\Phi_s$ is a translation invariant interaction for all $s$ then $ \alpha_s \circ T_x = T_x \circ \alpha_s$ for all $s$ and $x\in\Gamma$.
\end{lemma}

\begin{proof}
	This follows from
	\begin{align*}
	T_x( D_\Lambda(s)) & = T_x\left(\int_{-\infty}^{\infty} dt w_\gamma(t) \int_0^t du e^{i  u H_\Lambda(s)} H'_\Lambda(s) e^{-i u H_\Lambda(s)} \right)\\ 
	& = \int_{-\infty}^{\infty} dt w_\gamma(t) \int_0^t du e^{i  u H_{\Lambda+x}(s)} H'_{\Lambda+x}(s) e^{-i u H_{\Lambda+x}(s)}\\
	& = D_{\Lambda+x}(s).
	\end{align*}
	Since $D_\Lambda(s)$ is the generator of $\alpha_s$, the result follows.
\end{proof}

The most important property of the $\alpha_s$ is that they can be used to relate the low energy states along the path $\Phi_s$, similar to Proposition~\ref{prop:spectralflow}.
Let $\mathcal{S}_\Lambda(s)$ denote the set of states on $\calA_{\Lambda}$ that are mixtures of eigenstates with energy in the interval $I_0(s)$.
Define $\mathcal{S}(s)$ as the set of all weak$^*$ limit points of states that, when restricted to $\calA_\Lambda$, are elements of the sets $\mathcal{S}_\Lambda(s)$.
Then they are related by an automorphism of $\calA_\Gamma$:
\begin{thm}[\cite{BachmannMNS}]\label{thm:autoeq}
	The set $\mathcal{S}(s)$ is automorphically equivalent to $\mathcal{S}(0)$ for all $s \in [0,1]$.
	In particular, 
	\begin{equation}
	\mathcal{S}(s) = \mathcal{S}(0 ) \circ \alpha_s
	\end{equation}
	where $\alpha_s$ is the spectral flow on $\calA$.
\end{thm}

This result not only allows use to relate the different ground states, but is also relevant for the superselection sectors.
Indeed, a single anyon state can be obtained by creating a pair of excitations in the ground state, and sending one of them off to infinity.
In the unperturbed model, such two-anyon states have energy 2.
Hence, the single anyon state is in $\mathcal{S}(0)$ if we choose the interval $I_0(s)$ to be large enough.
This is a key observation in the stability proof in Section~\ref{sec:stability}.
A closer investigation of such weak$^*$ limits can be found in~\cite{ChaNN}.

\section{Equivalence of braided tensor categories}\label{app:braided}
Recall that a \emph{monoidal} (or \emph{tensor}) category is a category $\mathcal{C}$ together with a bifunctor $\otimes$ and an object $\iota$ that acts as the tensor unit.
Furthermore, one has to specify natural isomorphisms, encoding for example $(\rho \otimes \sigma) \otimes \tau \cong \rho \otimes (\sigma \otimes \tau)$, satisfying some coherence conditions.
All categories that we will be considering are \emph{strict}, which means that these isomorphisms are in fact identities.
To get a \emph{braided} tensor category one in addition has to specify natural isomorphisms $c_{\rho,\sigma} : \rho \otimes \sigma \to \sigma \otimes \rho$ satisfying the braid relations.
The main definitions and results on such categories can be found in~\cite{EtingofGNO,Mueger}.

In addition to the braided tensor structure, the categories that we consider here have a lot more structure.
In particular, they are modular.
We do not need much of this structure, but we do need that the Hom-sets are linear spaces over $\mathbb{C}$, with all relevant operations (e.g.\ composition of morphisms) being linear as well.
Furthermore we need that there is a contravariant anti-linear functor $*$ that is also an involution, and we require that the rest of the structure is compatible with the $*$-operation in the natural way.
In our case the linear structure and $*$-operation are induced by the underlying $C^*$-algebra containing the intertwiners (and in fact this also induces a norm).
Using this it is easy to see, for example, that an isomorphism in the category is the same thing as a unitary in the $C^*$-algebra.

The main result of this paper is that the sector theory of the perturbed theory is the same as that of the unperturbed theory when viewed as braided categories.
We say that two braided monoidal categories $\mathcal{C}$ and $\mathcal{D}$ are \emph{equivalent as braided monoidal categories} if there exists a braided functor $F : \mathcal{C} \to \mathcal{D}$ that is full, faithful and essentially surjective (and hence an equivalence of categories), see for example~\cite[Sect. 8.1]{EtingofGNO}.
One could also consider a more symmetric definition.
In particular, one can assume the existence of $F: \mathcal{C} \to \mathcal{D}$ braided, and $G: \mathcal{D} \to \mathcal{C}$ braided, such that there are monoidal natural isomorphisms $\alpha : \operatorname{id}_\mathcal{C} \Rightarrow G \circ F$ and $\beta: F \circ G \Rightarrow \operatorname{id}_{\mathcal{D}}$.
As is well-known the latter condition implies that $F$ is essentially surjective and full and faithful. 
Hence the second definition implies the first.

The converse is less obvious, and indeed it is not immediately clear that the first definition really defines an equivalence relation.
That is, given $F$ a braided monoidal functor that is also an equivalence of categories, is it possible to find $G$ as above?
The answer is yes.
As a remark, we note that it is essential that $F$ is braided.
Indeed, there are categories that are equivalent as monoidal categories, but \emph{not} as braided monoidal categories~\cite{EtinghofG}.

To find $G$, note first that $F$ full, faithful, and essentially surjective.
Hence $F$ is an equivalence of categories, and there exists $G : \mathcal{D} \to \mathcal{C}$ and natural transformations $\alpha : \operatorname{id}_\mathcal{C} \Rightarrow G \circ F$ and $\beta: F \circ G \Rightarrow \operatorname{id}_{\mathcal{D}}$.
Hence we have to show that $G$ is braided monoidal, and the natural isomorphisms are in fact natural \emph{monoidal} isomorphisms.
The latter follows from Proposition I.4.4.2 of~\cite{SaavedraRivano}, which also tells us that $G$ is a monoidal functor.
Since $G$ is full and faithful, it follows that there is a unique braiding $c'_{X,Y}$ on $\mathcal{D}$ such that $G(c'_{X,Y}) = c^{\,\mathcal{C}}_{G(X),G(Y)}$, where $c^{\,\mathcal{C}}_{X,Y}$ is the given braiding on $\mathcal{C}$.
Hence it remains to be shown that $c'_{X,Y} = c^{\,\mathcal{D}}_{X,Y}$.
But this follows from Proposition I.4.4.3.2 of~\cite{SaavedraRivano}.\footnote{A braiding is called a \emph{commutativity constraint} in~\cite{SaavedraRivano}. The author only considers the case where the braiding is a \emph{symmetry}, but for the result we need this extra condition is not necessary.}

The discussion can be summarized as the following lemma:
\begin{lemma}
	\label{lem:eqbraidcat}
	Let $F : \mathcal{C} \to \mathcal{D}$ be a braided monoidal functor between braided monoidal categories. If $F$ is an equivalence categories, then there is a functor $G: \mathcal{D} \to \mathcal{C}$ braided and monoidal, together with natural monoidal transformations $\alpha : \operatorname{id}_\mathcal{C} \Rightarrow G \circ F$ and $\beta : F \circ G \Rightarrow \operatorname{id}_\mathcal{D}$.
\end{lemma}

\bibliographystyle{abbrv}
\bibliography{alendbib}
\end{document}